\documentclass[envcountsame,envcountsect]{llncs}

\usepackage{microtype}
\usepackage{mathtools,amssymb,mathrsfs}
\usepackage[table,svgnames]{xcolor}
\usepackage[colorlinks=true,linkcolor=black,citecolor=MidnightBlue,urlcolor=MidnightBlue]{hyperref}
\usepackage{xspace}
\usepackage{tikz}
\usepackage{pifont}
\usepackage{adjustbox}
\usepackage{subcaption}
\usepackage{algorithm}
\usepackage{algorithmicx}
\usepackage{graphicx}
\usepackage{algpseudocode}
\usepackage{amsmath}

\newcommand{\negA}{\vspace{-0.05in}}
\newcommand{\negB}{\vspace{-0.1in}}

\newcommand{\posA}{\vspace{0.08in}}
\newcommand{\mysection}[1]{\negB\section{#1}\negA}
\newcommand{\mysubsection}[1]{\negA\subsection{#1}\negA}

\newcommand{\myparagraph}[1]{\par\smallskip\par\noindent{\bf{}#1:~}}

\newcommand{\eps}{\varepsilon}

\newcommand{\comment}[1]{}

\begin{document}
\sloppy
\setcounter{page}{1} 
\title{An APTAS for Bin Packing with Clique-graph Conflicts}
\author{Ilan Doron-Arad, Ariel Kulik \and Hadas Shachnai}
\institute{Computer Science Department, Technion, Haifa 3200003, Israel. \mbox{E-mail: {\tt \{idoron-arad,kulik,hadas\}@cs.technion.ac.il.}}}

\maketitle
	

\begin{abstract}
\label{sec:abstract}
We study the following variant of the classic {\em bin packing} problem.
Given a set of items of various sizes, partitioned into groups, 
find a packing of the items in a minimum number of identical (unit-size) bins, 
such that no two items of the same group are assigned to the same bin. 
This problem, known as {\em bin packing with clique-graph conflicts},
has natural applications in storing file replicas, security in cloud computing and signal distribution.

Our main result is an {\em asymptotic polynomial time approximation scheme (APTAS)} for the problem, improving upon the best known ratio of $2$. 
As a key tool, we apply 
a novel {\em Shift \& Swap} technique which generalizes 
the classic linear shifting technique to scenarios 
allowing conflicts between items.
The major challenge of packing {\em small} items using only a small number of extra bins is tackled through an intricate combination of enumeration and a greedy-based approach that utilizes the rounded solution of a {\em linear program}.

\comment{
We use our scheme to obtain APTASs for 
two non-trivial subclasses of instances of {\em bin packing with interval-graph conflicts}, where the interval graph is either {\em proper}, or consists of {\em constant length intervals}, thus improving the best known ratio of $\frac{7}{3}$ for such instances.
}
\end{abstract}

\negA
\negA
\negA
\negA
\comment{}

\section{Introduction}\label{Introduction}
\label{sec:intro}

In the classic {\em bin packing (BP)} problem, we seek a packing of
items of various sizes into a minimum number of unit-size bins.
This fundamental problem arises in a wide variety of contexts and has been studied extensively since the early 1970's.
In some common scenarios, the input is partitioned into 
{\em disjoint groups}, such that items in the same  
group are {\em conflicting} and therefore cannot be packed together. 
For example, television and radio
stations often assign a set of programs to their channels. 
Each program falls into a genre such as
comedy, documentary or sports on TV, or various musical genres on radio. To maintain a diverse daily schedule of programs, the station would like to avoid broadcasting two programs of the same genre in one channel. Thus, we have a set of items (programs) partitioned into groups (genres) that need to be packed into a set of bins (channels), such that items belonging to the same group cannot be packed together.

We consider this natural variant of the classic bin packing problem that we call {\em group bin packing (GBP)}.
Formally, the input is a set of $N$ items $I = \{ 1, \dots , N \}$ with corresponding sizes $s_1,...,s_N \in (0,1]$, partitioned into $n$ disjoint groups $G_1,...,G_n$, i.e., $I=G_1 \cup G_2 \cup \ldots \cup G_n$. 
The items need to be packed in unit-size bins. 
A packing is {\em feasible} if the total size of items in each bin does not exceed the bin capacity, and no two items from the same group are packed in the same bin. We seek a feasible packing of all items in a minimum 
number of unit-size bins. 
We give in Appendix~\ref{sec:applications} some natural applications of GBP.
 
Group bin packing can be viewed as a special case of {\em bin packing with conflicts (BPC)}, in which the input is a set of items $I$, each having size in $(0,1]$, along with a conflict graph $G=(V,E)$. An item 
$i \in I$ is represented by a vertex 
$i \in V$, and there is an edge $(i,j) \in E$ if items $i$ and $j$ cannot be packed in the same bin. The goal is to pack the items in a minimum number of unit-size bins such that items assigned to each bin form an {\em independent set} in $G$.

Indeed, GBP is the special case where the conflict graph is a union of
cliques. Thus, GBP is also known as {\em bin packing with clique-graph conflicts} (see Section~\ref{sec:related_work}).
\negA

\comment{\subsection{Group Bin Packing via Clustering}
In the the big data era, the traditional computing paradigm involving a single machine which
receives a problem instance and returns its corresponding output
has become less relevant. Indeed, 
input instances are way too big to be tackled by a single machine. Instead,
new computing paradigms that involve many machines, each receiving a small portion of the input
instance, have been introduced, such as  MapReduce~\cite{DG08}, Hadoop~\cite{W12}, Dryad~\cite{IB+07} and Spark~\cite{Z+10}. The common practice  is to partition a massive input instance into smaller instances
and to process each 
separately by its own machine. Then, the individual
solutions are assembled into a solution for the entire instance.
In many cases this assembling stage cannot be pursued (e.g., 
due to locality issues, high costs, or privacy considerations). 
The performance measure considered in such scenarios is the {\em price of clustering}, evaluating the 
degradation in solution quality relative to the optimum due to partitioning the instance into clusters and solving for each cluster separately (see the formal definition below).

Our study of group bin packing encompasses also massive instances, for which a solution is obtained via clustering. We formulate necessary and sufficient conditions for obtaining an efficient global solution for GBP via clustering.

\subparagraph*{Clustering Model:}
Let $I$ be an instance of an optimization problem $\Pi$, and let $I_1, \ldots , I_m$ be a
partition of $I$ into clusters, where $I_i$ is the sub-instance assigned to cluster $i$, $1 \leq i  \leq m$.
We denote by $OPT(I), OPT(I_i)$ the number of bins used in an optimal solution for a GBP instance $I$ and $I_i$, respectively.
We use the notion of {\em price of clustering (PoC)} as defined in~\cite{PoC:A}.
Let $PoC(\Pi)$ denote the price of clustering for a given optimization problem $\Pi$.
Then the price of clustering for GBP is given by
$$PoC(GBP) = sup_{I,I_i}\frac{\sum_{I_i}OPT(I_i)}{OPT(I)}.$$
While computation models based on clustering have been in use in the past two decades, the price of clustering
measure was studied for the first time only very recently~\cite{PoC:A}. 

We note that PoC is defined with respect to the globally optimal solution and the optimal solution for each cluster.
However, in practice, due to the hardness of the underlying optimization problem, one might be able to use only an approximation algorithm, ${\cal A}$. Hence, it is natural to define a performance measure relating to such algorithms. In particular, given an algorithm ${\cal A}$ for GBP, we define
$$PoC_{{\cal A}}(GBP) = sup_{I,I_i}\frac{\sum_{I_i}{\cal A}(I_i)}{{\cal A}(I)},$$
where ${\cal A}(I)$ and ${\cal A}(I_i)$ are the number of bins used by ${\cal A}$ for the instance $I$ and $I_i$, respectively.
As noted in~\cite{PoC:A}, the price of clustering can be substantially affected by the partition used for a given problem instance. Thus, it is natural to seek a {\em good} partition. 

We focus in this paper on attributes of partitions that guarantee that $PoC(GBP)$ is bounded by some constant $d>1$.}

\subsection{Contribution and Techniques}
\label{sec:Contribution and Techniques}
Our main result (in Section~\ref{sec:APTAS}) is an APTAS for the group bin packing problem,
improving upon the best known ratio of $2$~\cite{allornothing}.\footnote{We note that $2$ is the best known absolute as well as {\em asymptotic} approximation ratio for the problem
(see Section~\ref{sec:related_work}). We give formal definitions of 
absolute/asymptotic ratios in Section~\ref{sec:preliminaries}.}
\comment{
In Section~\ref{sec:interval_graphs}, we extend this result to two non-trivial subclasses of {\em bin packing with interval-graph conflicts},
where the interval graph is either {\em proper}, or consists of {\em constant length intervals}. The best known ratio for such instances
is $\frac{7}{3}$, due to Epstein and Levin~\cite{epstein2008bin}.
}

Existing algorithms for BPC often rely on 
initial {\em coloring} of the instance.
This enables to apply in later steps 
known techniques for bin packing, considering each color class (i.e., a subset of non-conflicting items) separately.
In contrast, our approach uses 
a refined packing of the original instance while eliminating
conflicts, thus generalizing techniques for classic BP.

Our first technical contribution is an enhancement of the {\em linear shifting} technique of~\cite{Vega1981BinPC}. This enables our scheme to enumerate in polynomial time over packings of relatively large items, while guaranteeing that these packings 
respect the group constraints.
Our {\em Shift \& Swap} technique considers the set of large items that are associated with many different groups as a classic BP instance, i.e., the group constraints are initially relaxed. Then the scheme applies to these items the linear shifting technique of~\cite{Vega1981BinPC}. In the process, items of the same group may be packed in the same bin. 
Our {\em Swapping} algorithm resolves all conflicts, 
with no increase in the total number of bins used  (see Sections~\ref{sec:Rounding of Large and Medium Items} and~\ref{sec:Large Items and Medium Items from Large Groups}).

A common approach used for deriving APTASs for BP is to pack in a 
bounded number of extra bins a set of discarded small items of total size $O(\eps)OPT$,
where $OPT=OPT(I)$ is the minimum number of bins required for packing the given instance, $I$, and $\eps \in (0,1)$ is the accuracy parameter of the scheme.
As shown in Appendix~\ref{sec:gbp_bp_example}, this approach may fail for GBP, e.g.,  
when the discarded items belong to the same group. Our second contribution is an algorithm that overcomes this hurdle.
The crux is to find a set of small items of total size $O(\eps)OPT$ containing $O(\eps)OPT$ items from each group. This would enable to pack these items in a small number of extra bins. Furthermore, 
the remaining small items should be feasibly assigned to partially packed $OPT$ bins. Our algorithm identifies such sets of small items
through an intricate combination of enumeration and
a greedy-based approach that utilizes the rounded solution of a {\em linear program}.

 \comment{In Section~\ref{sec:POC} we consider the price of clustering for GBP. We show that $PoC(GBP)$ crucially depends on two main attributes of each cluster $I_i$: $S(I_i)$, the total size of items in $I_i$, and $v_{max}^i$, the maximum cardinality of a group within $I_i$. Our second technical contribution is a non-trivial construction that enables to obtain a tight bound of $1 + \frac{1}{\alpha}$ on the price of clustering for GBP where
$R_i = \frac{S(I_i)}{v_{max}^i} \geq \alpha$,
 for any $0 < \alpha < 1$ and $1 \leq i \leq m$. Our analysis yields tight bounds for all possible ranges for $R_i$ (see Table~\ref{tbl:results}). Finally, using the bounds for $PoC(GBP)$, we derive bounds on $PoC_{\cA}(GBP)$, where ${\cA}$  is the APTAS presented in this paper.}


\subsection{Related Work}
\label{sec:related_work}

The classic bin packing problem is known to be NP-hard. Furthermore, it cannot be approximated within a ratio better than $\frac{3}{2}$, unless P=NP. 
This ratio is achieved by the simple First-Fit Decreasing algorithm~\cite{S94}.
The paper~\cite{Vega1981BinPC} 
presents an APTAS for bin packing, which uses at most $(1+\eps)OPT+1$ bins, for any fixed $\eps \in (0, 1/2)$. The paper~\cite{karmarkar1982efficient}  
gives an approximation algorithm that uses at most $OPT+O(\log^2(OPT))$ bins. 
The additive factor was improved in~\cite{Ro13} to
$O(\log OPT \cdot \log \log OPT)$. 
For comprehensive surveys of known results for BP see, e.g.,~\cite{C+13,C+17}.

The problem of {\em bin packing with conflicts (BPC)} 
was introduced in~\cite{JO97}. As BPC includes as a special case the classic {\em graph coloring} problem,
it cannot be approximated within factor 
$N^{1-\eps}$ for an input of $N$ items, for all $ \eps >0$, unless $P=NP$~\cite{Z07}. Thus, most of the research work focused on obtaining approximation algorithms for BPC on classes of conflict graphs that can be optimally colored in polynomial time.
Epstein and Levin~\cite{epstein2008bin} presented sophisticated algorithms for two such classes, namely, a $\frac{5}{2}$-approximation for BPC with a {\em perfect} conflict 
graph,\footnote{For the subclass of interval graphs the paper~\cite{epstein2008bin} gives a $\frac{7}{3}$-approximation algorithm.} 
and $\frac{7}{4}$-approximation for a {\em bipartite} conflict graph.

The hardness of approximation of GBP (with respect to absolute approximation ratio) follows from the hardness of BP, which is the special case of GBP where the conflict graph is an independent set.
A $2.7$-approximation algorithm for general instances follows from a result of~\cite{JO97}. 
Oh and Son~\cite{OS95} showed that a simple algorithm based on First-Fit outputs a packing of any GBP instance $I$ in 
$1.7OPT+2.19v_{max}$ bins, where $v_{max}= \max_{1 \leq j \leq n} | G_j|$.
The paper~\cite{mccloskey2005approaches} shows that some special cases of the problem are solvable in polynomial time. 
The best known ratio for GBP is $2$ due to~\cite{allornothing}.

Jansen~\cite{jansen1999approximation} presented an 
{\em asymptotic fully polynomial time approximation scheme (AFPTAS)}
for BPC on d-inductive conflict graphs,\footnote{A graph $G$ is {\em d-inductive} if the vertices of $G$ can be numbered such that each vertex is connected by an edge to at most $d$ lower numbered vertices.}  where $d \geq 1$ is some constant.
The scheme of~\cite{jansen1999approximation} uses for packing a given instance $I$  at most $(1+\eps)OPT + O(d/\eps^2)$ bins.
This implies that GBP admits an AFPTAS on instances where the maximum clique size is some constant $d$.
Thus, the existence of an asymptotic approximation scheme for general instances remained open.

Das and Wiese~\cite{DW17} introduced the problem of makespan minimization with bag constraints.
In this generalization of the classic makespan minimization problem, each job belongs to a {\em bag}. The goal is to schedule the jobs on a set of $m$ identical machines, for some $m \geq 1$, such that no two jobs in the same bag are assigned to the same machine, and the makespan is minimized. For the classic problem of makespan minimization with no bag constraints, there are known 
{\em polynomial time approximation scheme (PTAS)}~\cite{hochbaum1987using,MS:10} as well as
{\em efficient polynomial time approximation scheme (EPTAS)}~\cite{hochba1997approximation,alon1998approximation,MS:5,JKV16}. Das and Wiese~\cite{DW17} developed a PTAS for the problem with bag constraints. Later, Grage et al.~\cite{Jansen_et_al:2019} obtained an EPTAS.


\mysection{Preliminaries: Scheduling with Bag Constraints}
\label{sec:preliminaries}

Our scheme is inspired by the elaborate framework of Das and Wiese~\cite{DW17} 
for makespan minimization with bag constraints.
For completeness, we give below an overview of the scheme of~\cite{DW17}.
Given a set of jobs $I$
partitioned into bags and $m$ identical machines,
let $p_\ell > 0$ be the processing time of job $\ell \in I$.
The instance is scaled such that the optimal makespan is $1$.
The jobs and bags are then classified using the next lemma. 
\negA

\begin{lemma}
\label{lem:k_val}
For any instance $I$ and $\eps \in (0,1)$, there is an integer $k \in \{1,...,\lceil \frac{1}{\eps^2} \rceil \}$ such that $\sum_{\ell \in I : \: p_{\ell} \in [\varepsilon^{k+1},\varepsilon^{k})} \: p_{\ell} \leq \varepsilon^{2} m$.
\end{lemma}
\label{sec:dw}

\negA

A job $\ell$ is {\em small} if  
$p_{\ell} < \varepsilon^{k+1}$, {\em medium} if $p_{\ell} \in [\varepsilon^{k+1},\varepsilon^{k})$ and {\em large} if $p_{\ell} \geq \varepsilon^{k}$, where 
$k$ is the value found in Lemma~\ref{lem:k_val}. 
A bag is {\em large} if the number of large and medium jobs it contains is at least $\eps m$, and {\em small} otherwise. 

The scheme of~\cite{DW17} initially enumerates
over {\em slot patterns} so that large and medium jobs from large bags are optimally assigned  to the machines in polynomial time.
The enumeration is enhanced by using dynamic programming and a flow network to schedule also the large jobs from small bags. The medium jobs in each small bag are scheduled across the $m$ machines almost evenly, causing only small increase to the makespan.
 The small jobs are partitioned among {\em machine groups} with the same processing time and containing jobs from the same subset of large bags.
Then, a greedy approach is used with respect to the bags to schedule the jobs within each machine group, such that the overall makespan is at most $1+O(\eps)$.

Our scheme classifies the items and groups similar to the classification of jobs and bags in~\cite{DW17}. We then apply enumeration over patterns to pack the large and medium items.
Thus, Lemmas~\ref{lem:few_large_groups}, \ref{lem:dw1} and \ref{lem:classes} in this paper are adaptations of results obtained in \cite{DW17}. However, the remaining components of our scheme are different. One crucial difference is our use of a Shift \& Swap technique to round the sizes of large and medium items.
Indeed, rounding the item sizes using the approach of~\cite{DW17} may cause overflow in the bins,
requiring a large number of extra bins to accommodate the excess items.
 Furthermore, packing the {\em small} items 
 using $O(\eps) OPT$ extra bins requires new ideas (see Section~\ref{sec:APTAS}).
 
We use standard definitions of approximation ratio and {\em asymptotic} approximation ratio. Given a minimization problem  $\Pi$, let ${\cal A}$ be a polynomial-time algorithm  for $\Pi$. For an instance $I$ of $\Pi$, denote by 
$OPT(I)$ and ${\cal A}(I)$ the values of an optimal solution and the solution returned
by ${\cal A}$ for $I$, respectively. 
We say that ${\cal A}$ is a $\rho$-approximation algorithm for $\Pi$,
for some $\rho \geq 1$, if ${\cal A}(I) \leq  \rho \cdot OPT (I)$ 
for any instance $I$ of $\Pi$.  
 ${\cal A}$ is an {\em asymptotic} $\rho$-approximation for $\Pi$ if there is a constant $c\in \mathbb{R}$ such that ${\cal A}(I) \leq  \rho \cdot  OPT (I)+c$  for any instance $I$ of $\Pi$. An APTAS  for $\Pi$ is a family of algorithms $(A_{\eps})_{\eps>0}$ such that $A_\eps$ is a polynomial-time asymptotic $(1+\eps)$-approximation for each  $\eps>0$. When clear from the context, we use $OPT=OPT(I)$. 


\section{An APTAS for GBP}\label{APATS}
\label{sec:APTAS}

In this section we present an APTAS for GBP.
Let $OPT$ be the optimal number of bins for an instance $I$.
Our scheme uses as a subroutine a BalancedColoring algorithm proposed in \cite{allornothing} for the {\em group packing} problem (see the details in Appendix~\ref{sec:balanced_col}). 
Let $S(I)$ be the total size of items in $I$, i.e., $S(I) = \sum_{\ell \in [N]} s_\ell$.
Recall that $v_{max}$ is
the maximum cardinality of any group.
 The next lemma follows from a result of \cite{allornothing}.
 
 
 \begin{lemma}
	\label{lem:balancedColoring}
	Let $I$ be an instance of GBP. Then BalancedColoring packs $I$ in at most $\max \{ 2S(I), S(I) + v_{max} \}$ bins.
\end{lemma}

 By the above, given an instance $I$ of GBP, we can guess $OPT$ in polynomial time, by iterating over all integer values in $[1,\max \{ 2S(I), S(I) + v_{max}]$ and taking the minimal number of bins for which a feasible solution exists. 
 
 \comment{
 W.l.o.g., we may assume that $|G_j|= OPT$ for all $1 \leq j \leq n$.\footnote{Otherwise, we can add to $G_j$ dummy items of size $0$, with no increase to the number of bins in an optimal solution.}}

 \comment{
Using binary search, we can "guess" the optimal number of bins for packing a given GBP instance, $I$, and denote it by $OPT(I)$. , since otherwise $v_{max} \leq \frac{1}{\eps}$, implying that the conflict graph is $d$-inductive, where $d \leq 1/\eps$ is a constant. For such instances, GBP admits an AFPTAS~\cite{jansen1999approximation}.  
To simplify the scheme, we add "dummy items" to $I$, i.e, items of size 0, so that each group has the same number of items, without changing $OPT$. Thus, we assume from now on that each group has $OPT$ items. 
 . The packing is obtained in several steps, such that in each step we pack only a subset of the items in the same $OPT$ bins through all steps, and some steps require $O(\varepsilon)OPT$ extra bins, such that there is no overflow in any of the bins. We pack all discarded items from all steps together in $O(\varepsilon)OPT+c$ extra bins.  Hence, at the end of all steps, all items are packed feasibly in $(1+O(\varepsilon)) \cdot OPT+c$ bins. Our partition uses the next lemma. 
}
\comment{\begin{lemma}
\label{lem:k_val}
Given a GBP instance $I$ and $\eps \in (0,1)$, there is an integer $k \in \{1,...,\lceil \frac{1}{\eps^2} \rceil \}$ such that $\sum_{\ell \in I : \: s_{\ell} \in [\varepsilon^{k+1},\varepsilon^{k})} \: s_{\ell} \leq \varepsilon^{2} \cdot OPT$.
\end{lemma}}

Similar to Lemma~\ref{lem:k_val}, we can find a value of $k$, $1\leq k \leq \lceil \frac{1}{\eps^2} \rceil$,
satisfying 
$\sum_{\ell \in I : \: s_{\ell} \in [\varepsilon^{k+1},\varepsilon^{k})} \: s_{\ell} \leq \varepsilon^{2} \cdot OPT$. Now, we classify 
item $\ell$ as {\em small} if $s_{\ell} < \varepsilon^{k+1}$, {\em medium} if $s_{\ell} \in [\varepsilon^{k+1},\varepsilon^{k})$ and {\em large} otherwise.
A group is {\em large} if the number of large and medium items of that group is at least $\eps^{k+2} \cdot OPT$, and 
{\em small} otherwise. Given an instance $I$ of GBP and a constant $\eps \in (0,1)$, we also assume that $OPT > \frac{3}{\eps^{k+2}} $ (otherwise, the conflict graph is $d$-inductive, where $d$ is a constant, and the problem admits an AFPTAS~\cite{jansen1999approximation}).

\comment{\footnote{Otherwise, the conflict graph is $d$-inductive, where $d \leq \frac{2}{\eps ^{k+2}}$ is a constant, and the problem admits an AFPTAS~\cite{jansen1999approximation}.}} \comment{The next lemma follows from a result of ~\cite{DW17}. }

\begin{lemma}
\label{lem:few_large_groups}
There are at most $\frac{1}{\eps^{2k+3}}$ large groups.
\end{lemma}
 
    

\subsection{Rounding of Large and Medium Items}
\label{sec:Rounding of Large and Medium Items}
We start by reducing the number of distinct sizes for the large and medium items.
Recall that in the linear shifting technique we are given a BP instance of $N$ items
and a parameter $Q \in (0,N]$.
The items are sorted
in non-increasing order by sizes and then partitioned into classes. Each class (except maybe the last one) contains $\max\{Q,1\}$ items. The items in class $1$ (i.e., largest items) are discarded (the discarded items are handled in a later stage of the algorithm). The sizes of items in each class are then rounded up to the maximum size of an item in this class. For more details see, e.g.,~\cite{Vega1981BinPC}.

We apply linear shifting to the large and medium items in each large group
with parameter $Q = \lfloor \eps^{2k+4} \cdot OPT \rfloor$. Let $I, I'$ be the instance before and after the shifting over large groups, respectively.

\begin{lemma}
\label{lem:shiftingNotIncreaseOPT}
$OPT(I') \leq OPT(I)$.
\end{lemma}

\begin{lemma}
\label{lem:shiftingCanBeUsedForI}
 Given a feasible packing of $I'$ in $OPT$ bins, we can find a feasible packing of $I$ in $(1+O(\eps))OPT$ bins. 
\end{lemma}


Next, we round the sizes of large items in small groups. As the number of these groups may be large, we use the following {\em Shift \& Swap} technique.
We merge all of the large items in small groups into a single group, to which we apply linear shifting with parameter $Q = \lfloor{2\eps \cdot OPT}\rfloor$.
 In addition to items in class 1, which are discarded due to linear shifting, 
 we also discard the items in the last size class; these items are packed in a new set of bins
 (see the proof of Lemma~\ref{lem:discardFinal} in Appendix~\ref{Omitted Proofs}). 
 
\begin{lemma}
	\label{lem:rounding1}
	After rounding, there are at most $O(1)$ distinct sizes of large and medium items from large groups, and large items from small groups.
\end{lemma}

 Relaxing the {\em feasibility} requirement for the packing of rounded large items from small groups, the statements of Lemma~\ref{lem:shiftingNotIncreaseOPT} and  Lemma~\ref{lem:shiftingCanBeUsedForI} hold for these items as well.
 To obtain a feasible packing of these items, we apply a Swapping subroutine which resolves the possible conflicts caused while packing the items.

\comment
{\subsection{Rounding of Large and Medium Items}
\label{sec:Rounding of Large and Medium Items}

$Q = \lfloor \eps \cdot OPT \rfloor$

Our method of packing the large and medium items (see Section~\ref{sec:Large Items and Medium Items from Large Groups}) runs in polynomial time in $N$ only if there are at most $O(1)$ different sizes of large and medium items, where $0 < \varepsilon < 1$ is a constant. We use rounding of item sizes to obtain this property, by  shifting technique, described as follows for completeness.

{\em Shifting technique} with Parameter $Q$ is a rounding scheme for BP which is described as follows. Given an instance $I$, The items of $I$ are sorted in a non decreasing order by their sizes. Then, the first $Q$ items belong to size class 1, the next $Q$ items belong to size class 2, etc. Then, the items of the last size class are discarded, and for any other size class we round the sizes of its items to the size of the maximal item in the size class. Let $I_L$ be the set of discarded items. The resulting instance is denoted by $I'$ and the following inequalities hold: $OPT(I' \cup I_L) \leq OPT(I) +Q$. A more detailed description of the technique can be found in \cite{Vega1981BinPC}. In particular we apply shifting technique with parameter $Q = \lfloor{\varepsilon \cdot OPT}\rfloor$ on all large and medium items. The discarded items are of total size of $\lfloor{\varepsilon \cdot OPT}\rfloor$. We note that the total number of large and medium items is at most $\frac{OPT}{\eps^{k+1}}$. Thus, by the selection of $Q$, we obtain a constant number of sizes.
}



Our scheme packs in each step a subset of items, using $OPT$ bins, while discarding some items. The discarded items are packed later in a set of $O(\eps) \cdot OPT + 1$ extra bins.
In Section~\ref{sec:Large Items and Medium Items from Large Groups} we pack the large  and medium items 
using enumeration over patterns 
followed by our Swapping algorithm to resolve conflicts. Section~\ref{subsection:NEW_Small_Items} presents an algorithm for packing the small items by combining recursive enumeration (for relatively ``large" items) with
a greedy-based algorithm that utilizes the rounded solution of a linear program
(for relatively ``small" items).
In Section~\ref{sec:putting it all together} we show that the 
components of our scheme combine together to an APTAS for GBP. 

\subsection{Large and Medium Items}
\label{sec:Large Items and Medium Items from Large Groups}
\label{subsection:Medium Items of Small Groups}
The large items and medium items from large groups are packed in the bins using {\em slot patterns}. Let $G_{i_1}, \ldots, G_{i_L}$ be the large groups, and
let `u' be a label representing all the small groups. Given the modified instance $I'$,
a slot is a pair $(s_{\ell},j)$, where $s_{\ell}$ is the rounded size of a large or medium item $\ell \in I'$ and $j \in \{i_1, \ldots, i_L\} \cup \{u\}$. A {\em pattern} is a multiset $\{t_1,\ldots, t_{\beta}\}$ 
for some $1 \leq \beta \leq {\lfloor{\frac{1}{\varepsilon^{k+1}}}\rfloor}$,
where $t_i$ is a slot for each $i \in [\beta]$.\footnote{Recall that the number of medium/large items that fit in a single bin is at most ${\lfloor{\frac{1}{\varepsilon^{k+1}}}\rfloor}$.}

\comment{A slot is characterized by a size, and by a {\em label}. A label can represent explicitly one of the large groups, or can denote `small group' with no indication to which small group it belongs. Denote by $u$ the label for all the small groups. Let $G_{i_1}, \ldots, G_{i_L}$ be the large groups.    
Formally, a slot is a pair $(s_{\ell},j)$, where $s_{\ell}$ is the size of an item $\ell \in I$ and $j \in \{i_1, \ldots, i_L\} \cup \{u\}$. A {\em pattern} is a multiset $\{t_1,\ldots, t_{\beta}\}$ containing at most ${\lfloor{\frac{1}{\varepsilon^{k+1}}}\rfloor}$ elements, where $t_i$ is a slot for each $i \in [\beta]$.  }


\begin{lemma}
\label{lem:dw1}
By using enumeration over patterns, we find a pattern for each bin for the large and medium items, such that these patterns correspond to an optimal solution. The running time is $O(N^{O(1)})$.
\end{lemma}


Given slot patterns corresponding to an optimal solution, large and medium items from large groups can be packed optimally, since they are identified both by a label and a size. On the other hand, large items from small groups are identified solely by their sizes. A greedy packing of these items, relating only to their corresponding patterns, may result in conflicts (i.e., two large items of the same small group are packed in the same bin). Therefore, we incorporate a process of swapping items of the same (rounded) size between their hosting bins, until there are no conflicts.

Given an item $\ell$ that conflicts with another item in bin $b$, for an item $y$
in bin $c$ such that $s_{\ell} = s_y$, $swap(\ell,y)$ is {\em bad} if it causes a conflict (either
because $y$ conflicts with an item in bin $b$, $\ell$ conflicts with an
item in bin $c$, or $c = b$); otherwise, $swap(\ell,y)$ is {\em good}.
We now describe our 
algorithm for packing the large items from small groups. 

Let $\zeta$ be the given slot patterns for $OPT$ bins. Initially, the items are packed by these patterns, where items from small groups are packed ignoring the group constraints. This can be done simply by placing an arbitrary item of size $s$ from some small group in each slot $(s,u)$. If $\zeta$ corresponds to an optimal solution, we meet the capacity constraint of each bin. However, this may result with conflicting items in some bins. 
Suppose there is a conflict in bin $b$. Then for one of the conflicting items, $\ell$,
we find a good $swap(\ell,y)$ with item $y$ in a different bin, such that $s_{y} = s_{\ell}$. We repeat this process until there are no conflicts. We give the pseudocode of Swapping in Algorithm~\ref{Alg:Swapping}. 


\begin{algorithm}[h]
	\caption{$Swapping(\zeta, G_1, \ldots, G_n)$}
	\label{Alg:Swapping}
	\begin{algorithmic}[1]
		\State{Pack the large and medium items from large groups in slots corresponding to their sizes and by labels.}
		\State{Pack large items from small groups in slots corresponding to their sizes.}
		\While{there is an item $\ell$ involved in a conflict}
		\State Find a good $swap(\ell,y)$ and resolve the conflict.
		\EndWhile
	\end{algorithmic}
\end{algorithm}

\begin{theorem}
\label{lem:swap}
Given a packing of large and medium items by slot patterns corresponding to an optimal solution, Algorithm~\ref{Alg:Swapping} resolves all conflicts in polynomial time. 
\end{theorem}

\comment{
\negA
\negA
\begin{proof}
We prove that for each conflict involving an item $\ell \in G_i$ of size $s_\ell$ in bin $b$, there is an item $y \in G_j \neq G_i$ of size $s_{y} = s_\ell$ in bin $c \neq b$, such that $swap(\ell,y)$ is good. Consider 
a packing of large and medium items by a slot pattern corresponding to an optimal solution. Then, the items are packed in $OPT$ bins with no
overflow, and the only conflicts may occur among items from small groups.

Due to shifting with parameter $Q = \lfloor 2\eps \cdot OPT \rfloor$ for large items from small groups, there are $\lfloor 2\eps \cdot OPT \rfloor-1$ items of size $s_\ell$ in addition to $\ell$ (recall that the last size class, which may contain less items, is discarded).
 We prove that the number of items $y$ for which $swap(\ell,y)$ is bad is at most $\lfloor 2\eps \cdot OPT \rfloor-2$; therefore, there exists an item $y$ of size $s_{\ell}$, for which $swap(\ell,y)$ is good. 
 We note that $swap(\ell,y)$ is bad if (at least) one of the following holds: $(i)$ $y$ belongs to a group $G_j$ which has an item in bin $b$, or
 $(ii)$ there is an item from $G_i$ in bin $c$.

We handle $(i)$ and $(ii)$ separately. $(i)$ The number of items of size $s_\ell$ from groups $G_j$ that contain an item in bin $b$ is bounded by $\frac{1}{\eps^{k}} \cdot \eps^{k+2} OPT = \eps^2 OPT$ 
 (i.e., the number of large items from small groups in bin $b$ times the number of items in a small group that has an item in $b$). 
 Indeed, since all of 
 these groups are small, each group contains at most $\eps^{k+2} OPT$ large items. $(ii)$ The number of items of size $s_\ell$ in bins which contain items from $G_i$ is bounded by $\frac{1}{\eps^{k}} \eps^{k+2} OPT = \eps^2 OPT$ (i.e., the number of items in a bin times the number of items in group $G_i$); since $G_i$ is small, it contains at most $\eps^{k+2} OPT$ large items. 
Using the union bound, the number of bad swaps 
  for $\ell$, i.e., $swap(\ell,y)$ for some item $y$,
   is at most 
   $2\eps^2 OPT$. We have $2\eps^2OPT < \eps OPT < \eps OPT+\eps OPT-3 \leq \lfloor 2\eps \cdot OPT \rfloor-2.$ The first inequality holds since we may assume that $\eps < \frac{1}{2}$.
For the second inequality, we note that $OPT > \frac{3}{\eps^{k+2}} > \frac{3}{\eps}$. 

We conclude that in the size class of $\ell$ there is an item $y$
such that $swap(\ell,y)$ is good. 
We now show that the Swapping algorithm is polynomial in $N$. We note that items of some group are in conflict only if they are placed in the same bin. As these are only large items, an item may conflict with at most $\frac{1}{\eps^{k}}$ items. Hence, there are at most $\frac{N}{\eps^{k}} = O(N)$ conflicts.
As finding a good swap takes at most $O(N)$, the overall running time of Swapping is $O(N^2)$. \qed
\end{proof}
}


 We use the Swapping algorithm for each possible guess of patterns to obtain a feasible packing of the large items and medium items from large groups in $OPT$ bins. 
 \comment{
 Unfortunately, even if we pack the small items from small groups in an optimal manner, we may still violate the group constraints, as the feasible packing of large items from small groups may not be optimal. Indeed, this is due to the fact that slots allocated to items from small groups do not specify the group from which an item is selected. This may cause a conflict between  a small item and a large item in the same small group, when packed in the same bin. We resolve such conflicts  in the process of packing the small items from small groups (see Section~\ref{subsection:Small Items from Small Groups}).
}

Now, we discard the medium items from small groups and pack them later in a new set of bins with other discarded items. 
This requires only a small number of extra bins 
(see the proof of Lemma~\ref{lem:discardFinal} in Appendix~\ref{Omitted Proofs}). 

\comment{
\mysubsection{Small Items from Large Groups}
\label{subsection:Small Items from Large Groups}

We now partition the $OPT$ bins, partially packed with large and medium items,
into {\em types}. Each type contains bins having the same total size of packed large/medium items; also, the items packed in each bin type belong to the same subset of large groups, and the same number of slots is allocated in these bins to items from small groups.
\negA
\negA

\begin{lemma}
\label{lem:classes}
There are $O(1)$ types of bins.
\end{lemma}

\negA
\negA

Algorithm $EnumGroups$ enumerates over all feasible packings of the small items from large groups, by assigning groups of items to {\em bin types}. In the sequel, the items assigned to each type are packed in the bins belonging to this type. Given a set of groups with remaining small items, $G_1, \ldots , G_W$, and the bin types $B_1, \ldots, B_R$, we first apply shifting to the items of each group $G_1, \ldots , G_W$, separately, with parameter $Q = \left \lfloor{\varepsilon^{2k+4} \cdot OPT}\right \rfloor$. Then, the set of (at most) $Q$ largest items are discarded from each group. Similar to the rounding for large items from large groups, Lemmas~\ref{lem:shiftingNotIncreaseOPT} and~\ref{lem:shiftingCanBeUsedForI} guarantee that a packing of the rounded instance can be used for packing the original instance, with $O(\eps)OPT$ extra bins used for the discarded items. An assignment of groups to the bin types is {\em feasible} if, for every bin type $B_{\ell}$, the following conditions are met: $(i)$ there are at most $|B_{\ell}|$ items from each group $G_j$ assigned to $B_{\ell}$, and $(ii)$ the total size packed in $B_{\ell}$, including the previously packed medium and large items, is at most $|B_{\ell}|$.

We now describe the steps of EnumGroups. 
Initially, the algorithm guesses a feasible partition of the items in $G_j$ among the bin types $B_1,\ldots,B_R$, for each group $G_j$, $1 \leq j \leq W$.
Let $G_1(B_\ell), \ldots, G_W(B_\ell)$ be the set of small items assigned to bin type $B_\ell$, $1 \leq \ell \leq R$. We start by packing the items in $G_1(B_\ell)$; then, the items in $G_2(B_\ell), \ldots , G_W(B_\ell)$ are packed by a recursive call to EnumGroups.
All feasible partitions of groups to bin types are enumerated by EnumGroups. In this process, a partition that corresponds to an optimal solution is considered. We give the pseudocode of EnumGroups in Algorithm~\ref{Alg:small item from large groups} (see Appendix~\ref{sec:Algorithm LargeGroups}).


\negA
\negA

\begin{theorem}
\label{lem:larggroupComplexity}
\label{lem:larggroupCorrectness}
Let $W \geq 1$ be some constant. Then Algorithm~\ref{Alg:small item from large groups} outputs in polynomial time a packing which corresponds to an optimal packing of the items in $G_1, \ldots , G_W$.
\end{theorem}
\negA
\negA
 \negA

\comment{\begin{lemma}
\label{lem:larggroupComplexity}

Let $W \geq 1$ be some constant. Then the running time of Algorithm~\ref{Alg:small item from large groups} is polynomial in $N$.
\end{lemma}}

\comment{\begin{lemma}
\label{lem:discarded}
The discarded items from the shifting technique can be packed in $2\varepsilon \cdot OPT+1$ additional bins.
\end{lemma}}

\negA

\comment{\begin{algorithm}[h]
	\caption{$GreedyPack(\{G_{i_1}^s, \ldots, G_{i_H}^s\}, b_1, \ldots, b_{OPT})$}
	\label{Alg:greedyPack}
	\begin{algorithmic}[1]
		\State{Sort the bins in a non-increasing order by the total size of packed items. \label{step:greedy}}
		\State{Renumber the bins from $1$ to $OPT$ using this order.}
		\For {$j = 1,\ldots,H$}
		\State{Sort $G^s_{i_j}$ in a non-increasing order by sizes.}
		\State{Let $y_{i_j}$ be the largest item in $G_{i_j}^s$.}
		\EndFor
		\For {$z = 1,\ldots,OPT$}
		\State{Add to bin $z$ the items $y_{i_1}, \ldots, y_{i_H}$.}
		\While{total size of bin $z$ $>$ 1}
		\State{Select a group $G_{i_j}^s \in \{G_{i_1}^s, \ldots, G_{i_H}^s\}$ s.t. $y_{i_j}$ is not last in $G_{i_j}^s$.}
		\If {cannot complete last step}
		\State{return $failure$}
		\EndIf
		\State{Return $y_{i_j}$ to $G_{i_j}^s$.}
		\State{Let $y'_{i_j}$ to be the next largest item in $G_{i_j}^s$.}
		\State{Add $y'_{i_j}$ to bin $z$. }
		\EndWhile
		\For {$j = 1,\ldots,H$}
		\If {$G_{i_j}^s$ has a large item in bin $z$} 
		\State{\label{step:discardsmalloverlarge} discard the small item.}
		\EndIf
		\EndFor
		\EndFor
	\end{algorithmic}
\end{algorithm}}

\negA
}
\subsection{Small Items}
\label{subsection:NEW_Small_Items}

Up to this point, all large items and the medium items from large groups are feasibly packed in $OPT$ bins. We proceed to pack the small items. Let $I_0, B$ be the set of unpacked items and the set of $OPT$ partially packed bins, respectively. The packing of the small items is done in four phases: an {\em optimal phase}, an {\em eviction phase}, a {\em partition phase} and a {\em greedy phase}. 

The optimal phase is an iterative process consisting of a constant number of iterations.
In each iteration, a subset of bins is packed with a subset of items whose (rounded) sizes are large relative to the free space in each of these bins. As these items belong to a 
{\em small} collection of groups among $G_1, \ldots, G_n$, they can be selected using enumeration. Thus, we obtain a packing of these items which corresponds to an optimal solution.
For packing the remaining items, we want each item to be small relative to the free space in its assigned bin. To this end, in the eviction phase  we discard from some bins items
of non-negligible size (a single item from each bin).
Then, in the partition phase, the unpacked items are partitioned into a constant number of sets satisfying certain properties, which guarantee that these items can be feasibly packed 
in the available free space in the bins. Finally, in the greedy phase, the items in each set are packed in their allotted subset of bins greedily, achieving a feasible packing of all items, except for a small number of items from each group, of small total size. The pseudocode of our algorithm for packing the small items is given in Algorithm~\ref{Alg:PackSmallItems}.

\myparagraph{The optimal phase}
For any $b \in B$, denote by $f^0_b$ the free capacity in bin $b$, i.e., $f^0_b = 1-\sum_{\ell \in b} s_{\ell}$. We say that item $\ell$ is {\em $b$-negligible} if $s_{\ell} \leq \eps^2 f^0_b$, and $\ell$ is {\em $b$-non-negligible} otherwise. We start by classifying the bins into two disjoint sets. 
Let $E_0 = \{b \in B |~ 0 < f^0_b < \eps \}$ and $D_0 = B \setminus E_0$.

We now partition $B$ into {\em types}. Each type contains bins having the same total size of packed large/medium items; also, the items packed in each bin type belong to the same set of large groups, and the same number of slots is allocated in these bins to items from small groups. Formally, for each pattern $p$ we denote by $t_{p}$ the subset of bins packed with $p$.\footnote{For the definition of patterns see Section~\ref{sec:Large Items and Medium Items from Large Groups}.} Let $T$ denote the set of bin types. Then $|T|=|P|$, where $P$ is the set of all patterns. The {\em cardinality} of type $t \in T$ is the number of bins of this type.
We use for the optimal phase algorithm $RecursiveEnum$ (see the pseudocode in Algorithm~\ref{Alg:RecursiveEnum}). 

\begin{lemma}
\label{lem:classes}
There are $O(1)$ types before Step~\ref{step:beg} of Algorithm~\ref{Alg:RecursiveEnum}.
\end{lemma}

Once we have the classification of bins, each type $t$ of cardinality smaller than $1/\eps^4 $ is padded with empty bins so that 
$|t| \geq 1/\eps^4$. 
An item $\ell$ is {\em $t$-negligible} if $\ell$ is $b$-negligible for all bins $b$ of type $t$ (all bins in the same type have the same free capacity), and {\em $t$-non-negligible} otherwise. Denote by $I'_t$ the large/medium items that are packed in the bins of type $t$, and let $I_t(g)$ be the set of small items that are packed in $t$ in some solution $g$ (in addition to $I'_t$).
For any $1 \leq i \leq n$, 
a group $G_i$ is {\em $t(g)$-significant} if $I_t(g)$ contains at least $\eps^4 |t|$ $t$-non-negligible items from $G_i$, and $G_i$ is {\em $t(g)$-insignificant} otherwise.

{\em RecursiveEnum} proceeds in iterations.
In the first iteration, it
{\em guesses} for each type $t \subseteq E_0$ a subset of the items $I_{t}(g_{opt}) \subseteq I_0$, where $g_{opt}$ corresponds to an optimal solution for completing the packing of $t$. Specifically, $RecursiveEnum$ initially guesses $L(t,g_{opt})$ groups that are $t(g_{opt})$-significant: $G_{i_1}, \ldots, G_{i_{L(t,g_{opt})}}$. For each $G_{i_j}, j \in \{1,\ldots,L(t,g_{opt})\}$, the algorithm guesses which items of $G_{i_j}$ are added to $I_{t(g_{opt})}$. Since the number of guesses might be exponential, we apply to $G_{i_j}$
linear shifting as follows. 
Guess $\lceil \frac{1}{\eps^3} \rceil$ {\em representatives} in $G_{i_j}$, of sizes $s_{\ell_1} \leq s_{\ell_2} \leq \ldots s_{\ell_{\lceil 1/\eps^3 \rceil}}$. 
The $k$th representative is the largest item in size class $k$, $1 \leq k \leq \lceil \frac{1}{\eps^3} \rceil$ for the linear shifting of $G_{i_j}$ in type $t$.
Using the parameter $Q^t_{i_j} = \eps^3|t|$, the item sizes in class $k$ are rounded up to $s_{\ell_k}$, for $1 \leq k \leq \lceil \frac{1}{\eps^3} \rceil$. Given a correct guess of the representatives, the actual items in size class $k$ are selected at the end of 
algorithm $RecursiveEnum$ (in Step~\ref{step:sizeclassGreedy}).
 Denote the chosen items from $G_{i_j}$ to bins of type $t$ by $G_{i_j}^t$.

We now extend the definition of patterns for each type $t$. A slot is a pair  $(s_{\ell},j)$, where $s_{\ell}$ is the (rounded) size of a $t$-non-negligible item $\ell \in 
I_t(g_{opt})$, and there is a label for each $t(g_{opt})$-significant group $G_{i_j}$,
$j \in \{1, \ldots, L(t,g_{opt})\}$.\footnote{Note that we do not need a label for the $t(g_{opt})$-insignificant groups, because their items are packed separately.} A {\em t-pattern} is a multiset $\{q_1,\ldots, q_{\beta_t}\}$ containing at most ${\lfloor{\frac{1}{\varepsilon^{2}}}\rfloor}$ elements, where $q_i$ is a slot for each 
$i \in \{1, \ldots, \beta_t \}$. Now, for each type $t \in T$ we use enumeration over patterns for assigning  $G_{i_1}^t, \ldots, G_{i_{L(t,g_{opt})}}^t$ to bins in $t$. This completes the first iteration, and the algorithm proceeds recursively. 

We now update $D_0, E_0$ for the next iteration by removing from $E_0$ bins $b$ that have a considerably large free capacity with respect to $f^0_b$. For each $b \in B$, let $f^1_{b}$ be the capacity available in $b$ after iteration 1.
Then $E_1 = \{b \in E_0 |~ 0 < f^1_b < \eps f^0_b\}$ and 
$D_1 = B \setminus E_1$.

Now, each type $t \in T$ is partitioned into {\em sub-types} that differ by the packing of $I_t(g_{opt})$ in the first iteration. 
The set of types $T$ is updated to contain these sub-types.
At this point, a recursive call to $RecursiveEnum$ computes for each bin type $t \subseteq E_1$ a guessing and a packing of its $t$-non-negligible items.\footnote{An item is $b$-non-negligible w.r.t $f^1_b$ in this iteration, or w.r.t $f^h_b$ in iteration $h+1, h \in \{0,\ldots,\alpha-1\}$.} We repeat this recursive process $\alpha = \frac{1}{\eps}+5 = O(1)$ times. 

Let $G_i^t$ be the subset of items (of rounded sizes) assigned from $G_i$ to bins of type $t$ at the end of $RecursiveEnum$, for $1 \leq i \leq n$ and $t \in T$. Recall that  the algorithm did not select specific items in $G_i^t$; that is, we only have their rounded sizes and the number of items in each size class. The algorithm proceeds to pack items from $G_i$ in all types $t$ for which $G_i$ was $t(g_{opt})$-significant in some iteration.
Let $T_{G_i}$ be the set of these types. The algorithm considers first the type 
$t \in T_{G_i}$ for which the class $C$ of largest size items contains the item of maximal size, where the maximum is taken over all types $t \in T_{G_i}$.
The algorithm packs in bins of type $t$ the $Q_i^t$ largest remaining items in $G_i$ in the slots allocated to items in $C$; it then proceeds similarly to the remaining size classes in types $t \in T_{G_i}$ and the remaining items in $G_i^t$. 

\comment{
\begin{lemma}
$E_{\alpha} \subseteq E_{\alpha-1} \subseteq \ldots E_0$. In addition, the total size available in $E_0$ is at most $\eps OPT$. 
\end{lemma}

\begin{proof}
Let there be a bin $b \in E_i$. After the i+1 iteration, this bin can either be added to $E_{i+1}$ 
    $E_i \subseteq E_{i-1}, i \in \{1,\ldots,\alpha\}$ 
\end{proof}
}

\begin{algorithm}[htb]
	\caption{$RecursiveEnum(I_0,B)$}
	\label{Alg:RecursiveEnum}
	\begin{algorithmic}[1]
		\State{Let $f^0_b$ be the remaining free capacity in bin $b \in B$. \label{step:beg}}
        \State{Let $E_0 = \{b \in B | 0 < f^0_b < \eps \}$ and $D_0 = B \setminus E_0$.}
		\State{Denote by $T$ the collection of bin types.}
		\For {$h = 0,\ldots,\alpha$}
		\For{all types $t \subseteq E_h$}
		\If {$|t| < \frac{1}{\eps^4}$}
		\State{increase the cardinality of $t$ to $\frac{1}{\eps^4}$.\label{step:padding}}
		\EndIf
		\State{Guess $t(g_{opt})$-significant groups: $G_{i_1}, \ldots, G_{i_{L(t,g_{opt})}}$\label{step:t-significantGroups}}
		\For{$j = 1,\ldots,L(t,g_{opt})$}
		\State{Guess the number of items from $G_{i_j}$ to be added to bins of type $t$.\label{step:GuessNumberOfItems}}
		\State{Guess a representative for each size class of $t$-non-negligible items of $G_{i_j}$ 
		$~~~~~~~~~~~~~~~~$ for linear shifting.\label{step:Representative}}
		\EndFor
		\State{Guess $|t|$ $t$-patterns for bins in $t$ using the sizes after linear shifting of
			$~~~~~~~~~~~~~~~~$ $G_{i_1}, \ldots ,G_{i_{L(t,g_{opt})}}$.\label{step:t-pattern}}
		\State{Replace type $t$ in $T$ by all of the sub-types of $t$.\label{step:subTypes}}
		\EndFor
		\State{Let $f^{h+1}_b$ be the remaining free capacity in bin $b \in B$.}
		
		\State{Let $E_{h+1} = \{b \in E_h | 0 < f^{h+1}_b < \eps f^h_b\}$ and 
			$D_{h+1} = B \setminus E_{h+1}$.
			\label{step:UpdateDE}}
		
		\EndFor
		\State{Complete the packing of all size classes by assigning items greedily.\label{step:sizeclassGreedy}}
	\end{algorithmic}
\end{algorithm}

\begin{lemma}
\label{lem:recursiveOPT}
The following hold for RecursiveEnum: (i) the running time is polynomial; (ii) the increase in the number of bins in Step~\ref{step:padding} is at most $\eps OPT$; (iii) In Step~\ref{step:Representative} we discard at most $\eps OPT$ items from each group of total size at most $\eps OPT$. (iv) One of the guesses in Steps~\ref{step:t-significantGroups},~\ref{step:Representative} corresponds to an optimal solution.
\end{lemma}

\myparagraph{The eviction phase}
One of the guesses in the optimal phase corresponds to an optimal solution. For simplicity, henceforth assume that we have this guess. Recall that $E_{\alpha}$ is the set of all bins $b$ for which $0 < f^\alpha_b < \eps f^{\alpha-1}_b$. In Step~\ref{step:eviction} of $PackSmallItems$ (Algorithm~\ref{Alg:PackSmallItems})
 we evict an item from each $b \in E_{\alpha}$ such that the available capacity of $b$ increases to at least $\frac{f^\alpha_b}{\eps}$. This is done greedily: 
 consider the bins in $E_\alpha$ one by one in arbitrary order. From each bin
 discard a small item $\ell \in G_i$, for some $G_i$, $1 \leq i \leq n$, such that the following hold: (i) $s_\ell \geq \frac{f^\alpha_b}{\eps}$, and (ii) less than $\eps OPT$ items were discarded from $G_i$ in this phase.
 Since $\alpha$ is large enough, this phase can be completed successfully, as shown below. Let $T = \{t_1, \ldots,t_{\mu}, t'\}$ be the types after the optimal phase, where $t'$ is a new type such that $|t'| = \eps OPT$. Bins of type $t'$ are empty, i.e., each bin $b$ of type $t'$ has free space $1$.   
Denote by $f(t)$ the free space in each bin $b$ of type $t$ after the eviction phase, and let $I_L$ be the large items from small groups (already packed in the bins). 

\begin{lemma}
\label{lem:eviction}
After Step~\ref{step:ExtraType} of $PackSmallItems$ there exists a partition of $I_{\alpha}$ into types $I_{t_1},\ldots,I_{t_\mu}, I_{t'}$, for which the following hold. For each $t \in T$, (i) 
$|G_j^t| = |G_j \cap I_t| \leq |t|-|(I'_t \setminus I_L) \cap G_j|$, for all $1 \leq j \leq n$. (ii) for any $\ell \in I_t: s_{\ell} \leq \eps f(t)$, and (iii) $S(I_t) \leq f(t) |t|$.
\end{lemma}

We explain the conditions of the lemma below.
\myparagraph{The partition phase}
Let $T$ be the set of types after Step~\ref{step:ExtraType} of Algorithm~\ref{Alg:PackSmallItems}, and
$I_\alpha$ the remaining unpacked items.\footnote{Recall that we consider only items that were not discarded in previous steps, as discarded items are packed in a separate set of bins.}
We seek a partition of $I_{\alpha}$ into subsets associated with bin types such that
the items assigned to each type $t$ are relatively tiny; also, the total size and the cardinality of the set of items assigned to $t$ allow to feasibly pack these items in bins of this type. This is done by proving that a polytope representing the conditions in Lemma~\ref{lem:eviction} has vertices at points which are integral up to a constant number of coordinates. Each such coordinate, $x_{\ell,t}$, corresponds to a fractional selection of some item $\ell \in I_\alpha$ to type $t \in T$.
We use $G_j$  to denote the subset of remaining items in $G_j$, $1 \leq j \leq n$.

Formally, we define a polytope $P$ as the set of all points ${x} \in [0,1]^{I_\alpha \times T}$ which satisfy the following constraints.
\begin{equation*}
	\begin{array}{lcl}
	\forall \ell \in I_{\alpha}, t \in T \textnormal{ s.t. } s_{\ell} > \eps f(t)&~:~& \displaystyle x_{\ell,t} = 0\\
	\rule{0pt}{1.8em}
	\forall t \in T&~:~& \displaystyle \sum_{\ell \in I_{\alpha}} x_{\ell,t} s_\ell \leq f(t) |t|\\
	\rule{0pt}{1.8em}
	\forall \ell \in I_{\alpha}&~:~& \displaystyle \sum_{t \in T} x_{\ell,t} = 1\\
	\rule{0pt}{1.8em}
	\forall 1 \leq j \leq n, t \in T&~:~& \displaystyle  \sum_{\ell \in G_j} x_{\ell,t} \leq |t|-|(I'_t \setminus I_L) \cap G_j|
	\end{array}
\end{equation*}

The first constraint refers to condition $(ii)$ in Lemma~\ref{lem:eviction}, which implies that items assigned to type $t$ need to be tiny w.r.t the free space in the bins of this type. The second constraint reflects condition $(iii)$ in the lemma, which guarantees that the items in $I_t$ can be feasibly packed in the bins of type $t$. The third constraint ensures that overall  
each item $\ell \in I_\alpha$ is (fractionally) assigned exactly once.

The last constraint reflects condition $(i)$ in Lemma~\ref{lem:eviction}. Overall, we want to have at most $|t|$ items of $G_j$ assigned to bins of type $t$. Recall that these bins may already contain large/medium items from $G_j$ packed in previous steps. While large/medium items from {\em large} groups are packed optimally, the packing of large items from {\em small} groups, i.e., $I_L$, is not necessarily optimal. In particular, the items in $I_L$ packed by our scheme in bins of type $t$ may not appear in these bins in the optimal solution  $g_{opt}$ to which our packing corresponds.
Thus, we exclude these items and only require that the number of items assigned from $G_j$ to bins of type $t$ is bounded by $|t|-|(I'_t \setminus I_L) \cap G_j|$.

\begin{theorem}
\label{thm:partitionPhase}
 Let $x \in P$ be a vertex of $P$. Then, $$|\{\ell  \in I_{\alpha} ~ |~\exists t\in T: ~ x_{\ell,t} \in (0,1)\}| = O(1).$$ 
\end{theorem}

By Theorem~\ref{thm:partitionPhase}, we can find a feasible partition (with respect to the constraints of the polytope) by finding a vertex of the polytope, and then discarding the $O(1)$ fractional items. These items can be packed in $O(1)$ extra bins. 
By Lemma~\ref{lem:eviction} we have that $P\neq \emptyset$; thus, a vertex of $P$ exists and the partition can be found in polynomial time.

\myparagraph{The greedy phase}
In this phase we pack the remaining items using algorithm $GreedyPack$ (see the pseudocode in Algorithm~\ref{Alg:greedyPack}).
 Let $G_{1}^t, \ldots, G_{n}^t$ be the items in $I_t$ from each group, and let $S(I_t)$ be the total size of these items, i.e., $S(I_t) = \sum_{j = 1}^{n} \sum_{\ell \in G_{j}^t} s_{\ell}$.

\begin{algorithm}[htb]
	\caption{$GreedyPack(I_t = \{G_{i_1}^t, \ldots, G_{i_H}^t\}, t = \{b_1,\ldots,b_{|t|}\}$}
	\label{Alg:greedyPack}
	\begin{algorithmic}[1]
		\For {$j = 1,\ldots,H$}
		\State{Sort $G^t_{i_j}$ in a non-increasing order by sizes.}
		\EndFor
		\State{Let $y_{i_j}$ be the largest remaining item in $G_{i_j}^t$, $j=1, \ldots , H$.}
		\For {each bin $b \in t$}
		\State{Add to bin $b$ the items $y_{i_1}, \ldots, y_{i_H}$.}
		\While{total size of items packed in bin $b$ $>$ 1}
		\State{Select a group $G_{i_j}^t \in \{G_{i_1}^t, \ldots, G_{i_H}^t\}$ such that $y_{i_j}$ is not last in $G_{i_j}^t$.}
		\If {cannot complete last step}
		\State{return $failure$}
		\EndIf
		\State{Return $y_{i_j}$ to $G_{i_j}^t$.}
		\State{Let $y'_{i_j}$ be the next largest item in $G_{i_j}^t$.}
		\State{Add $y'_{i_j}$ to bin $b$. }
		\EndWhile
		\For {$j = 1,\ldots,H$}
		\If {$G_{i_j}^t$ has a large item in bin $b$} 
		\State{\label{step:discardsmalloverlarge} discard the small item.}
		\EndIf
		\EndFor
		\EndFor
	\end{algorithmic}
\end{algorithm}

\comment{
Algorithm GreedyPack starts by sorting the small items in each small group in non-increasing order by sizes. 
Denote the items of $G_{j}^t$ in the sorted order by $s_{1,j_t}, s_{2,j_t},\ldots$.
Let $S_{x,t} (G_{j}^t)$ be the total size of the $x$ largest items in $G_{j}^t$ for an integer $x \geq 0$; that is, the total size of items of sizes $s_{1,j_t}, \ldots , s_{x,j_t}$.\footnote{An empty sum is equal to $0$.} 

Then, $S_{x,t} = \sum_{j=1}^n S_{x,t} (G_{j}^t)$ is the total size of the $x$ largest items in the groups $G_{1}^t, \ldots, G_{n}^t$. A different number of items is discarded from each group, depending on the value of $S_{\eps |t|,t}$ (Step~\ref{step:discardCapacity}).\footnote{Assume for simplicity that $\eps |t| \in \mathbb{N}$.} Specifically,
 if $S_{\eps |t|+1,t} > \eps |t|$ then 
 discarding the $(\eps |t|+1)$ largest items from each small group might
 require using a large number of extra bins.
 Thus, the algorithm finds the first integer $x$, such that $\eps |t| \geq S_{x-1,t}$ and $\eps |t| < S_{x,t}$.
 In case $x = 1$ the 
 algorithm discards the largest item in each group starting from $G_{j}^t$, until the total size of discarded items exceeds $\eps |t|$.  If $x > 1$, we discard the $x$ largest items in each group $G_{j}^t, \ldots, G_{n}^t$. The difference between the two cases above is made to guarantee that we discard items of total size $\Theta(\eps)|t|$.

Assume now that $S_{\eps |t|+1,t} \leq \eps |t|$. We discard the largest $\eps |t|+1$ items of each group $G_{1}^t, \ldots, G_{n}^t$. 
}

We now describe the packing of the remaining items in $I_t$ in bins of type $t$. First, we add $2\eps |t|$ extra bins to $t$. The extra bins are empty and thus have capacity $1$; however, we assume that they have capacity $f(t) \leq 1$. This increases the overall number of bins in the solution by $2\eps OPT$. 
Consider the items in each group in non-increasing order by sizes. For each bin $b \in t$ in an arbitrary order, GreedyPack assigns to $b$ the largest remaining item in each group $G^t_1, \ldots, G^t_n$. If an overflow occurs, replace an item from some group $G_{j}^t$ by the next item in $G_{j}^t$. This is repeated until there is no overflow in $b$. W.l.o.g., we may assume that $|G_j|= OPT$ for all $1 \leq j \leq n$; thus, $b$ contains one item from each group (otherwise, we can add to $G_j$ dummy items of size $0$, with no increase to the number of bins in an optimal solution).

Recall that the large items from small groups are packed using the Swapping algorithm, that yields a feasible packing. Yet, it does not guarantee that the small items can be added without causing conflicts. Hence, GreedyPack may output a packing in which a small and large item from the same small group are packed in the same bin. Such conflicts are resolved by discarding the small item in each.

\begin{lemma}
\label{lem:5ecoupledItems}
The total size of items discarded in $GreedyPack$ in Step~\ref{step:discardsmalloverlarge} due to conflicts 
is at most $\eps OPT$, and at most $\eps^{k+2} \cdot OPT$
items are discarded from each group.
\end{lemma}

\begin{proof}

The number of items discarded from each group is at most $\eps^{k+2} \cdot OPT$, since all groups are small.
 Assume that the total size of these items is strictly larger than $\eps OPT$. Since each discarded item is {\em coupled} with a large conflicting item from the same group, whose size is at least $1/\eps$ times larger (recall that the medium items are discarded), this implies that the total size of large conflicting items is greater than $OPT$. Contradiction. \qed  
 
 \end{proof}
\begin{algorithm}[htb]
	\caption{$PackSmallItems(I_0,B)$}
	\label{Alg:PackSmallItems}
	\begin{algorithmic}[1]
		\For {each guess of $RecursiveEnum(I_0,B)$}
		\For {$b \in E_{\alpha}$}
		\State{evict from $b$ the largest item $\ell$ satisfying: $\ell$ is small, and less than $\eps OPT$ $~~~~~~~~~~~$ items where evicted from $G_i$, where $\ell \in G_i$.
			\label{step:eviction}}
		\EndFor
		\State{Add to $T$ a new type $t'$ consisting of $\eps OPT$ empty bins.\label{step:ExtraType}}
		\State{Compute a feasible partition of $I_{\alpha}$ into the types in $T$.\label{step:partition}}
		\For {$t \in T$}
		\State{Add $2\eps |t|$ extra bins to $t$.\label{step:discardCapacity}}
		\State{Assign $I_t$ to bins of type $t$ using $GreedyPack(I_t,t)$.\label{step:packGreedy}}
		\EndFor
		\EndFor
		
	\end{algorithmic}
\end{algorithm}


\comment{$T_E = \{t_{E,1}, \ldots, t_{E,e}\}$, $T_D = \{t_{D,1}, \ldots, t_{D,d}\}$ be all recursively generated sub-types that their bins are in $E_{\alpha}$ and $D_{\alpha}$, respectively. Denote by $f(t) = f^{\alpha}_b, b \in t$ for every type $t$. Assume that $f(t_{E,i}) \leq f(t_{E,i+1}), i \in [e-1]$ and similarly for $f(t_{D,i})$. We list all types in a specific order $\sigma$, as follows. First, $t_{E,i}$ is before $t_{E,i+1}$ for all $i \in [e-1]$, and similarly, $t_{D,i}$ is before $t_{D,i+1}$ for all $i \in [d-1]$. Second, $t_{E,i}$ is before $t_{D,j}$ if and only if $f(t_{E,i}) < \eps f(t_{D,j})$.   }

\comment{We now describe algorithm $GreedyPack$, which packs the remaining items besides a few more discarded items. Consider the order $\sigma$ of the types, and the items in each group in non-increasing order by sizes. For each bin $z \in t$ in this order, GreedyPack assigns to $z$ the largest remaining item in each group with two conditions: (i) if $E_{\alpha}$, then only $t$-negligible items are added to $z$. Otherwise, if $t \in D_{\alpha}$, only items $\ell \in G_i$ with $s_{\ell} \leq f^{\alpha_z}$ are added to $z$. (ii) item $\ell$ is added to $z$ if and only if $z$ does not contain an item from $G_i$.

If an overflow occurs, replace an item from arbitrary group $G_{i}$
by the next item in This group by the sorted order of the group. This process is repeated until there is no overflow. Note that, $z$ contains one item from each group.\footnote{Recall that $|G_j|= OPT$ $\forall~ 1 \leq j \leq n$.}
Let $items(t)$ be all items that belong to type $t$ besides large items from small groups (recall that these items are packed greedily and we cannot assume that they are packed in correspondence to an optimal solution). }

\begin{lemma}
\label{lem:greedyInType}
For any $t \in T$, given a parameter $0<\delta<\frac{1}{2}$ and a set of items $I_t$ such that (i) $|G_j^t| \leq |t|-|(I'_t \setminus I_L) \cap G_j|$; (ii) for all $\ell \in I_t: s_{\ell} \leq \delta f(t)$, and (iii) $S(I_t) \leq (1-\delta) f(t) |t|$, $GreedyPack$ finds a feasible packing of $I_t$ in bins of type $t$. 
\end{lemma}

\begin{lemma}
\label{lem:greedySummary}
Algorithm~\ref{Alg:PackSmallItems} assigns in Step~\ref{step:packGreedy} to $OPT$ bins all items except for $O(\eps)OPT$ items from each group, of total size $O(\eps)OPT$. 
\end{lemma}

\comment{For every $b \in B$, denote by $f^0_b$ the capacity of $b$ excluding items packed in $b$, i.e., $f^0_b = 1-\sum_{\ell \in b} s_{\ell}$. We say that item $\ell$ is {\em negligible in $b$} if $s_{\ell} \leq \eps f^0_b$, and $\ell$ is {\em non-negligible in $b$} otherwise. 

The packing of $I_s$ in $B$ is done iteratively, where each iteration consists of three steps: {\em partition phase}, {\em initial packing} and {\em permanent packing}. We give a short overview of these phases below, and elaborate in the following sections. 

In the partition phase, the items are partitioned into the bins such that each $b \in B$ is assigned with at most one negligible in $b$ item from each small group without overflow and with a small number of discarded items from each group. In addition, each bin $b$ is assigned with at most $\frac{1}{\eps}$ items, which are non-negligible in $b$, where conflicts and overflows are allowed for the latter items. An item $\ell$ is considered as {\em non-negligible} if $\ell$ is assigned to bin $b$ in the partition phase and $\ell$ is non-negligible in $b$.  

In the initial packing phase, a subset of non-negligible items is discarded, such that the remaining non-negligible items  }

\comment{
\mysubsection{Small Items from Small Groups}
\label{subsection:Small Items from Small Groups}
Up to this point we packed in $OPT$ bins the large items, as well as the medium and small items from large groups. We proceed to pack the small items from small groups. To this end,
we define a {\em profile} as a feasible packing of all items previously packed. An {\em almost optimal profile} is a profile for which there is an assignment of the small items from small groups into the bins adhering to two conditions: $(i)$ no overflow in any of the bins, and $(ii)$ the only conflicts allowed are between a small item and a large item from the same small group.

Algorithm $SmallGroups$ handles the packing of the small items of small groups; it combines guessing with several greedy-based components.
 Let $G_{i_1}^s, \ldots, G_{i_S}^s$ be the small items in the small groups, and let $V$ be the total size of these items, i.e., $V = \sum_{j = 1}^{S} \sum_{\ell \in G_{i_j}^s} s_{\ell}$.


Algorithm SmallGroups starts by sorting the small items in each small group in non-increasing order by sizes. 
Denote the items of $G_{i_j}^s$ in the sorted order by $s_{1,i_j}, s_{2,i_j},\ldots$.
Let $V_x (G_{i_j}^s)$ be the total size of the $x$ largest items in $G_{i_j}^s$ for an integer $x \geq 0$; that is, the total size of items of sizes $s_{1,i_j}, \ldots , s_{x,i_j}$.\footnote{An empty sum is equal to $0$.} 
Then, $V_x = \sum_{j=1}^S V_{x} (G_{i_j}^s)$ is the total size of the $x$ largest items in the groups $G_{i_1}^s, \ldots, G_{i_S}^s$. Let $m=\lfloor \eps OPT \rfloor$, then 
$m \geq 1$.
 A different number of items is discarded from each group, depending on the value of $V_{m+1}$. Specifically,
 if $V_{m+1} > \eps OPT$ then 
 discarding the $(m+1)$ largest items from each small group will
 require using a large number of extra bins.
 Thus, the algorithm finds the first integer $x$, such that $\eps OPT \geq V_{x-1}$ and $\eps OPT < V_x$.
 In case $x = 1$ the 
 algorithm discards the largest item in each group starting from $G_{i_1}^s$,
 until the total size of discarded items exceeds
  $\eps OPT$. If $x > 1$, we discard the $x$ largest items in each small group $G_{i_1}^s, \ldots, G_{i_S}^s$.

Assume now that $V_{m+1} \leq \eps OPT$. We discard the largest $m+1$ items of each small group $G_{i_1}^s, \ldots, G_{i_S}^s$. To enable an efficient packing of the remaining items, it is crucial to have the maximum size of an item sufficiently small. We obtain this property by handling separately all {\em special} groups $G_{i_j}^s$ for which $s_{m+2,i_j} > \frac{\varepsilon \cdot V}{OPT}$. Denote this set of groups by $A$.
 \negA
 \begin{lemma}
\label{lem:disgroups}
$|A| \leq \frac{1}{\eps^2}$
\end{lemma}
\negA
\negA
\begin{proof}
Consider a group $G_{i_j}^s \in A$. Since $G_{i_j}^s$ is {\em special}, it holds that $s_{m+2,i_j} > \frac{\varepsilon \cdot V}{OPT}$.
	Also, as $s_{1,i_j} \geq s_{2,i_j}\geq \ldots \geq s_{m+2,i_j}$, the total size of $G_{i_j}^s$ is larger than
	$(m+2) \cdot \frac{\eps V}{OPT} = (\lfloor \eps OPT \rfloor +2) \frac{\eps V}{OPT} > \eps^{2} V$. 
	Since the total size of the small items in small groups is equal to $V$
	we have the statement of the lemma. \qed
\end{proof}
\negA
 Hence, algorithm
 EnumGroups can be used to obtain in polynomial time a packing of $A$ which corresponds to an optimal solution (see Theorem~\ref{lem:larggroupComplexity}).

We now describe algorithm $GreedyPack$, which packs the remaining items.
Consider the $OPT$ bins in non-increasing order by the total size of packed items, and the items in each group in non-increasing order by sizes. For each bin $z$ in this order, GreedyPack assigns to $z$ the largest remaining item in each group. If an overflow occurs, replace an item from some group $G_{i_j}^s$
by the next item in 
$G_{i_j}^s$.
This is repeated until there is no overflow. Note that, given an almost optimal profile, we are guaranteed to find a packing of $z$ with no overflow. Moreover, $z$ contains one item from each group.\footnote{Recall that $|G_j|= OPT$ $\forall~ 1 \leq j \leq n$.}

By definition, an almost optimal profile does not guarantee that large items from small groups are packed optimally. 
Hence, GreedyPack may output a packing in which a small and large item from the same small group are packed in the same bin. Such conflicts are resolved by discarding the small item. 
\negA
\begin{lemma}
\label{lem:5ecoupledItems}
The total size of items discarded in $GreedyPack$ 
is at most $\eps OPT$, and at most $\eps^{k+2} \cdot OPT$
items are discarded from each group.
\end{lemma}
\begin{proof}
The number of items discarded from each group is at most $\eps^{k+2} \cdot OPT$, since all groups are small.
 Assume that the total size of these items is strictly larger than $\eps OPT$. Since each discarded item is {\em coupled} with a large conflicting item from the same group, whose size is at least $1/\eps$ times larger (recall that the medium items are discarded), this implies that the total size of large conflicting items is greater than $OPT$. Contradiction. \qed    
\end{proof}
\negA
The pseudocodes of GreedyPack and SmallGroups are given in Algorithm~\ref{Alg:greedyPack} and Algorithm~\ref{Alg:smallGroups}, respectively.\footnote{For the pseudocode of GreedyPack see Appendix~\ref{sec:algorithmGreedyPack}.} 

\begin{theorem}
\label{lem:smallgroup}
 Given an almost optimal profile, Algorithm~\ref{Alg:smallGroups} packs feasibly in polynomial time all the small items from small groups.
\end{theorem}

We prove the theorem using the following lemmas. Given a packing of an instance $I$ of GBP in $OPT$ bins, the {\em free space} is given by $OPT- S(I)$. The {\em free space on average} is $\mu = \frac{OPT-S(I)}{OPT}$.

\negA

 \begin{algorithm}[h]
	\caption{$SmallGroups(G_{i_1}^s, \ldots, G_{i_S}^s, b_1, \ldots, b_{OPT})$}
	\label{Alg:smallGroups}
	\begin{algorithmic}[1]
	\If {$V_{m+1} > \eps OPT$}
	\State{Find the minimal $x$ such that $V_x > \eps OPT$.}
    \If {$x = 1$}   
    \State{$SUM \leftarrow 0$}
    \State{$j \leftarrow 1$}
    \While{$SUM < \eps OPT$}
      \State{Discard $s_{1,i_j}$. \label{step:discardx=1} }
      \State{$SUM \leftarrow SUM + s_{1,i_j}$}
      \State{$j \leftarrow j+1$}
\EndWhile
\Else
\State{discard the largest $x$ items of each small group $G_{i_1}^s, \ldots, G_{i_S}^s$ \label{step:discard1}}
\EndIf
\If {$GreedyPack(\{G_{i_1}^s, \ldots, G_{i_S}^s\}, b_1, \ldots, b_{OPT}) = failure$ \label{step:greedystep2}}
     \State{return $failure$} 
	\EndIf
\Else
	\State{Discard the largest $m+1$ items of each small group $G_{i_1}^s, \ldots, G_{i_S}^s$. \label{step:discardregular}}
	\State{A small group $G_{i_j}^{s}$ is considered as {\em special} if $s_{m+2,i_j} > \frac{\varepsilon \cdot V}{OPT}$.    \label{step:discardA}}
	\State{Let $A \subseteq \{G_{i_1}^s, \ldots, G_{i_S}^s\}$ be the set of  special groups.}
	\State{For each $G_{i_j}^s \in A$ apply shifting to items in $G_{i_j}^s$ with parameter $Q = \lfloor \eps^{3} OPT \rfloor$.\label{step:linearGrouping}}
	\State{Compute bin types $B_1,...,B_t$. \label{step:largeitem}}
	\For{$\hat{b}_1, \ldots, \hat{b}_{OPT}$, result of $EnumGroups(B_1,\ldots,B_t$,$A$,$b_1,\ldots,b_{OPT}$) \label{step:largegroups}}
    \If {$GreedyPack(\{G_{i_1}^s, \ldots, G_{i_S}^s\} \setminus A, \hat{b}_1, \ldots, \hat{b}_{OPT}) = failure$ \label{step:greedystep1}}
    \State{Return $failure$.}
    \EndIf
    \EndFor
    \EndIf
	\end{algorithmic}
\end{algorithm}




\begin{lemma}
\label{lem:helpgreedy}
Given $OPT$ bins and an almost optimal profile, free space on average at least $\mu \geq 0$ and the largest remaining item of size at most $\mu$,
GreedyPack
yields a feasible packing of the remaining small items from small groups.
\end{lemma}
\negA
\begin{proof}
We prove the claim by induction on $OPT$. For the base case, let $OPT = 1$. Since there is only one bin in an optimal solution, each group consists of a single item.
Also, if the free space is non-negative, the total size of all items is at most the capacity of the bin (since we have an almost optimal profile). Hence, we can pack all items feasibly.

For the induction step, assume the claim holds for $OPT-1$ bins. Now, suppose that 
there are $OPT$ bins, partially packed with an almost optimal profile. GreedyPack initially assigns items to the bin of maximum total size.
Now, consider two cases.

$(i)$ After packing the first bin, the total size of items in this bin is less than $1-\mu$. Hence, the first bin contains the largest item left in each small group.
As the bins are sorted in non-decreasing order by free capacity, and the items in each group are sorted in non-increasing order by sizes, we can feasibly pack all items in the remaining $OPT-1$ bins.

$(ii)$ The first bin is packed with total size at least $1-{\mu}$. 
We distinguish between two sub-cases. $(a)$ Prior to applying GreedyPack, the first bin  had free space less than $\mu$. Then, clearly, the total free space in the remaining $OPT-1$ bins is at least $(OPT-1) \mu$.
$(b)$ Otherwise, the initial free space in the first bin is at least $\mu$. As the bins are sorted in non-decreasing order by free space, the total free space in the other bins is at least $(OPT-1) \mu$.

Also, there are $OPT -1$ items remaining in each group. Hence, by the induction hypothesis, the remaining items can be packed in bins $2, \ldots, OPT$.
\qed  
\end{proof}

\negA

\begin{lemma}
\label{lem:smallO1groups}
Given an almost optimal profile, if $V_{m+1} \leq \eps OPT$ then GreedyPack outputs (in Step~\ref{step:greedystep1} of Algorithms~\ref{Alg:smallGroups}) a feasible packing of all small items of small groups.
\end{lemma}

\begin{proof}
	In Step~\ref{step:discardregular} algorithm SmallGroups discards 
	the $m+1$ largest small items of each small group, whose total size 
	is at least $\eps V$. Thus, the free space on average is at least $\mu \geq \frac{\eps V}{OPT}$. In addition, the maximum size of any remaining item is at most $\frac{V\cdot\varepsilon}{OPT}$ (since we use algorithm EnumGroups for special groups). There are at most $OPT$ items in each group.
	Hence, given an almost optimal profile, by Lemma~\ref{lem:helpgreedy}, all the remaining small items can be packed feasibly. \qed    
\end{proof}
\negA
The proof of the next lemma is similar to the proof of Lemma~\ref{lem:smallO1groups}. 
\negA
\begin{lemma}
\label{lem:additionalsmallitems}
Given an almost optimal profile, if $V_{m+1} > \eps OPT$, GreedyPack outputs (in Step~\ref{step:greedystep2} of Algorithms~\ref{Alg:smallGroups}) a feasible packing of all small items of small groups.
\end{lemma}

 \negA

\noindent{\bf Proof of Theorem~\ref{lem:smallgroup}:}
\label{proof:theorem10}
By Lemmas~\ref{lem:smallO1groups} and~\ref{lem:additionalsmallitems}, for any value of $V_{m+1}$ all small items are packed feasibly, given an almost optimal profile. Since our scheme enumerates over all possible slot patterns, it is guaranteed to reach an almost optimal profile. For the running time, 
we first note that, by Theorem~\ref{lem:larggroupComplexity} and Lemma~\ref{lem:disgroups}, the time complexity of EnumGroups (in Step~\ref{step:largegroups} of SmallGroups) is polynomial. Also, algorithm GreedyPack (called in Steps~\ref{step:greedystep2} and~\ref{step:greedystep1} of SmallGroups) runs in time
$O(N \cdot OPT)$. 
Finally, algorithm SmallGroups packs the small items with no violation of bin capacities while avoiding group conflicts. Thus, the resulting packing is feasible. \qed   
  }
\subsection{Putting it all Together}
\label{sec:putting it all together}
It remains to
show that the items discarded throughout the execution of the scheme can be packed in a small number of extra bins.


\begin{lemma}
\label{lem:discardFinal}
The medium items from small groups and all discarded items can be packed in $O(\eps) \cdot OPT$ extra bins.
\end{lemma}

Algorithm~\ref{Alg:summarize} summarizes the steps of our scheme (see Appendix~\ref{sec:Approximation Scheme Pseudocode}).
\begin{theorem}
\label{theorem:APTAS}
There is an APTAS for the group bin packing problem.
\end{theorem}

\comment{

\negA
 \negA

\mysection{Bin Packing with Interval-graph Conflicts}
\label{sec:interval_graphs}

Consider instances of BPC with {\em interval-graph conflicts (IC)}. We show that the scheme in Section~\ref{sec:APTAS} can be extended to yield APTASs
for two non-trivial subclasses of such conflict graphs, namely, {\em proper}
and {\em constant length} interval graphs. 
Applying our scheme for such instances requires 
refinements of the definitions of slots and patterns (for packing the large items). Recall that for GBP each slot is associated with a
group. For IC the items may belong to multiple groups (i.e., maximal cliques in the interval graph); thus, a slot may be associated with a subset of groups.

For the small items, we generalize our greedy-based approach. A key 
idea in GreedyPack is to assign to the current bin a maximum total size while taking an item from each group. However, GreedyPack does not necessarily work on instances where items belong to multiple groups. Hence, we use dynamic programming to find for the current bin (in the course of GreedyPack) a feasible packing, taking an item from each small group (one item may cover multiple groups), such that the total packed size is maximized. The proofs of the next theorems and the complete details are given in Appendix~\ref{An APTAS for Proper Interval Conflict Graph}.

\begin{theorem}
\label{thm:proper}
Given an instance $I$ of bin packing with proper interval conflict graph and a fixed $\eps \in (0,1)$, there is a polynomial-time algorithm that packs $I$ in at most
$(1 +\eps) OPT(I) +1$ bins.
\end{theorem}

\begin{theorem}
\label{thm:ConstantLength}
Given an instance $I$ of bin packing with constant length interval conflict graph and a fixed $\eps \in (0,1)$, there is a polynomial-time algorithm that packs $I$ in at most
$(1 +\eps) OPT(I) +1$ bins.
\end{theorem}

}
\bibliographystyle{splncs04}
\bibliography{bibfile}

\begin{thebibliography}{10}
\providecommand{\url}[1]{\texttt{#1}}
\providecommand{\urlprefix}{URL }
\providecommand{\doi}[1]{https://doi.org/#1}

\bibitem{allornothing}
Adany, R., Feldman, M., Haramaty, E., Khandekar, R., Schieber, B., Schwartz,
  R., Shachnai, H., Tamir, T.: All-or-nothing generalized assignment with
  application to scheduling advertising campaigns. ACM Transactions on
  Algorithms (TALG)  \textbf{12}(3),  1--25 (2016)

\bibitem{alon1998approximation}
Alon, N., Azar, Y., Woeginger, G.J., Yadid, T.: Approximation schemes for
  scheduling on parallel machines. Journal of Scheduling  \textbf{1}(1),
  55--66 (1998)

\bibitem{boinc04}
Anderson, D.P.: Boinc: A system for public-resource computing and storage. In:
  Fifth IEEE/ACM international workshop on grid computing. pp. 4--10. IEEE
  (2004)

\bibitem{JOUR}
Anderson, D.P.: {BOINC}: A platform for volunteer computing. Journal of Grid
  Computing  \textbf{18},  99 -- 122 (2017)

\bibitem{C+17}
Christensen, H.I., Khan, A., Pokutta, S., Tetali, P.: Approximation and online
  algorithms for multidimensional bin packing: A survey. Computer Science
  Review  \textbf{24},  63--79 (2017)

\bibitem{C+13}
Coffman, E.G., Csirik, J., Galambos, G., Martello, S., Vigo, D.: Bin packing
  approximation algorithms: survey and classification. In: Handbook of
  combinatorial optimization, pp. 455--531 (2013)

\bibitem{DW17}
Das, S., Wiese, A.: On minimizing the makespan when some jobs cannot be
  assigned on the same machine. In: 25th Annual European Symposium on
  Algorithms, {ESA}. pp. 31:1--31:14 (2017)

\bibitem{ennajjar2017securing}
Ennajjar, I., Tabii, Y., Benkaddour, A.: Securing data in cloud computing by
  classification. In: Proceedings of the 2nd international Conference on Big
  Data, Cloud and Applications. pp.~1--5 (2017)

\bibitem{epstein2008bin}
Epstein, L., Levin, A.: On bin packing with conflicts. SIAM Journal on
  Optimization  \textbf{19}(3),  1270--1298 (2008)

\bibitem{Vega1981BinPC}
{Fernandez de la Vega}, W., Lueker, G.S.: Bin packing can be solved within 1 +
  $\varepsilon$ in linear time. Combinatorica  \textbf{1},  349--355 (1981)

\bibitem{Jansen_et_al:2019}
Grage, K., Jansen, K., Klein, K.M.: An {EPTAS} for machine scheduling with
  bag-constraints. In: The 31st ACM Symposium on Parallelism in Algorithms and
  Architectures. pp. 135--144 (2019)

\bibitem{guerine2019provenance}
Guerine, M., Stockinger, M.B., Rosseti, I., Simonetti, L.G., Oca{\~n}a, K.A.,
  Plastino, A., de~Oliveira, D.: A provenance-based heuristic for preserving
  results confidentiality in cloud-based scientific workflows. Future
  Generation Computer Systems  \textbf{97},  697--713 (2019)

\bibitem{hochba1997approximation}
Hochbaum, D.S. (ed.): Approximation Algorithms for {NP}-Hard Problems. PWS
  Publishing Co., USA (1996)

\bibitem{hochbaum1987using}
Hochbaum, D.S., Shmoys, D.B.: Using dual approximation algorithms for
  scheduling problems theoretical and practical results. Journal of the ACM
  \textbf{34}(1),  144--162 (1987)

\bibitem{HK56}
Hoffman, A.J., Kruskal, J.B.: Integral boundary points of convex polyhedra. In:
  Linear Inequalities and Related Systems.(AM-38), Volume 38, pp. 223--246.
  Princeton University Press (1956)

\bibitem{jansen1999approximation}
Jansen, K.: An approximation scheme for bin packing with conflicts. Journal of
  combinatorial optimization  \textbf{3}(4),  363--377 (1999)

\bibitem{MS:5}
Jansen, K.: An {EPTAS} for scheduling jobs on uniform processors: using an
  {MILP} relaxation with a constant number of integral variables. SIAM Journal
  on Discrete Mathematics  \textbf{24}(2),  457--485 (2010)

\bibitem{JKV16}
Jansen, K., Klein, K., Verschae, J.: Closing the gap for makespan scheduling
  via sparsification techniques. In: 43rd International Colloquium on Automata,
  Languages, and Programming (ICALP),. pp. 72:1--72:13 (2016)

\bibitem{JO97}
Jansen, K., {\"{O}}hring, S.R.: Approximation algorithms for time constrained
  scheduling. Inf. Comput.  \textbf{132}(2),  85--108 (1997)

\bibitem{karmarkar1982efficient}
Karmarkar, N., Karp, R.M.: An efficient approximation scheme for the
  one-dimensional bin-packing problem. In: 23rd Annual Symposium on Foundations
  of Computer Science. pp. 312--320. IEEE (1982)

\bibitem{MS:10}
Leung, J.Y.: Bin packing with restricted piece sizes. Information Processing
  Letters  \textbf{31}(3),  145--149 (1989)

\bibitem{EAFR}
{Lin}, Y., {Shen}, H.: Eafr: An energy-efficient adaptive file replication
  system in data-intensive clusters. IEEE Transactions on Parallel and
  Distributed Systems  \textbf{28}(4),  1017--1030 (2017)

\bibitem{mccloskey2005approaches}
McCloskey, B., Shankar, A.: Approaches to bin packing with clique-graph
  conflicts. Computer Science Division, University of California (2005)

\bibitem{OS95}
Oh, Y., Son, S.: On a constrained bin-packing problem. Technical Report
  CS-95-14  (1995)

\bibitem{Ro13}
Rothvo{\ss}, T.: Approximating bin packing within {O}(log {OPT} * log log
  {OPT)} bins. In: 54th Annual {IEEE} Symposium on Foundations of Computer
  Science. pp. 20--29. {IEEE} Computer Society (2013)

\bibitem{S94}
Simchi-Levi, D.: New worst-case results for the bin-packing problem. Naval
  Research Logistics (NRL)  \textbf{41}(4),  579--585 (1994)

\bibitem{vazirani}
Vazirani, V.V.: Approximation Algorithms. Springer-Verlag, Berlin, Heidelberg
  (2001)

\bibitem{Z07}
Zuckerman, D.: Linear degree extractors and the inapproximability of max clique
  and chromatic number. Theory of Computing  \textbf{3}(1),  103--128 (2007)

\end{thebibliography}

\appendix

\section{Applications of Group Bin Packing}
\label{sec:applications}
\subsection{Storing File Replicas}
Different versions (or replicas) of critical data files are distributed to servers around the network \cite{EAFR}. Each server has its storage capacity and can thus be viewed as a {\em bin}. Each data file is an {\em item}. The set of replicas of each data file forms a {\em group}.
To ensure better fault tolerance, 
 replicas of the same data file must be stored on different servers. The problem of storing a given set of file replicas on a minimal number of servers in the network can be cast as an instance of GBP.

\subsection{Security in Cloud Computing}
Computational projects of large data scale, such as scientific experiments or simulations, often rely on cloud computing. Commonly, the project data is also stored in the cloud. In this setting, a main concern is that a malicious entity might gain access to confidential data~\cite{ennajjar2017securing}. To strengthen security, data is dispersed among multiple cloud storage 
services~\cite{guerine2019provenance}. Projects are fragmented into critical tasks, so that no single task can reveal substantial information about the entire project.
 Then, each task is stored on a different storage service.
 Viewing a cloud storage service as a {\em bin} and each project as a {\em group} containing a collection of critical tasks ({\em items}), the problem of storing a set of projects on a minimal number of (identical) storage services yields an instance of GBP.

\comment{\subsubsection{Cloud Scheduling}
Often, jobs submitted for processing on cloud servers are deadline-sensitive; that is, such jobs need to be completed by a given deadline (see, e.g.,~\cite{sensitiveJobs}). 
Due to efficiency considerations, some jobs cannot be processed on the same server (for example, to better balance the load, any two CPU intensive jobs are assigned to different servers~\cite{jansen1999approximation}). To enable the completion of all jobs in a batch by their deadlines, more servers can be added. Viewing each server as a 
{\em bin} and the jobs as {\em items}, can be seen as the problem of 
{\em bin packing with conflicts}. In the special case where the input consists of groups of conflicting jobs,
it can be considered as an instance of GBP.}

\subsection{Signal Distribution}
Volunteer computing allows researchers and organizations to harvest computing capacity from volunteers, e.g.,
donors among the general public. 
The principal framework for volunteer
computing is the Berkeley Open Infrastructure for Network Computing, popularly known as
BOINC~\cite{boinc04,JOUR}. 
Such distributed systems must dispense work items to clients. The clients can be viewed as bins containing work
items. Each item requires some amount of processing time. A client will contribute only a fixed
number of processor cycles per day. Assume that work items that are correlated (such as signals from
the same region of the sky) can be verified against each other. To avoid tampering, 
signals are distributed so that no client processes more than one signal from the same region. Viewing signals from the same region as {\em groups}, we have an instance of GBP.

\section{Group Bin Packing vs. BP} 
\label{sec:gbp_bp_example}

A common approach in developing asymptotic approximation schemes for the bin packing problem is to distinguish between 
{\em large} and {\em small} items. Initially, the small items are {\em discarded} from the instance, and an (almost) optimal packing is obtained for the large items. The small items are then added in the remaining free space, with possible use of a small number of extra bins (see, e.g.,~\cite{Vega1981BinPC} and the comprehensive survey in~\cite{vazirani}). Unfortunately, when handling a GBP instance, this approach may not lead to an APTAS. Indeed, it may be the case that all of the small items belong to a {\em single} group. Thus, 
a large number of extra bins may be required to accommodate small items which cannot be added to previously packed bins. We give a detailed example below.

\begin{figure}[h!]
\centering
\subcaptionbox{}
{\includegraphics[width=6.0cm]{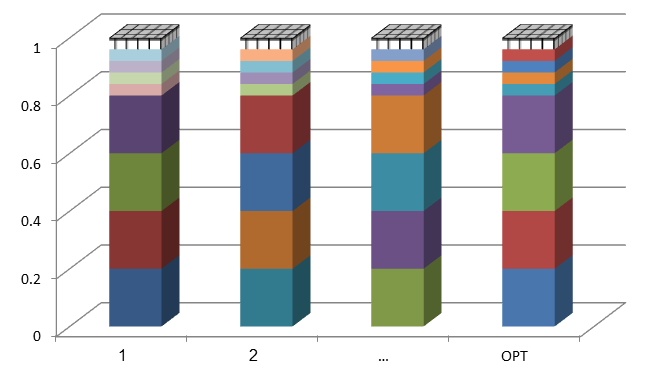}}
\subcaptionbox{}
{\includegraphics[width=6.0cm]{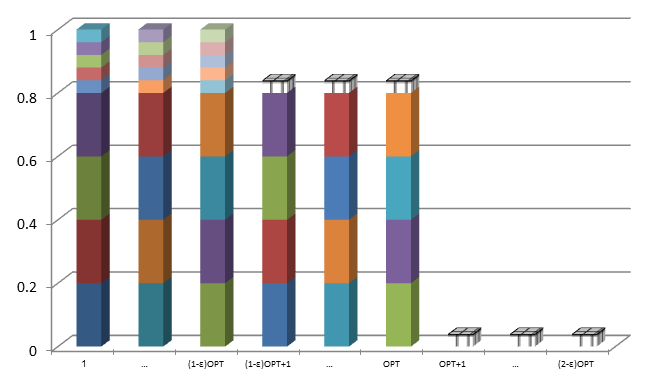}}
\caption{(a) An optimal packing. (b) A packing in $(2-\varepsilon) OPT$ bins.}
\end{figure}

Figure $1$ $(a)$ and $(b)$ illustrate two different packings of a given GBP instance $I$ with $N$ items and $n$ groups, where $I = \{G_1, \ldots, G_n\}$. The content of each group is given as a multiset, where each item is represented by a number (its size) in the range $(0,1]$. The instance consists of $n_1$ groups of a single item of size $\frac{1}{5}$: $G_1 = \{\frac{1}{5}\}, G_2 = \{\frac{1}{5}\}, \ldots, G_{n_1} = \{\frac{1}{5}\} \: $; $n_2$ groups with a single item of size $\frac{\varepsilon}{5}$: $G_{n_1+1} = \{\frac{\varepsilon}{5}\}, \ldots, G_{n_1+n_2} = \{\frac{\varepsilon}{5}\}$, and one group with $\hat{N}$ items of size $\frac{\varepsilon}{5}$; that is,
$G_n = \{\frac{\varepsilon}{5}, \ldots,\frac{\varepsilon}{5} \}$, where $n=n_1+n_2+1$ and $|G_n| = \hat{N}$. 
In addition, $n_1 = 4 {\hat N}$ and $n_2 = {\hat N} \cdot (\frac{1}{\eps} -1)$.\footnote{For simplicity, assume that $1/\eps$ and $\eps {\hat N}$ are integers.}
The checkered boxes (the upper rectangles on each column) represent items in $G_n$. All other groups contain a single item. Each of these other groups is represented by a box of a different color, of size corresponding to the size of the single item in the group, which can be either $\frac{1}{5}$ or $\frac{\varepsilon}{5}$. Each column represents a bin of unit capacity. 

Figure $1 (a)$ shows a packing where each bin contains $4$ items of size $\frac{1}{5}$ and $\frac{1}{\eps}$ items of size $\frac{\eps}{5}$; among these items, exactly one item belongs to $G_n$. Thus, the total number of bins is $\hat N$. Since each bin is full, this packing is optimal, i.e., $OPT={\hat N}$.

Now, suppose that, initially, all items of sizes larger than $\delta$, for some $\delta \in (0, \frac{1}{5})$, are packed optimally. The small items are then added in the free space in a greedy manner. Specifically, starting from the first bin, small items are added until the bin is full, or until it contains an item from each group. We then proceed to the next bin. Figure $1 (b)$ shows a packing in which the large items (each of size $\frac{1}{5}$) are packed first optimally, $4$ items in each bin, using $\frac{n_1}{4}={\hat N}$ bins. Then, the items in $G_{n_1+1}, \ldots , G_{n_1+n_2}$ are added greedily so that the first $(1 -\eps){\hat N}$ bins are full. Then, the first $\eps {\hat N}$
items in $G_n$ are packed in the remaining bins and a set of $(1-\eps){\hat N}$ bins is added for the remaining items. Overall, the number of bins used is 
$(2-\eps){\hat N}=(2-\eps)OPT$.

\section{Balanced Coloring}
\label{sec:balanced_col}
Our scheme uses as a subroutine an algorithm proposed in~\cite{allornothing} for the {\em group packing} problem. For completeness, we include an outline of the algorithm adapted to handle instances of our problem. Given a GBP instance $I$ with a set of groups $G_1, \ldots, G_n$,
denote by $S(I)$ the total size of items in $I$, i.e., $S(I) = \sum_{\ell \in [N]} s_\ell$, where $[N] = \{1,\ldots,N\}$. Let $v_j$ be the number of items in group $G_j$. 
Consider the following {\em balanced coloring} of the groups. Color the items of $G_1$ in arbitrary order, using  $v_1$  colors (so that each item is assigned a distinct color). Now, sort the items in $G_2$  in non-increasing order
by size. Scanning the sorted list, add the next item in $G_2$ to the color class of
minimum total size which does not contain an item in $G_2$. We handle similarly the items in $G_3, \ldots,
G_n$. Then, each color class can be packed, using First-Fit, as a bin packing instance, i.e., with no group constraints. Algorithm~\ref{Alg:balancedColoring} is the pseudocode of  BalancedColoring. 

\begin{algorithm}[h]
	\caption{$BalancedColoring(G_1,\ldots,G_n)$}
	\label{Alg:balancedColoring}
	\begin{algorithmic}[1]
	\State{Let $v_j=|G_j|$ be the cardinality of $G_j$ and $v_{max}= \max_{1 \leq j \leq n} v_j$ be the maximum cardinality of any group.} 	
	\State{Partition arbitrarily the items of $G_1$ into $v_{1}$ color classes, such that each item is assigned a distinct color.}
	\State{Add $v_{max}-v_{1}$ empty color classes.}
	\For{$j = 2, \ldots, n$}
	\State{Sort $G_j$ in a non-increasing order by item sizes.}
	\State{Let ${i_1}, \ldots, {i_{v_j}}$ be the items in $G_j$ in the sorted order. }
	\For{$\ell = 1, \ldots, v_j$}
	\State{Add $i_{\ell}$ to a color class of minimum total size with no items from $G_j$.}
	\EndFor
	\EndFor
	\For{$k = 1, \ldots, v_{max}$}
	\State{Pack the items in color class $c_k$ in a new set of bins using First-Fit.}
	\EndFor
	\end{algorithmic}
\end{algorithm}


\section{Approximation Scheme for GBP}
\label{sec:Approximation Scheme Pseudocode}
Algorithm~\ref{Alg:summarize} gives the pseudocode of our APTAS for GBP. 

\begin{algorithm}[h]
	\caption{$Approximation Scheme(I, G_1,\ldots,G_n, \eps)$}
	\label{Alg:summarize}
	\begin{algorithmic}[1]
	\raggedright
		\State{Guess $OPT$.}
		\State{Let $b_1,\ldots,b_{OPT}$ be $OPT$ empty bins.}
		\State{Find an integer $k \in \{1,...,\frac{1}{\varepsilon^2}\}$ such that $\sum_{\ell \in I : \: s_{\ell} \in
				[\varepsilon^{k+1},\varepsilon^{k})} \: s_{\ell} \leq \varepsilon^{2} \cdot OPT$.}
		\State{An item $\ell$ is small if $s_{\ell} < \varepsilon^{k+1}$, medium if $s_{\ell} \in [\varepsilon^{k+1},\varepsilon^{k})$ and large otherwise.}
		\State{Define a group as large if it contains at least $\eps^{k+2} OPT$ large and medium items, and as small otherwise.}
		\For{large groups $G_i$}
		\State{\parbox[t]{\dimexpr\textwidth-\leftmargin-\labelsep-\labelwidth}{Apply linear shifting to medium and large items from $G_i$ with parameter $Q = \lfloor \eps^{2k+4} \cdot OPT \rfloor$.}}
		\EndFor
		\State{Apply linear shifting jointly to all large items from small groups with parameter $Q = \lfloor 2\eps \cdot OPT \rfloor$ and discard the last size class.}
		\For{each guess $\zeta$ of slot patterns for the bins}	
		\State{Discard medium items from small groups.}
		\State{Pack large and medium items by $Swapping(\zeta, G_1, \ldots, G_n)$. }
		\State{pack the remaining items using $PackSmallItems(I_0,B)$}
		\EndFor	
	\end{algorithmic}
\end{algorithm}

\comment{
\section{Balanced Coloring}\label{BalancedColoring}
\label{appendix:BalancedColoring}

Algorithm~\ref{Alg:balancedColoring} is the pseudocode of algorithm BalancedColoring. 

\begin{algorithm}[h]
	\caption{$BalancedColoring(G_1,\ldots,G_n)$}
	\label{Alg:balancedColoring}
	\begin{algorithmic}[1]
	\State{Let $v_j=|G_j|$ be the cardinality of $G_j$ and $v_{max}= \max_{1 \leq j \leq n} v_j$ be the maximum cardinality of any group.} 	
	\State{Partition arbitrarily the items of $G_1$ into $v_{1}$ color classes, such that each item is assigned a distinct color.}
	\State{Add $v_{max}-v_{1}$ empty color classes.}
	\For{$j = 2, \ldots, n$}
	\State{Sort $G_j$ in a non-increasing order by item sizes.}
	\State{Let ${i_1}, \ldots, {i_{v_j}}$ be the items in $G_j$ in the sorted order. }
	\For{$\ell = 1, \ldots, v_j$}
	\State{Add $i_{\ell}$ to a color class of minimum total size with no items from $G_j$.}
	\EndFor
	\EndFor
	\For{$k = 1, \ldots, v_{max}$}
	\State{Pack the items in color class $c_k$ in a new set of bins using First-Fit.}
	\EndFor
	
	\State{Pack all discarded items by $BalancedColoring$.\label{step:balanced}}
	\end{algorithmic}
\end{algorithm}




\section{An example: Group Bin Packing vs. BP}\label{Group Bin Packing vs. BP}
\label{appendix:Group Bin Packing vs. BP}

\begin{figure}[h!]
\centering
\subcaptionbox{}
{\includegraphics[width=6.0cm]{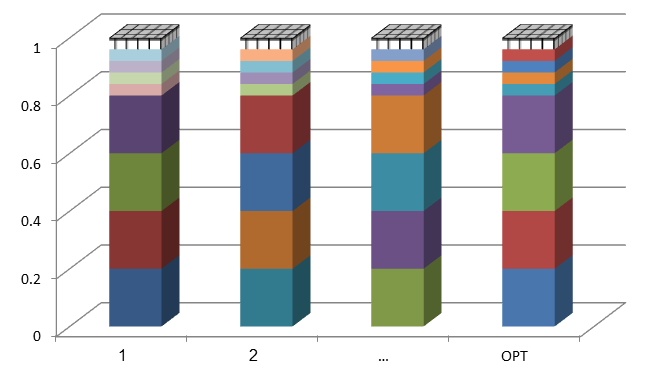}}
\subcaptionbox{}
{\includegraphics[width=6.0cm]{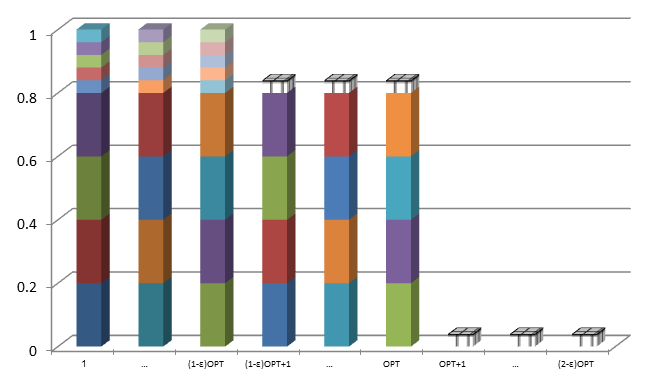}}
\caption{(a) An optimal packing. (b) A packing in $(2-\varepsilon) OPT$ bins.}
\end{figure}

Figure $1$ $(a)$ and $(b)$ illustrate two different packings of a given GBP instance $I$ with $N$ items and $n$ groups, where $I = \{G_1, \ldots, G_n\}$. The content of each group is given as a multiset, where each item is represented by a number (its size) in the range $(0,1]$. The instance consists of $n_1$ groups of a single item of size $\frac{1}{5}$: $G_1 = \{\frac{1}{5}\}, G_2 = \{\frac{1}{5}\}, \ldots, G_{n_1} = \{\frac{1}{5}\} \: $; $n_2$ groups with a single item of size $\frac{\varepsilon}{5}$: $G_{n_1+1} = \{\frac{\varepsilon}{5}\}, \ldots, G_{n_1+n_2} = \{\frac{\varepsilon}{5}\}$, and one group with $\hat{N}$ items of size $\frac{\varepsilon}{5}$; that is,
$G_n = \{\frac{\varepsilon}{5}, \ldots,\frac{\varepsilon}{5} \}$, where $n=n_1+n_2+1$ and $|G_n| = \hat{N}$. 
In addition, $n_1 = 4 {\hat N}$ and $n_2 = {\hat N} \cdot (\frac{1}{\eps} -1)$.\footnote{For simplicity, assume that $1/\eps$ and $\eps {\hat N}$ are integers.}
The checkered boxes (the upper rectangles on each column) represent items in $G_n$. All other groups contain a single item. Each of these other groups is represented by a box of a different color, of size corresponding to the size of the single item in the group, which can be either $\frac{1}{5}$ or $\frac{\varepsilon}{5}$. Each column represents a bin of unit capacity. 

Figure $1 (a)$ shows a packing where each bin contains $4$ items of size $\frac{1}{5}$ and $\frac{1}{\eps}$ items of size $\frac{\eps}{5}$; among these items, exactly one item belongs to $G_n$. Thus, the total number of bins is $\hat N$. Since each bin is full, this packing is optimal, i.e., $OPT={\hat N}$.

Now, suppose that, initially, all items of sizes larger than $\delta$, for some $\delta \in (0, \frac{1}{5})$, are packed optimally. The small items are then added in the free space in a greedy manner. Specifically, starting from the first bin, small items are added until the bin is full, or until it contains an item from each group. We then proceed to the next bin. Figure $1 (b)$ shows a packing in which the large items (each of size $\frac{1}{5}$) are packed first optimally, $4$ items in each bin, using $\frac{n_1}{4}={\hat N}$ bins. Then, the items in $G_{n_1+1}, \ldots , G_{n_1+n_2}$ are added greedily so that the first $(1 -\eps){\hat N}$ bins are full. Then, the first $\eps {\hat N}$
items in $G_n$ are packed in the remaining bins and a set of $(1-\eps){\hat N}$ bins is added for the remaining items. Overall, the number of bins used is 
$(2-\eps){\hat N}=(2-\eps)OPT$.

\comment{
In (a) there is a packing of $I$ such that there are 4 items of size $\frac{1}{5}$ and $\frac{1}{\eps}$ items of size $\frac{\varepsilon}{5}$ in each bin such that exactly one of them is of group $G_n$. In each bin there is one item of group $G_n$, thus there are $\hat{N}$ bins in the packing. Since each bin is packed using all of its capacity, we conclude that (a) is an optimal packing, hence $OPT = \hat{N}$. (b) is a packing of $I$ by an algorithm that packs optimally all items larger than $\varepsilon$ (assuming $\varepsilon < \frac{1}{5}$), and all items with sizes smaller or equal to $\varepsilon$ are packed greedily by the following way. The algorithm tries to pack each bin with arbitrary remaining items, until it cannot add anymore items, due to the capacity constraint of the bin or that the bin already has an item from all of the groups; then, the algorithm continues to the next bin and so on. The packing in (b) is resulted by packing the items of group $G_n$ last, such that the large items are packed optimally, the same way as in (a), and additionally, the first $(1-\varepsilon)\hat{N}$ bins are packed greedily, to the fullest of their capacity with items of size $\frac{\eps}{5}$ on top of the larger items. Then, the algorithm packs items of group $G_n$, however after packing $\varepsilon \hat{N}$ bins with an item from $G_n$, the algorithm needs to use a new set of $(1-\varepsilon)\hat{N}$ bins, in order to avoid conflicts of two items from group $G_n$ in one bin. Overall, in (b) the packing uses $(2-\varepsilon) \hat{N}$ bins, that is $(2-\varepsilon) OPT$ bins.
}

\comment{\section{Price of Clustering for GBP}
\label{app:POC}

In the next theorems we derive tight bounds for $PoC(GBP)$ depending on the range of values for $R_i$, $1 \leq i \leq m$.

\begin{theorem}
\label{theorem:PoC1}
If $R_i \geq 1$ $\forall i \in [m]$  then $PoC(GBP) = 2$.
\end{theorem}

\begin{proof}
If $R_i \geq 1$ then $$v_{max}^{i} \leq S(I_i) \rightarrow v_{max}^{i}+S(I_i) \leq 2S(I_i)$$
Using Lemma~\ref{lem:balancedColoring}, we have that 
$PoC(GBP) \leq sup_{I,I_i}\frac{\sum_{I_i}2S(I_i)}{OPT(I)}$.
As the total size of items is a lower bound on $OPT(I)$, we have
$$PoC(GBP) \leq sup_{I,I_i}\frac{\sum_{I_i}2S(I_i)}{S(I)} = sup_{I,I_i} \frac{2S(I)}{S(I)} = 2$$

We now show that $2$ is also a lower bound.
Consider a partition of an instance into clusters, such that
$$I_1 = \{G_1 = \{1\},G_2 = \{\frac{1}{N}\}\}, I_2 = \{G_3 = \{1\}, G_4 = \{\frac{1}{N}\}\}, \ldots , I_N = \{G_{2N-1} = \{1 \}, G_{2N} = \{ \frac{1}{N} \} \}$$
In words, each cluster consists of two groups: the first contains a single item of size $1$, and the second contains a single item of size $\frac{1}{N}$.
For the above partition $v^i_{max} = 1$; thus, $R_i = 1+\frac{1}{N} \geq 1, i \in [N]$. In the optimal packing, all items of size $\frac{1}{N}$ can be packed in a single bin, while each item of size 1 is packed in a separate bin. Thus, $OPT(I) = N+1$. Also, $OPT(I_i) = 2$ for all $i$; thus, $\sum_{I_i}OPT(I_i) = 2N$. Hence, $$\lim_{N \rightarrow \infty}\frac{2N}{N+1} = 2.$$ 
It follows that  $PoC(GBP) = 2$.
\qed    \end{proof} }

\comment{ 
 \begin{theorem}
 \label{theorem:PoCk}
if $\forall i \in [m] \: R_i > k$ then $PoC(GBP)_{k \rightarrow \infty} = 1.691...$, and $PoC(GBP)_{k \geq 3} \leq 1.93558$.
\end{theorem}
\begin{proof}
if $\forall i \in [m] \: R_i > k$ then we can use BalancedColoring with $v_{max}^{i}$ colors, for every cluster $I_i$. this assures us that for every two colors $c_i,c_j:i,j \in [v_{max}^{i}]$ their sizes hold: $|c_i-c_j| \leq 1$. This guarantees that after the coloring each color contains size of items of strictly more than $k-1$, hence each color can be seen as a regular bin packing instance such that it needs at least $k$ bins to pack its items. Each color is indeed a BP instance (i.e, GBP instance such that each group is with one item) since we split each group to the  $v_{max}^{i}$ colors, thus there is at most one item from that group in a color after the coloring.

Denote by $c_1,...,c_{v_{max}^{i}}$ the color sets. By looking at some color group $c_j$ as a bin packing instance, then: $OPT(c_j) \geq k$ because the sizes of the bins are $1$, and a lower bound on the optimum is the total size of the items (or closest integer value from above for non integer total size), thus we need at least $k$ bins in order to pack $c_j$.
Consider each color as a new "cluster", thus since the partition into color sets does not necessarily yield the optimum for each cluster, the solution achieved by packing optimally each color set separately is an upper bound on $OPT(I_i)$. Therefore:

$$\displaystyle{PoC(GBP) = sup_{I,I_i}\frac{{\sum_{I_i}OPT(I_i)}}{OPT(I)} \leq sup_{I,I_i}\frac{\sum_{I_i}\sum_{c_j \in c1,...,c_{v_{max}^{i}}}OPT(c_j)}{OPT(I)}}$$

We deduce that these are exactly the conditions for the results of previous works on classic BP with clustering \cite{PoC:A,PoC:LE}, since $OPT(c_j) \geq k$, Therefore:
$PoC(GBP)_{k \rightarrow \infty} \leq 1.691...$, and $PoC(GBP)_{k \geq 3} \leq 1.93558$ 

We now show that $PoC(GBP)_{k \rightarrow \infty} \geq 1.691...$ therefore $PoC(GBP)_{k \rightarrow \infty} = 1.691...$.

Every instance $I$ of BP can be represented as a GBP instance $I'$ such that, each item is in a group with exactly one item. Consequently, A given BP instance $I$  can be transformed into a GBP instance with the same ratio of $\frac{{\sum_{I_i}OPT(I_i)}}{OPT(I)}$, thus $PoC(GBP)(I')$  must be an upper bound on $PoC(BP)(I)$  . Hence $PoC(GBP)_{k \rightarrow \infty} = 1.691...$ by \cite{PoC:A}. 
\qed    \end{proof}

The next theorem shows that in order to have $PoC(GBP) \leq 2$, the condition of theorem ~\ref{theorem:PoC1} is not only sufficient, it is also necessary. The next Theorem also reveals a significant difference between $PoC(BP)$ and $PoC(GBP)$: in BP , because $v^i_{max} = 1$, if $S(I_i) = R_i \geq \alpha, \alpha < 1$ then $PoC(BP) =  \frac{1}{\alpha}$ (can be proven by the fact that $\sum_{I_i}OPT(I_i) \leq \frac{S(I_i)}{\alpha}$), where for GBP we get by the next theorem that $R_i \geq \alpha$, $\forall \alpha < 1$, then $PoC(GBP) = 1+\frac{1}{\alpha}$, which is significantly higher.  

\begin{theorem}
\label{theorem:PoCalpha}
For all $\alpha \in (0,1)$, if $R_i \geq \alpha$ $\forall i \in [m] \:$, then the price of clustering for GBP $(PoC(GBP))$ is $1+\frac{1}{\alpha}$.
\end{theorem}
\begin{proof}
We first prove that $PoC(GBP) \leq 1+\frac{1}{\alpha}$. In addition, we prove that there is a sequence $I^1, \ldots, I^z, \ldots$ such that each element of the sequence is divided into clusters $I^z,\{I_i\}_{i = 1}^{m_z}$, such that $\displaystyle{ \lim_{z \rightarrow \infty}\frac{{\sum_{I_i}OPT(I^z_i)}}{OPT(I^z)} \geq 1+\frac{1}{\alpha}}$. Therefore, $PoC(GBP) \geq 1+\frac{1}{\alpha}$ and we conclude that $PoC = 1+\frac{1}{\alpha}$ under the given constraints.

\begin{enumerate}

\item 
Let there be a GBP instance $I$ and a partition of $I$ into $m$ clusters, $I_1, \ldots, I_m$ such that $R_i \geq \alpha, \forall i \in [m]$. By Lemma~\ref{lem:balancedColoring} we know that any GBP instance I holds $OPT(I) \leq max \{2S(I),S(I)+v_{max}\}$. Let $D \subseteq [m]$ be all indices $i$ for the clusters which hold $2S(I_i) > S(I_i)+v^i_{max}$.  Therefore, 

\[
\begin{array}{ll}
 PoC(GBP) & = \displaystyle{sup_{I,I_i}\frac{\sum_{I_i}OPT(I_i)}{OPT(I)}} \\
 
 \\
 
 & \leq \displaystyle{sup_{I,I_i}\frac{\sum_{i \in D}OPT(I_i)+\sum_{i \in [m] \setminus D}OPT(I_i)}{OPT(I)}} \\
 
  \\

& \leq \displaystyle{sup_{I,I_i}\frac{\sum_{i \in D}2S(I_i)+\sum_{i \in [m] \setminus D}(S(I_i)+v^i_{max})}{OPT(I)}} \\

 \\

& \leq \displaystyle{sup_{I,I_i}\frac{\sum_{i \in D}2S(I_i)+\sum_{i \in [m] \setminus D}(S(I_i)+\frac{S(I_i)}{\alpha})}{OPT(I)}} \\  \\

& \leq \displaystyle{sup_{I,I_i}\frac{\sum_{I_i}(1+\frac{1}{\alpha})S(I_i)}{OPT(I)}} \\  \\

& \leq \displaystyle{sup_{I,I_i}\frac{\sum_{I_i}(1+\frac{1}{\alpha})S(I_i)}{S(I)}} \\  \\

& = \displaystyle{sup_{I,I_i}\frac{(1+\frac{1}{\alpha}) S(I)}{S(I)}} = 1+\frac{1}{\alpha}
 \end{array}
 \]

The third inequality is true because it is given that $R_i \geq \alpha, \alpha < 1$ thus $ \frac{S(I_i)}{\alpha} \geq v^i_{max}$. The forth inequality is true since $S(I) \leq OPT(I)$ and $\alpha < 1$.
\item  We show that $PoC(GBP) \geq 1+\frac{1}{\alpha}$. our way of achieving that is to define a sequence of GBP instances and partitions of them into clusters, with increasing price of clustering, where the price of clustering for the limit of the sequence is at least $1+\frac{1}{\alpha}$. We choose deliberately a partition into clusters such that its price of clustering is high: in each cluster we put two groups, one group with items of sizes close to $0$, and in the other the sizes of the items are close to $1$ such that $R_i \geq \alpha$ and we cannot pack together two items from the two groups due to the capacity constraint. For each group, there is another group in a different cluster such that the total size of any two items together, one from each of group is $1$. The remaining items can be packed in an asymptotically negligible number of bins. The upper bound follows by a concrete definition of such a sequence.

 We define the sequence of GBP instances and for each a corresponding partition into clusters as follows.  Denote the sequence by $a_1 = (I^1,\{I_i\}_{i = 1}^{m_1}), \ldots, a_z = (I^z,\{I^z_i\}_{i = 1}^{m_z}), \ldots$, where $m_z$ will be defined later in order to adjust its value to the analysis. 

Let there be $\alpha < 1$. We now define each element of the sequence, $a_z$. Let ${{\gamma}_z}, {{\beta}_z} \in \mathbb{N}, {{\gamma}_z} < {{\beta}_z} \leq z$ such that $\frac{{{\gamma}_z}}{{{\beta}_z}} \geq \alpha$ and in addition $\frac{{{\gamma}_z}}{{{\beta}_z}} - \alpha$ is minimal. We deduce that  $\lim_{z \rightarrow \infty}\frac{{{\gamma}_z}}{{{\beta}_z}} = \alpha$. 
 We define $I^z$ and its partition into $m_z$ clusters. Cluster $I^z_i, i \in [m_z-1]$ is defined as follows.
 
 $$I^z_{i} = \{G^z_{2{i}-1},G^z_{2{i}}\}$$ 
$$G^z_{2{i}-1} = \{s^z_1,...,s_{{{\beta}_z}}\}, s^z_{\ell} = \sum_{j = 1}^{{i+1}}{e^{-zj}} \: \forall \ell \in [{{{\beta}_z}}]$$ 
$$G^z_{2{i}} = \{s^z_1,...,s_{{{\gamma}_z}}\}, s^z_{\ell} = 1-\sum_{j = 1}^{i}{e^{-zj}} \: \forall \ell \in [{{{\gamma}_z}}]$$

E.g., $I^z_1$ is defined as:  
$$I^z_1 = \{G^z_1,G^z_2\}$$ 
$$G^z_1 = \{s1,...,s^z_{{{\beta}_z}}\}, s^z_{\ell} = e^{-z}+e^{-2z} \: \forall \ell \in [{{{\beta}_z}}]$$ 
$$G^z_2 = \{s^z_1,...,s^z_{{{\gamma}_z}}\}, s^z_{\ell} = 1-e^{-z} \: \forall \ell \in [{{{\gamma}_z}}]$$

In addition, $I^z_{m_z}$ is defined as: 

$$I^z_{{m_z}} = \{G^z_{2{m_z}-1},G^z_{2{m_z}}\}$$ 
$$G^z_{2{m_z}-1} = \{s1,...,s^z_{{{\beta}_z}}\},s^z_{\ell} = e^{-z} \: \forall \ell \in [{{{\beta}_z}}]$$ 
$$G^z_{2{m_z}} = \{s^z_1,...,s^z_{{{\gamma}_z}}\},s^z_{\ell} = 1-\sum_{j = 1}^{{m_z}}{e^{-zj}} \: \forall \ell \in [{{{\gamma}_z}}]$$

Let there be $I^z_i$. We notice that $R_i \geq \alpha$ because $S(I^z_i) > {{\gamma}_z}, v^i_{max} = {{\beta}_z}$ and $\frac{\gamma_z}{\beta_z} \geq \alpha$. Due to the fact that there are no two items which can be accommodated together in $I^z_i$ we  deduce that $OPT(I^z_i) = {{{\beta}_z}}+{{{\gamma}_z}}, \forall i \in [m_z-1]$.
In addition: $OPT(I^z_m) = {{{\beta}_z}}$ by packing one item from $G^z_{2m-1}$ with one item from $G^z_{2m}$ in each bin. Hence: $$\sum_{I^z_i} OPT(I^z_i) = (m_z-1)({{{\beta}_z}}+{{{\gamma}_z}})+{{{\beta}_z}}$$

We now find a packing for $I^z$ without the partition into clusters. Each item of $G^z_4$ is packed with an item from $G^z_1$, each item of $G^z_6$ is packed with an item from $G^z_3$,..., each item of $G^z_{2j}, \forall j \in \{2,3, ...,m_z\}$ is packed with an item from $G^z_{2j-3}$. In addition, each item from $G^z_{2}$ is packed with an item from $G^z_{2m_z-1}$. 

All the remaining items are ${{\beta}_z}-{{\gamma}_z}$ items from each odd numbered group. Depends on the value of $m_z$, the remaining items can be packed in ${{\beta}_z}-{{\gamma}_z}$ bins, i.e., we pack in ${{\beta}_z}-{{\gamma}_z}$ bins an item from each odd numbered group. The total size of each of these ${{\beta}_z}-{{\gamma}_z}$ bins is $\sum_{j = 1}^{m_z}(m_z+1-j)e^{-zj}$. We define $m_z$ such that each of these bins is packed feasibly. 

We demand that $\sum_{j = 1}^{{m_z}}({m_z}+1-j)e^{-zj} \leq \sum_{j = 1}^{{m_z}}({m_z}+1)e^{-zj} \leq 1$ thus we get by sum of a geometric sequence:  $$(m_z+1)(\frac{e^{-z}-e^{-z^{m_z+1}}}{1-e^{-z}}) \leq 1$$

We argue that $m_z = \frac{e^z}{2}$ holds the last inequality thus can make a feasible packing as we described.

 \[
\begin{array}{ll}
 \frac{{\sum_{I^z_i}OPT(I^z_i)}}{OPT(I^z)} & \geq \displaystyle{\frac{(m_z-1)({{\beta}_z}+{{\gamma}_z})+{{\beta}_z}}{{{\gamma}_z}\cdot m_z+{{\beta}_z}-{{\gamma}_z}}} \\  \\
 
 & = \displaystyle{\frac{m_z {{\beta}_z}+(m_z-1)\cdot {{\gamma}_z}}{{{\gamma}_z} \cdot (m_z-1)+{{\beta}_z}-{{\gamma}_z}+{{\gamma}_z}}} \\  \\
 & = \displaystyle{\frac{m_z {{\beta}_z}}{{{\gamma}_z} \cdot m_z+{{\beta}_z}-{{\gamma}_z}}+\frac{(m_z-1)\cdot {{\gamma}_z}}{{{\gamma}_z} \cdot (m_z-1)+{{\beta}_z}}} \\  \\
 
 & \geq \displaystyle{\frac{m_z {{\beta}_z}}{{{\gamma}_z} \cdot m_z+z}+\frac{(m_z-1)\cdot {{\gamma}_z}}{{{\gamma}_z} \cdot (m_z-1)+z}} \\  \\
 
 & = \displaystyle{\frac{{\frac{e^z}{2}} {{\beta}_z}}{{{\gamma}_z} \cdot {\frac{e^z}{2}}+z}+\frac{({\frac{e^z}{2}}-1)\cdot {{\gamma}_z}}{{{\gamma}_z} \cdot ({\frac{e^z}{2}}-1)+z}} \\ \\
 
 & = \frac{1}{\displaystyle{\frac{{{\gamma}_z} \cdot {\frac{e^z}{2}}+z}{{\frac{e^z}{2}} {{\beta}_z}} }} +                           \frac{1}{\displaystyle{\frac{{{\gamma}_z} \cdot ({\frac{e^z}{2}}-1)+z}{({\frac{e^z}{2}}-1)\cdot {{\gamma}_z}}}} \\ \\
 
 & \geq \displaystyle{\frac{1}{\alpha+\frac{1}{z} + \frac{2z}{e^z {{\beta}_z}}  }  + \frac{1}{1+\frac{z}{(e^z/2-1){{\gamma}_z}}}}
 \end{array}
 \]

The first inequality is true since by the packing that we describe above, we get $OPT(I^z) \leq K \cdot m_z + {{\beta}_z} - {{\gamma}_z}$. The second inequality is true since we know that ${{\beta}_z}-{{\gamma}_z} \leq z,{{\beta}_z} \leq z$. The third equality is true since we defined $m_z = \frac{e^z}{2}$. The third inequality is true Because $\frac{{{\gamma}_z}}{{{\beta}_z}}$ is as close as possible to $\alpha$ from all the numbers up to $b$, we deduce that $\alpha \leq \frac{{{\gamma}_z}}{{{\beta}_z}} \leq \alpha+\frac{1}{z}$.

Hence: 
 $$\displaystyle{\lim_{z \rightarrow \infty} \frac{1}{\alpha+\frac{1}{z} + \frac{2z}{e^z {{\beta}_z}}  }  + \frac{1}{1+\frac{z}{(e^z/2-1){{\gamma}_z}}} = 1+\frac{1}{\alpha}}$$

In summary, there is a sequence of instances, such that $$\lim_{z \rightarrow \infty} \frac{{\sum_{I^z_i}OPT(I^z_i)}}{OPT(I^z)} \geq 1+\frac{1}{\alpha}$$
thus because PoC is defined to be the Supremum upon any $I,\{I_i\}_{i = 1}^{m}$, it means that $PoC \geq 1+ \frac{1}{\alpha}$

\end{enumerate}
Concluding the proof, we show that $PoC(GBP) \leq 1+ \frac{1}{\alpha}$ and that $PoC(GBP) \geq 1+ \frac{1}{\alpha}$, hence $PoC(GBP) = 1+ \frac{1}{\alpha}$

\qed    \end{proof}

Denote the APTAS that we obtain in Section~\ref{sec:APTAS} by ${\cal A}$. Let $v^{min}_{max}$ be the size of the maximal group in the cluster where the maximal group is minimal among all other clusters. 

\begin{lemma}
\label{lem:maxmin}
For all $\alpha \in (0,1)$, if $R_i \geq \alpha$ $\forall i \in [m] \:$ then $\frac{m}{S(I)} \leq \frac{1}{v^{min}_{max} \alpha}$
\end{lemma}

 \begin{proof}
\[
 \displaystyle{\begin{array}{ll}
 m \cdot \alpha & \leq \displaystyle{ \sum_{i \in [m]} R_i} \\ \\
 & = \displaystyle{\sum_{i \in [m]} \frac{S(I_i)}{v^i_{max}}} \\ \\
 & = \displaystyle{\frac{\sum_{i \in [m]} \prod_{j \neq i} v^j_{max} S(I_i)}{\prod_{i \in [m]} v^i_{max}}} \\ \\
 & \leq \displaystyle{\frac{\sum_{i \in [m]} \prod_{j \neq {min}} v^j_{max} S(I_i)}{\prod_{i \in [m]} v^i_{max}}} \\ \\
 & = \displaystyle{\frac{\prod_{j \neq {min}} v^j_{max} \sum_{i \in [m]}  S(I_i)}{\prod_{i \in [m]} v^i_{max}}} \\ \\
 & = \displaystyle{\frac{S(I)}{v^{min}_{max}}}

 \end{array}}
 \]
 Therefore, we can use the inequalities above to conclude that $\frac{m}{S(I)} \leq \frac{1}{v^{min}_{max} \alpha}$.
\qed    \end{proof}

The next lemma show bounds on $PoC_{{\cal A}}(GBP)$ using the previous theorems. 

\begin{lemma}
\label{lem:pocA}
For all $\alpha \in (0,1)$, if $R_i \geq \alpha$ $\forall i \in [m] \:$, $PoC_{{\cal A}}(GBP) = 1+\frac{1}{\alpha}+\frac{1}{v^{min}_{max} \alpha(1+\eps)}$
\end{lemma}

\begin{proof}
\[
\begin{array}{ll}
 PoC_{{\cal A}}(GBP) & = \displaystyle{sup_{I,I_i}\frac{\sum_{I_i}{\cal A}(I_i)}{{\cal A}(I)}} \\ \\
 & = \displaystyle{sup_{I,I_i}\frac{\sum_{I_i}(1+\eps)OPT(I_i)+1}{(1+\eps)OPT(I)+1}} \\ \\
 & = \displaystyle{sup_{I,I_i}\frac{m+(1+\eps)\sum_{I_i}OPT(I_i)}{(1+\eps)OPT(I)+1}} \\ \\
 & \leq \displaystyle{sup_{I,I_i}\frac{m+(1+\eps)\sum_{I_i}OPT(I_i)}{(1+\eps)OPT(I)}} \\ \\
  & \leq \displaystyle{\frac{m}{(1+\eps)S(I)}+PoC(GBP)} \\ \\
 & \leq \displaystyle{\frac{1}{v^{min}_{max} \alpha(1+\eps)}+PoC(GBP) = 1+\frac{1}{\alpha}+\frac{1}{v^{min}_{max} \alpha(1+\eps)}}
 \end{array}
 \]
 
 The second equality is true by Section~\ref{sec:APTAS}. The third inequality is true By Lemma~\ref{lem:maxmin}. 
 
 The lower bound can be achieved similarly to the proof of the lower bound in Theorem~\ref{theorem:PoCalpha}. 
\qed    \end{proof}

Similarly to the proof of Lemma~\ref{lem:pocA}, we can prove that if $R_i > k \geq 3, \forall i \in [m]$, then $PoC_{{\cal A}}(GBP) \leq 1.93558+\frac{1}{v^{min}_{max} k(1+\eps)}$ and $lim_{k \rightarrow \infty} PoC_{{\cal A}}(GBP) = 1.691...$ (see the proof of Lemma~\ref{lem:maxmin}). In addition, it stems from the above that if $R_i \geq 1, \forall i \in [m]$, then $PoC_{{\cal A}}(GBP) \leq 2+\frac{1}{v^{min}_{max}(1+\eps)}$

}

\section{Omitted Proofs}\label{omitted-proofs}
\label{appendix:ommited proofs}

\noindent{\bf Proof of Lemma~\ref{lem:k_val}:}
	Assume that, for a given instance $I$, 
	$$\forall k \in \{1,...,\lceil \frac{1}{\varepsilon^2} \rceil \}:  \sum_{\ell \in I :  s_{\ell} \in [\varepsilon^{k+1},\varepsilon^{k})}  s_{\ell} > \varepsilon^{2} \cdot OPT.$$
	For each value of $k$, there is a distinct subset of items whose sizes are in $[\varepsilon^{k+1},\varepsilon^{k})$. 
	Therefore,
	$$\sum_{\ell \in I}  s_{\ell} > \sum_{k = 1}^{\lceil \frac{1}{\varepsilon^{2}}\rceil} \varepsilon^{2} \cdot OPT \geq  OPT.$$ 
	Hence, the items cannot be packed in $OPT$ bins. Contradiction.
\qed    
\posA

\noindent{\bf Proof of Lemma~\ref{lem:few_large_groups}:}
Each large group contains at least $\eps^{k+2} OPT$ items that are large or medium; thus, the total size of a large group is at least $(\eps^{k+2} \cdot OPT)  \eps^{k+1} = \eps^{2k+3} OPT$. Since $OPT$ is an upper bound on the total size of the instance, there are at most $\frac{1}{\eps^{2k+3}}$ large groups. \qed
\posA
\noindent{\bf Proof of Lemma~\ref{lem:shiftingNotIncreaseOPT}:}
 Given a feasible packing  $\Pi$ of the instance $I$, we define a feasible packing $\Pi'$ of $I'$ as follows. For each large group $G_i$, pack items of class $2$ in bins where items of class 1 of $I$ are packed in $\Pi$, items of class 3 where items of class 2 of $I$ are packed in $\Pi$, etc. 

We note that $\Pi'$ is feasible since shifted items in class $r$, for $r > 1$, are no larger than any non-shifted item in class $r-1$. Moreover, there are no conflicts, since shifting is done for each group separately using $\Pi$, which is a feasible packing of $I$.
\qed   
\posA
\noindent{\bf Proof of Lemma~\ref{lem:shiftingCanBeUsedForI}:}
 Given a feasible packing $\Pi'$ of the instance $I'$, we define a feasible packing $\Pi$ of $I$ as follows. For each large group $G_i$, pack items of class 2 where items of class 2 of $I'$ are packed in $\Pi'$, items of class 3 where items of class 3 of $I'$ are packed in $\Pi'$, etc. The items of class $r$ in $I'$ are no smaller than the items of the corresponding class in $I$; thus, the capacity constraint is satisfied. Moreover, no conflict can occur, since $\Pi'$ is a feasible packing for $I'$.
 
 The discarded items can be packed in $O(\eps)OPT$ extra bins. 
 The number of discarded items from each large group is at most $\eps^{2k+4} OPT$. 
 By Lemma~\ref{lem:few_large_groups}, there are $\frac{1}{\eps^{2k+3}}$ large groups; thus, the number of discarded items is at most $\frac{1}{\eps^{2k+3}} \cdot \eps^{2k+4} \cdot OPT = \eps OPT$. It follows that these items fit in at most $O(\eps) OPT$ extra bins. Hence, the resulting packing of $I$ is feasible and uses at most $(1+O(\eps)) OPT$ bins. \qed

\posA
\noindent{\bf Proof of Lemma~\ref{lem:rounding1}:}
Clearly, $OPT$ is an upper bound on the total size of the large and medium items. Since each of these items has a size at least $\eps^{k+1}$, the overall
number of large and medium items is at most $\frac{OPT}{\eps^{k+1}}$. Hence, after shifting, the number of distinct sizes of large items from small groups is at most	
	
	$$\frac{\frac{OPT}{\varepsilon^{k+1}}}{\lfloor{2\varepsilon \cdot OPT}\rfloor} \leq \frac{\frac{\lfloor \eps OPT \rfloor+1}{\varepsilon^{k+2}}}{{\lfloor{2\varepsilon \cdot OPT}\rfloor}} \leq \frac{1}{\varepsilon^{k+2}}+\frac{1}{\varepsilon^{k+2}\lfloor{\varepsilon \cdot OPT}\rfloor} \leq \frac{2}{\varepsilon^{\frac{1}{\varepsilon^2}+3}} = O(1).$$
	
	The second inequality holds since $\lfloor{\varepsilon \cdot OPT}\rfloor \geq 1$. Using a similar calculation for each large group, we conclude that after shifting of these groups, there can be at most $\frac{\frac{OPT}{\eps^{k+1}}}{\lfloor{\eps^{2k+4} \cdot OPT}\rfloor} \leq \frac{2}{\eps^{3k+6}}$ distinct sizes for each group.
	 By Lemma~\ref{lem:few_large_groups}, there are at most $\frac{1}{\eps^{2k+3}}$ large groups. Hence, there can be at most $\frac{2}{\eps^{3k+6}} \cdot \frac{1}{\eps^{2k+3}} = \frac{2}{\eps^{5k+9}}$ distinct sizes for all large and medium items in large groups. In addition, there are $\frac{2}{\eps^{k+3}} = O(1)$ distinct sizes for large items from small groups. Thus, overall there are at most  $\frac{2}{\eps^{k+3}}+\frac{2}{\eps^{5k+9}} = O(1)$ distinct sizes for large items from small groups and large and medium items from large groups. 
\qed   

\posA
\noindent{\bf Proof of Lemma~\ref{lem:dw1}:}
	Let $L$ be the number of large groups. Denote by $G_{i_1}^{\ell}, \ldots, G_{i_L}^{\ell}$ the large and medium items in the large groups. Our scheme enumerates over
	all slot patterns for packing the medium and large items from large groups, and the large items from small groups. The slot patterns indicate how many items of each size are assigned to each bin from each large group.
	
	For the small groups, the slots do not specify to which small group the slots belong. Such slots only indicate that the selected item belongs to a small group. Denote by $T$ the set of slots for an instance $I$, and let $P$ be the set of patterns. Formally, a slot is a $2$-tuple $(s_{\ell},j)$, where $s_{\ell}$ is the size of an item, and $j \in \{i_1, \ldots, i_L\} \cup \{u\}$ labels one of the $L$ large groups, or any of the small groups, represented by a single label $u$. 
	Let $\beta$ be the number of slots in a pattern $p \in P$. It follows that $1 \leq \beta \leq {\lfloor{\frac{1}{\varepsilon^{k+1}}}\rfloor}$ since the number of medium/large items that fit in a single bin is at most ${\lfloor{\frac{1}{\varepsilon^{k+1}}}\rfloor}$.  Then, 
	$p$ is defined as a multi-set, i.e.,  $p=\{t_1,\ldots, t_{\beta}\}$, where  $t_i \in T$, for all $1 \leq  i \leq \beta$. 
	
	 By Lemma~\ref{lem:few_large_groups}, there are at most $\frac{1}{\varepsilon^{2k+3}}$ large groups; thus, the number of distinct labels is at most $\frac{1}{\varepsilon^{2k+3}}+1$. By Lemma~\ref{lem:rounding1}, after rounding the sizes of the large and medium items, there are at most $O(1)$ distinct sizes of these items.
	 Therefore, $|T| = O(1)$.
	  We conclude that $|P| \leq |T|^{\beta} = O(1)$.
	 
	 We proceed to enumerate over the number of bins packed by each pattern. The number of possible packings is $(OPT^{O(1)}) = O(N^{O(1)})$. One of these packings corresponds to
	 an optimal solution for the given instance $I$. At some iteration, this packing will be considered and used in later steps for packing the remaining items. 
	 This gives the statement of the lemma.
\qed    

\posA
\noindent{\bf Proof of Lemma~\ref{lem:classes}:}
Note that at most $\frac{1}{\varepsilon^{k+1}}$ large or medium items can be packed together in a single bin. 
By Lemma~~\ref{lem:rounding1}, after rounding there are at most $\frac{2}{\eps^{k+3}}+\frac{2}{\eps^{5k+11}}$ distinct sizes for these items. Therefore, the number of 
distinct total sizes for bins is at most $(\frac{2}{\eps^{k+3}}+\frac{2}{\eps^{5k+11}})^{\frac{1}{\varepsilon^{k+1}}}$.
By Lemma~\ref{lem:few_large_groups}, there are at most $\frac{1}{\eps^{2k+3}}$ large groups. Thus, the number of subsets of large groups is bounded by
$2^{\frac{1}{\eps^{2k+3}}}$. Each bin can also contain at most $\lceil \frac{1}{\eps^{k+1}} \rceil$ slots assigned to items from small groups. It follows that the total number of bin types is bounded by $$(\frac{2}{\eps^{k+3}}+\frac{2}{\eps^{5k+11}})^{\frac{1}{\varepsilon^{k+1}}} \cdot 2^{\frac{1}{\eps^{2k+3}}} \cdot \frac{1}{\varepsilon^{k+1}} = O(1)$$
\qed    

\comment{\begin{proof}[of Lemma~\ref{lem:larggroupCorrectness}]
	We prove by induction on $W$ that there is a partition of the items in $G_{1}, \ldots, G_{W}$ to the bin types, and a packing of the assigned items in each bin type. This packing corresponds to an optimal solution. For the base case, let $W=1$. In this case, we need to pack a single group, $G_1$. 
	Since the optimal solution consists of $OPT$ bins, one of the partitions of $G_1$ among bin types, considered in
	Step~\ref{step:1}, corresponds to an optimal solution.
	Given such optimal partition, the subset of items $G_1(B_\ell)$ can be feasibly packed in the bins of $B_\ell$, for all $1 \leq \ell  \leq R$. Note that since all bins in $B_\ell$ are {\em identical} (i.e., they have the same pattern and the same total size), any feasible packing of  
	$G_1(B_\ell)$ will be optimal.
	
	For the induction step, assume that the claim holds for $W-1$ groups. We prove the claim for $W$ groups. Algorithm~\ref{Alg:small item from large groups} initially 
	guesses (in Step~\ref{step:1}) a partition of the items in $G_1, \ldots, G_W$ among the bin types. Then, in Step~\ref{step:2} the algorithm considers each type of bins $B_\ell$ separately. If there is only one group to pack in $B_\ell$ then an optimal packing for this bin type exists, as shown for the base case.
	If items from two groups or more need to be packed in $B_\ell$, the algorithm packs $B_{\ell}$ with the items in $G_f(B_\ell)$, where $G_f$ is the lowest index group containing items which are assigned to $B_\ell$. Then, $B_{\ell}$ is split (in Step~\ref{step:subclasses}) into sub-types of bins, $B_{{\ell}_1},...,B_{{\ell}_h}$. By the induction hypothesis, one of the packings output by Algorithm~\ref{Alg:small item from large groups} for these groups corresponds to an optimal solution..
\qed    \end{proof}}

\posA
\noindent{\bf Proof of Theorem~\ref{lem:larggroupComplexity}:}
We first show that EnumGroups indeed outputs a packing for $G_1,\ldots, G_{W}$. Specifically, 
we prove by induction on $W$ that there is a partition of the items in $G_{1}, \ldots, G_{W}$ to the bin types, and a packing of the assigned items in each bin type. This packing corresponds to an optimal solution. For the base case, let $W=1$. In this case, we need to pack a single group, $G_1$. 
Since the optimal solution consists of $OPT$ bins, one of the partitions of $G_1$ among bin types, considered in
	Step~\ref{step:1}, corresponds to an optimal solution.
	Given such optimal partition, the subset of items $G_1(B_\ell)$ can be feasibly packed in the bins of $B_\ell$, for all $1 \leq \ell  \leq R$. Note that since all bins in $B_\ell$ are {\em identical} (i.e., have the same pattern and the same total size), any feasible packing of  
	$G_1(B_\ell)$ is optimal. 
	
	For the induction step, assume that the claim holds for $W-1$ groups. We prove the claim for $W$ groups. Algorithm~\ref{Alg:small item from large groups} initially 
	guesses (in Step~\ref{step:1}) a partition of the items in $G_1, \ldots, G_W$ among the bin types. Then, in Step~\ref{step:2} the algorithm considers each type of bins $B_\ell$ separately. If there is only one group to pack in $B_\ell$, then an optimal packing for this bin type exists, as shown for the base case.
	If items from two groups or more need to be packed in $B_\ell$, the algorithm packs $B_{\ell}$ with the items in $G_f(B_\ell)$, where $G_f$ is the lowest index group containing items which are assigned to $B_\ell$. Then, $B_{\ell}$ is split (in Step~\ref{step:subclasses}) into sub-types of bins, $B_{{\ell}_1},...,B_{{\ell}_h}$. By the induction hypothesis, one of the packings output by Algorithm~\ref{Alg:small item from large groups} for these groups corresponds to an optimal solution.

	We now show that algorithm EnumGroups has polynomial running time. Recall that EnumGroups
	is used twice: Initially, the algorithm is used 
	(in Section~\ref{subsection:Small Items from Large Groups}) to pack small items in large groups. Then, the 
	algorithm is called as a subroutine in algorithm SmallGroups
	(see Section~\ref{subsection:Small Items from Small Groups}), for packing small items in the {\em special} small groups.
	In both of these calls, the distinct number of item sizes in $G_i \in \{ G_1, \ldots, G_W \}$ is at most $1/\eps^{2k+4}$. Hence, the number of partitions of $G_i$ 
	among bin types is bounded by 
	${|G_i|+(R-1) \choose |G_i|}^{1/\eps^{2k+4}} = {OPT+(R-1) \choose OPT}^{1/\eps^{2k+4}}$. Since, by Lemma~\ref{lem:few_large_groups}, 
	the number of large groups is bounded by $W = O(1)$, and by Lemma~\ref{lem:classes}, the number of bin types is bounded by $R = O(1)$, 
	the time complexity of Step~\ref{step:1} is bounded by 
	$${{OPT+(R-1) \choose OPT}^{1/\eps^{2k+4}}}^{\frac{1}{\eps^{2k+3}}} \leq {\frac{(e(OPT+R-1))}{R}^{(R-1) \cdot O(1)}} = O(N^{O(1)}),$$ which is polynomial. The first inequality holds since ${a \choose b} \leq (e a)^b$ for positive integers $1 \leq b \leq a$. 
	
	Once all groups are partitioned among the bin types, the algorithm proceeds recursively and independently for each bin type; thus, it suffices to show that the recursion depth is $O(1)$, and that the number of bin types in each recursive call is bounded by a constant. By Lemma~\ref{lem:classes}, before the first recursive call to EnumGroups the number of bin types is $O_\eps(1)$. Since the number of 
	distinct item sizes in $G_i$ is also a constant, for all $1 \leq i \leq W$, the resulting number of sub-types of bins in each recursive call is also 
	$O(1)$. As for the number of recursive calls,  we note that since in each call the number of groups decreases by one, the recursion depth is
	$O(1)$. As each step has a polynomial running time, the overall running time of EnumGroups is polynomial.
\qed    

\posA
\noindent{\bf Proof of Lemma~\ref{lem:disgroups}:}
	Consider a group $G_{i_j}^s \in A$. Since $G_{i_j}^s$ is {\em special}, it holds that $s_{m+2,i_j} > \frac{\varepsilon \cdot V}{OPT}$.
	Also, as $s_{1,i_j} \geq s_{2,i_j}\geq \ldots \geq s_{m+2,i_j}$, the total size of $G_{i_j}^s$ is larger than
	$(m+2) \cdot \frac{\eps V}{OPT} = (\lfloor \eps OPT \rfloor +2) \frac{\eps V}{OPT} > \eps^{2} V$. 
	Since the total size of the small items in small groups is equal to $V$, the number of special groups satisfies $|A| \leq \frac{1}{\eps^2}$.
\qed    

\posA
\noindent{\bf Proof of Lemma~\ref{lem:helpgreedy}:}
For the base case, let $OPT = 1$. Since there is only one bin in an optimal solution, there is one item from each group (i.e., no conflicts).  Also, if the free space is non-negative, the total size of all items is at most the capacity of the bin (since we have an almost optimal profile), so we can pack all items feasibly.

For the induction step, assume the claim holds for $OPT-1$ bins. Now, suppose that 
there are $OPT$ bins, partially packed with an almost optimal profile. We start by packing the first bin (of maximal total size) using GreedyPack. We distinguish between two cases.

$(i)$ After packing the first bin, the total size of items in this bin is less than $1-\mu$. This means that the first bin contains the largest items left from each small group.
As the bins are sorted in non-decreasing order by free capacity, and the items in each group are sorted in non-increasing order by sizes, we can feasibly pack all items in the remaining $OPT-1$ bins.

$(ii)$ The first bin is packed with total size of at least $1-{\mu}$. 
After packing the first bin, we prove that there is sufficient free space left for using the induction hypothesis. Since the remaining free space in the first bin is at most $\mu$, the total remaining free space is $(OPT-1) \mu$. Hence, the free space on average in the remaining $OPT-1$ bins is at least $\mu$. Also, there are $OPT-1$ bins, and at most $OPT-1$ items from each group (as there is one item from each group in the first bin).
As all required conditions are satisfied, the induction hypothesis can be applied and the remaining items can then be packed in the last $OPT-1$ bins. \qed

\posA
\noindent{\bf Proof of Lemma~\ref{lem:smallO1groups}:}
	In Step~\ref{step:discardregular} of algorithm SmallGroups we discard 
	the $m+1$ largest small items of each small group, whose total size 
	is at least $\eps V$. Thus, the free space on average is at least $\mu \geq \frac{\eps V}{OPT}$. In addition, the maximum size of any remaining item is at most $\frac{V\cdot\varepsilon}{OPT}$ (since we use algorithm EnumGroups for special groups). There are at most $OPT$ items in each group.
	Hence, given an almost optimal profile, by Lemma~\ref{lem:helpgreedy}, all the remaining small items can be packed feasibly.

\qed    

\posA
\noindent{\bf Proof of Lemma~\ref{lem:additionalsmallitems}:}
	In Steps~\ref{step:discardx=1} and~\ref{step:discard1} of Algorithm SmallGroups, we discard the $x$ largest small items of each small group. The total size of these items is at least $\eps OPT$; thus, the free space on average is at least $\mu \geq \eps$. In addition, all the remaining items are of size at most $\eps^{k+1} \leq \mu$. There are at most $OPT$ items from each group. Hence, given an almost optimal profile, by Lemma~\ref{lem:helpgreedy}, all the remaining small items can be packed feasibly.
	
\qed    

\posA
\noindent{\bf Proof of Theorem~\ref{lem:smallgroup}:}
By Lemma~\ref{lem:smallO1groups} and Lemma~\ref{lem:additionalsmallitems}, for any value of $V_{m+1}$ all small items are packed feasibly, given an almost optimal profile. Since our scheme enumerates over all possible slot patterns, it is guaranteed to reach an almost optimal profile. For the running time, 
we first note that, by Theorem~\ref{lem:larggroupComplexity} and Lemma~\ref{lem:disgroups}, the time complexity of EnumGroups (in Step~\ref{step:largegroups} of SmallGroups) is polynomial. Also, algorithm GreedyPack (called in Steps~\ref{step:greedystep2} and~\ref{step:greedystep1} of SmallGroups) runs in time
$O(N \cdot OPT)$. 
Finally, algorithm SmallGroups packs the small items with no violation of bin capacities while avoiding group conflicts. Thus, the resulting packing is feasible.
\qed   

\posA
\noindent{\bf Proof of Lemma~\ref{lem:discardFinal}:}
We prove the claim by deriving a bound on the number of extra bins required for packing the items discarded throughout the execution of the scheme.
To this end, we first bound the total size of these items. We distinguish between the different steps of the scheme.

\begin{enumerate}
\item \label{1}
Rounding the sizes of large and medium items from large groups. By the proof of Lemma~\ref{lem:shiftingCanBeUsedForI}, the total size of items discarded due to shifting is at most ${\eps \cdot OPT}$, and at most $\eps ^{2k+4} OPT$ items from each large groups are discarded.

\item \label{smallshift}
Rounding the sizes of large items from small groups. The total size of items discarded due to the shifting is at most $\lfloor 2\eps \cdot OPT \rfloor$, and at most $\eps ^{k+2} OPT$ items are discarded from each small group, as this is the maximum number of large items in a small group. Also, we discard the items in the last size class, i.e., at most $\lfloor 2\eps \cdot OPT \rfloor$ items of total size at most $\lfloor 2\eps \cdot OPT \rfloor$.

\item \label{4} Packing medium items from small groups (Section~\ref{subsection:Medium Items of Small Groups}). By Lemma~\ref{lem:k_val},
 the total size of these items is at most $\eps^2 \cdot OPT$. Since the items belong to small groups, their total number in each group is at most $\eps^{k+2} OPT -1$.
 
 \item \label{2} Rounding the sizes of small items in large groups (before packing these items by algorithm EnumGroups).
 We discard 
  $\eps^{2k+4} OPT$ items  from each large group. By Lemma~\ref{lem:few_large_groups}, the number of large groups is $L \leq \frac{1}{\eps^{2k+3}}$.
 The size of each small item is at most $\varepsilon^{k+1}$. Therefore, the total size of discarded items is at most 
 $\eps^{2k+4} OPT \cdot L \cdot \eps^{k+1} \leq \eps \cdot OPT$. 
 
\item \label{3} Items discarded in algorithm SmallGroups.
\begin{enumerate}
\item 
In Step~\ref{step:discardx=1}, the total size of discarded items is at most $\eps \cdot OPT+ \eps^{k+1}$. The number of items discarded from each group is at most $x = 1$. 

\item
In Step~\ref{step:discard1}, the total size of discarded items is at most $2 \eps \cdot OPT$.
The number of items discarded from each group is at most $m= \lfloor \eps OPT \rfloor$.

\item In Step~\ref{step:discardregular}, the total size of discarded items is at most $V_{m+1} \leq m$. 
The number of items discarded from each group is exactly $m+1$.

\item
In Step~\ref{step:linearGrouping} we apply shifting to the items in special groups. By Lemma~\ref{lem:disgroups}, the number of these groups satisfies
$|A| \leq \frac{1}{\eps^2}$. At most $Q= \lfloor \eps^{3} \cdot OPT \rfloor$ small items are discarded from each group. Hence, the total size of these items is
at most $\frac{1}{\eps^2} \cdot \eps^{k+1} \lfloor \eps^{3} \cdot OPT \rfloor \leq \eps \cdot OPT$. The number of items discarded from each group is at most $m$.

\item
In GreedyPack some small items are discarded due to conflicts with large items. Since all of these items are from small groups, at most $\eps^{k+2} \cdot OPT$
items are discarded from each group. Assume that the total size of discarded items is larger than $\eps \cdot OPT$. Since each of these items
is {\em coupled} with a large conflicting item from the same group, whose size is at least $1/\eps^2$ times larger, this implies that the total size of large conflicting
items is at least $\frac{OPT}{\eps}$. Contradiction (as $OPT$ is an upper bound on the total size of the instance).

\end{enumerate}
\end{enumerate}

By the above discussion, the total size of discarded items is at most $\eps OPT+4\eps OPT+\eps^2 OPT+\eps OPT+3\eps OPT \leq 10 \eps OPT$ (In \ref{1},~\ref{smallshift}, 
\ref{4}, \ref{2}, \ref{3}, respectively, where we take the worst case in \ref{3}). The maximum number of discarded items from each group is at most $\eps OPT+\eps OPT+\eps OPT+2m+1+\eps^{k+2}OPT \leq 6\eps OPT+1$ (in Steps~\ref{step:discardregular}, \ref{step:linearGrouping} of algorithm SmallGroups, in addition to large items discarded due to shifting). Hence, we can use BalancedColoring to pack the  
items discarded throughout the execution of the scheme in at most $max\{ \lceil 2 \cdot 10\eps OPT \rceil, \lceil 8\eps OPT+6\eps OPT+1 \rceil \} \leq 20\eps OPT+1$ bins. The inequality holds since $OPT > \frac{1}{\eps}$. 
\qed   

\posA
\noindent{\bf Proof of Theorem~\ref{theorem:APTAS}:}
	The feasibility of the packing follows from the way algorithms EnumGroups, GreedyPack and SmallGroups assign items to the bins.
	We now bound the total number of bins used by the scheme. As shown in the proof of Lemma~\ref{lem:discardFinal}, given the parameter $\eps \in (0,1)$, the total number of extra bins used for packing the medium items from small groups and the discarded items is at most $20 \eps OPT +1$. Taking $\eps'= \eps/20$, we have that the total number of bins used by the scheme is $ALG(I) \leq (1+\eps)OPT +1$. As shown above, each step of the scheme has running time polynomial in $N$.

	\comment{
	We note that each item is packed by exactly one way, using the correctness of one of the following lemmas: large and medium items from large groups are packed feasibly by Lemma~\ref{lem:dw1}, the Large items from small groups are packed feasibly by Lemma~\ref{lem:dw2} and the medium items of small groups are packed in extra bins by Lemma~\ref{lem:discardFinal}. Moreover, The small items of large groups are packed feasibly by Theorem~\ref{lem:larggroupCorrectness} and the small items of small groups are packed feasibly by Theorem~\ref{lem:smallgroup}. 
	
	The $OPT$ bins $b_1, \ldots, b_{OPT}$ which we use through all steps are packed with no overflows, as we explicitly prove in previous steps. Combining the number of additional bins that we use through all steps of the packing: the medium items and the discarded items, we need overall $12 \varepsilon OPT+1$ new bins at most by Lemma~\ref{lem:discardFinal}. Assuming that $0 < \varepsilon < 1$ is the desired parameter for the algorithm. Let there be $\varepsilon' = \frac{\varepsilon}{12}$. We use $\varepsilon'$ as the parameter for the algorithm, thus the number of bins that we use is $ALG(I) \leq OPT+12\varepsilon' \cdot OPT+1 = OPT+12\cdot \frac{\varepsilon}{12} \cdot OPT+1 = (1+\varepsilon)OPT +1$. We show for each step that its running time is polynomial. Since there are $O(1)$ steps, the overall running time is polynomial in $N$.
}
\qed   

}

\comment{	
\section{An APTAS for GBP with Fixed Size Optimum}
\label{An APTAS for GBP with fixed size optimum}

We now consider instance $I$ of GBP for which
$OPT(I)  < \frac{3}{\eps^{k+2}}+1$ for some fixed $\eps > 0$. We call GBP on this subclass of instances {\em GBP with fixed size optimum (GFO)}.

Below we describe an adaptation of our scheme in Section~\ref{sec:APTAS} to yield an APTAS for GFO. We use in this scheme simpler classifications for items and the groups.
Given a fixed $\eps \in (0,1)$, an item $\ell$ is {\em large} if $s_{\ell} > \eps$, and {\em small} otherwise.
We show there is a constant number of large items, which can be optimally packed in polynomial time. Small items are packed using a variant of the SmallGroups algorithm.

\begin{theorem}
If $OPT(I)  < \frac{3}{\eps^{k+2}}+1$ then we can find a packing of $I$ in $(1+\eps)OPT(I)+1$ bins in time $N^{O(1)}$.
\end{theorem}

\begin{proof}
As before, we use for short $OPT=OPT(I)$. 
We first show that all items can be feasibly packed in at most
$(1+\eps) OPT+1$ bins. 

\begin{enumerate}
	
    \item \label{11} {\em Large groups}. The size of each large item is at least $\eps$, and the total size of large items is at most $OPT$, thus there can be at most $\frac{OPT}{\eps} = O(1)$ large items. Moreover, there is at least one large item in each large group; therefore, the number of large groups is bounded by the number of large items, which is $O(1)$. In each large group there are at most $OPT = O(1)$ items (small and large), so the number of items in large groups is $O(1)$. Hence, there is only a constant number of items in large groups, and all possible packings of items in large groups can be enumerated in constant time. One of the these packings corresponds to an optimal solution. 
    
    \item \label{22} {\em Small groups}. The packing of the rest of the small items (i.e., the small groups) is done similar to SmallGroups algorithm (Algorithm~\ref{Alg:smallGroups}). The changes in SmallGroups for GFO instances are the following.
     First, there is no need to discard conflicting items, as small groups in GFO instances do not contain large items. Secondly, there is no need to round the special groups, as each of the special groups already contains at most $O(1)$ distinct item sizes. 

Now, we distinguish between two cases.     
    \begin{enumerate}
    
    \item \label{33} If $\eps OPT \geq 1$ then, by Theorem~\ref{lem:smallgroup} and Lemma~\ref{lem:discardFinal}, Algorithm~\ref{Alg:smallGroups} generates a feasible packing in $(1+O(\eps))OPT+1$ bins, including the packing of all discarded items. 
    
    \item \label{44} Otherwise, $\eps OPT < 1$ and $m = 0$. We show that the scheme outputs a packing in at most $(1+\eps)OPT+1$ bins also when Algorithm~\ref{Alg:smallGroups} is used with $m = 0$. If the number of special groups is $O(1)$, then these groups can be packed as in some optimal solution in polynomial time (see Theorem~\ref{lem:larggroupComplexity}). We show below that there are indeed $O(1)$ such groups in GFO instances. In Step~\ref{step:discardregular} of SmallGroups for GFO, we discard exactly one item from each group. The total discarded size is at least $\frac{V}{OPT}$. Thus, on average, for each bin the total discarded size is at least $\frac{V}{OPT^2}$. We now show that the number of groups for which the next largest item is larger than $\frac{V}{OPT^2}$ is bounded by $OPT^2$ (a constant). Assume there are at least $OPT^2+1$ such groups. Then, the total size of small groups is at least $\frac{V}{OPT^2} \cdot (OPT^2+1) >  V$. Contradiction. Hence, we can use algorithm SmallGroups to pack the small groups feasibly in the $OPT$ bins in polynomial time. The total size of discarded items is at most $2 \eps OPT$, and there is at most one discarded item from each group.
    \end{enumerate}
 \end{enumerate} 

   By the above discussion, excluding the discarded items, all items can be packed feasibly in $OPT$ bins. If $\eps OPT \geq 1$, then we can bound the number of extra bins used for discarded items  
  similar to the proof of Lemma~\ref{lem:discardFinal}. Thus, we get a feasible packing of all items
    in at most $(1+O(\eps))OPT+1$ bins. 
    
    Otherwise, $\eps OPT < 1$, i.e., we are in Case~\ref{44}. 
    As $m=0$, items are discarded in algorithm SmallGroups only in
    Steps~\ref{step:discardx=1} or~\ref{step:discardregular}. 
    The total discarded size is then at most $2\eps OPT$, with at most one item discarded from each group. Hence, the discarded items can be packed feasibly in only one additional bin if their total size is at most the capacity of a bin, namely, if $2\eps OPT \leq 1$.    
    If $2\eps OPT > 1$,
    we can take $\eps' =: \frac{1}{2} \cdot \frac{1}{2 OPT} = \frac{1}{4OPT}$. We have that $2\eps' OPT \leq 1$ and the scheme yields a feasible packing in $OPT+1$ bins.
    
    For the running time of the scheme, we note that 
     both the enumeration and algorithm SmallGroups run in time $N^{O_{\eps'}(1)}$.
    Overall, we find a feasible packing of all items in $(1+\eps)OPT+1$ bins in time  $N^{O(1)}$.
\qed    \end{proof}

\section{Bin Packing with Interval-graph Conflicts}
\label{An APTAS for Proper Interval Conflict Graph}

We now consider generalizations of the scheme in Section~\ref{sec:APTAS} to BPC where conflicts are induced by an interval graph. In particular, we discuss the subclasses of {\em proper} and
{\em constant length} interval graphs.
Recall that an interval graph $G=(V,E)$ is the intersection graph of a family 
of intervals on the real-line. Furthermore, the 
maximal cliques of $G$ can be linearly ordered such that for any vertex  $v \in V$, the maximal cliques containing $v$ appear consecutively (see, e.g., \cite{LB62,GH64}). It follows that, given an instance $I$ of BPC with 
interval-graph conflicts, we can represent each item $\ell \in I$ as an interval on the real-line. We can then partition the items into maximal cliques $G_1, \ldots , G_n$, which can be ordered such that each item $\ell$ belongs to $k_\ell$ consecutive cliques in this order for some $k_\ell \geq 1$. 

We note that GBP is the special case where each item belongs to a {\em single} maximal clique.\footnote{Throughout the discussion we use interchangeably {\em cliques} and {\em groups} when referring to maximal subsets of conflicting items.} In Section~\ref{sec:proper_intervals} we discuss the extension of our scheme to BPC where conflicts are induced by a proper interval graph. We call this problem {\em bin packing with proper interval conflicts (PBP)}. In Section~\ref{sec:constant_length_int} we discuss the problem of {\em bin packing with constant length interval graph (BCI)}.\footnote{We note that the scheme in Section~\ref{sec:proper_intervals} can be viewed as a special case of the scheme in Section~\ref{sec:constant_length_int} where $c=2$. For clarity of the presentation we handle first this special case.}

\subsection{Proper Interval Graphs}
\label{sec:proper_intervals}

We note that when the conflict interval graph is proper, each item belongs to at most two consecutive maximal cliques. Consider an item $\ell \in I$; then, for some $1 \leq i \leq n$, $\ell \in G_i$ or 
 $\ell \in G_i \cap G_{i+1}$, where $G_{n+1}= \emptyset$. We give below an overview of the changes required in our scheme to handle such instances. We first classify the items and the cliques (or, {\em groups}) as in Section~\ref{sec:APTAS}. 
As before, after guessing $OPT$, we add to the instance dummy items so that $|G_i|=OPT$. 
In the proof of Theorem~\ref{thm:proper} we refer to the changes required in the scheme of Section~\ref{sec:APTAS} to obtain an APTAS for PBP.

\noindent{\bf Proof of Theorem~\ref{thm:proper}:}
We consider below the main components of the scheme as given in 
Section~\ref{sec:APTAS}. 
\negA
\negA
\paragraph{Large and Medium Items:}
We first slightly modify the rounding for the sizes of large and medium items. Specifically, we consider separately each pair of consecutive cliques of which at least one is a {\em large} group.
For each such pair of groups, we use linear shifting with parameter $Q = \lfloor \eps^{2k+4} \cdot OPT \rfloor$ for all items belonging to this pair of groups. Thus, Lemmas~\ref{lem:shiftingNotIncreaseOPT} and~\ref{lem:shiftingCanBeUsedForI} hold for PBP as well. 
We now take the set of items that belong to small groups (only) and apply to this
set linear shifting with parameter $Q = 2\lfloor 2\eps \cdot OPT \rfloor$.  
	
For packing the large items and medium items from large groups, we refine our definitions of {\em slots} and {\em patterns} as follows.
We assign a label to each large group. Also, we add a label for each pair of consecutive groups in which at least one group is large.
Specifically, we represent the pair $(G_r, G_{r+1})$ by the label 
$\rho_r$, for $1 \leq r \leq n-1$. For a slot that contains an item of size $s_\ell$ which belongs to $G_r$ (only) we use $(s_\ell, i_r)$,
while for an item of the same size that belongs to $G_r$ and $G_{r+1}$, where at least one is large, we use $(s_\ell, \rho_r)$. All slots for items that belong to small groups (only) share the same label. We note that, as the number of labels remains $O_\eps(1)$, we can use enumeration to optimally pack the large and medium items in large groups.

Now, consider large items in small groups. We need to show that these items can be feasibly packed in $OPT$ bins. We distinguish between two cases. If $OPT(I) > \frac{6}{\eps^{k+2}}$ then Theorem~\ref{lem:swap} holds, i.e., the Swapping algorithm resolves all conflicts.
For the complementary case, where $OPT(I) \leq \frac{6}{\eps^{k+2}}$, we can use (with slight modification) the APTAS for GFO (see Appendix~\ref{An APTAS for GBP with fixed size optimum}).

\negA
\negA
\paragraph{Small Items in Large Groups:}
We make the following changes in Algorithm~\ref{Alg:small item from large groups}. Recall that in Step~\ref{step:subclasses} of the algorithm, an item is added to each bin of type $B_{\ell}$. Now, we redefine the sub-types of bins in Step~\ref{step:subclasses2}. Specifically, 
$b_r, b_t \in B_{\ell}$  belong to different sub-types if the item added to $b_r$ is in group $G_i$, while the item added to $b_t$ is both in groups $G_i$ and $G_{i+1}$, for some $1 \leq i \leq W-1$. This is in contrast to Algorithm~\ref{Alg:small item from large groups} for GBP, in which adding items of the same size to $b_r$ and $b_t$ implies that $b_r$ and $b_t$ belong to the same sub-type.
The recursive call to Algorithm~\ref{Alg:small item from large groups} for sub-types containing items from group $G_{i+1}$ includes the groups $G_{i+2},\ldots, G_W$;
the recursive call for sub-types with items from $G_i$ (only) includes the groups $G_{i+1},\ldots, G_W$.
The number of distinct item sizes in each group after rounding is $O(1)$; thus, the running time of Algorithm~\ref{Alg:small item from large groups} remains polynomial, and Theorem~\ref{lem:larggroupComplexity} holds for PBP. 
\negA
\negA
\paragraph{Small Items from Small Groups:}
Recall the idea of the SmallGroups algorithm. First, discard items of total size $\eps V$ (or $\eps OPT$), and then pack greedily the maximum possible 
total size in each bin, taking one item from each group. We slightly modify this approach for solving PBP. Considering the bins in non-increasing order of packed size, we assign to each bin as before items from all groups, such that the total size of packed items in a bin is at
 least $1-\frac{\eps V}{OPT}$ (same as in Algorithm~\ref{Alg:smallGroups}). 
 We use dynamic programming (DP) to find an item from each small group (one item may cover two groups), such that the total size packed in the current bin is maximized. If the total size exceeds the bin capacity, we discard the largest item from one group and use the DP algorithm to compute a new packing.
 The process is repeated while the packed items cause an overflow in the current bin.
This guarantees that in each iteration the total size of the remaining items decreases, until a packing with no overflow is found. Given an almost optimal profile, and by discarding enough total size, we obtain a feasible packing (see Theorem~\ref{lem:smallgroup}).

Formally, let $G_1, \ldots, G_S$ be the small groups, ordered such that each item belongs to at most two consecutive groups. Let $B[i]$ denote the maximal
 total size which can be packed in a single bin, taking items from groups $G_1, \ldots, G_i$.\footnote{
 	The DP computes only the maximal solution value rather than the corresponding subset of items; this can be done in polynomial time and space.} Thus, $B[0] = 0$ and $B[1] = max \{s_{\ell} | \ell \in G_1\}$. Let $\ell_{i+1}$ ($\ell_{i, i+1}$) be the largest remaining item that belongs to $G_i$ ($G_{i}$ and $G_{i+1}$). If there is no such item
 then we set $s_{\ell_{i+1}} = 0$ ($s_{\ell_{i, i+1}} = 0$). We use the following recursion to compute $B[i+1]$, $1 \leq i \leq S-1$:
 
$$B[i+1] = max \{s_{\ell_{i+1}}+B[i] \: , \: s_{\ell_{i, i+1}}+B[i-1]\}$$

Thus, $B[S]$ can be computed in polynomial time. We prove by induction on $i$ that $B[i+1]$ is the maximal total size which can be packed in a single bin, taking items from groups $G_1, \ldots, G_i$. The claim holds trivially for $B[0], B[1]$. Assume that the claim holds for $B[i']$ such that  $i' < i+1$. For $B[i+1]$, we take the maximum between $s_{\ell_{i+1}}+B[i]$ and $s_{\ell_i,{i+1}}+B[i-1]$. This is because we must add an item from each of $G_{i}, G_{i+1}$. This can be done either 
by taking one item from each group or one item that belongs to both groups. Thus, using the induction hypothesis for $B[i-1]$ and $B[i]$ we have the claim. 

The above approach is a generalization of Algorithm~\ref{Alg:greedyPack}: the packing of each bin by GreedyPack is computed using the DP subroutine, which is called at most $N$ times  for each bin. Hence, overall we get a polynomial running time for packing the remaining small items in $OPT$ bins. 

 Let $b$ be the bin currently filled by GreedyPack, and let $s$ be the size of the largest remaining item. We show below that the total size packed in $b$ (i.e., $B[S]$) decreases at most by $2s$ between any two consecutive iterations of the DP subroutine. In each iteration, exactly one item is discarded. This item belongs to at most two groups. Thus, the value of $B[S]$ in the next iteration may decrease by the size of at most two items. Indeed, after each iteration of the DP subroutine, we choose a group $G_{i}$ and discard its largest remaining item, denoted by $h$. This item is then replaced by the next largest item in $G_{i}$, denoted by $\ell$. It holds that $s_h-s_\ell \leq s$, and consequently, this change decreases $B[S]$ at most by $s$. 

As the conflict graph is a proper interval graph, item $\ell$ may belong also to $G_{i+1}$ or $G_{i-1}$. Assume w.l.o.g. that $\ell \in G_{i+1}$. If $\ell$ is used for $B[S]$ in the next iteration, adding another item in $G_{i+1}$ will generate a conflict. Therefore, $B[S]$ may decrease by the size of an item in $G_{i+1}$ that was used in the previous iteration. The size of this item is at most $s$. We conclude that $B[S]$ decreases at most by $2s$ between any two consecutive iterations.

Recall that Lemma~\ref{lem:helpgreedy} (for GBP) requires that $\mu$, the free space on average, must be at least the size of the largest remaining item, $s$. It suffices to guarantee that $\mu \geq s$ since, between any two attempts to pack the current bin $b$, the total size packed in $b$ decreases by at most $s$ (see the proof of Lemma~\ref{lem:helpgreedy}). On the other hand, as shown above, for PBP instances, the total size packed in $b$ decreases at most by $2s$ between consecutive iterations.

Hence, to ensure the correctness of Lemma~\ref{lem:helpgreedy} for PBP, we need to double the increase in the free space on average. This can be done by increasing the total size of discarded items in Algorithm~\ref{Alg:smallGroups}. More specifically, Algorithm~\ref{Alg:smallGroups} will use the parameter $m = 2 \lfloor \eps OPT \rfloor$.
Note that the total size of discarded items remains $O(\eps) OPT$, including discarded items in Step~\ref{step:discardsmalloverlarge} of Algorithm~\ref{Alg:greedyPack}; thus, Lemma~\ref{lem:discardFinal} holds for PBP. \qed

\subsection{Constant Length Interval Graphs}
\label{sec:constant_length_int}
 
 In a {\em constant length} interval graph, each interval belongs to at most a constant number of maximal cliques.
 In the corresponding instances of BPC, each 
 item belongs to only a constant number $c = O(1)$ of cliques, which appear consecutively
  on the real-line. This is a generalization of BPC with proper interval conflict graphs, where an item belongs to at most two consecutive groups. 
 
 We describe below the changes required for generalizing the scheme in Section~\ref{sec:proper_intervals} to bin packing with constant length interval graphs.
The proof of Theorem~\ref{thm:ConstantLength} relies on the claims we prove in Section~\ref{sec:APTAS} and on the ideas in the proof of Theorem~\ref{thm:proper}.

\noindent{\bf Proof of Theorem~\ref{thm:ConstantLength}:}
We discuss below the main components of the scheme as given in Sections~\ref{sec:APTAS} and~\ref{sec:proper_intervals}. 
 
 \negA
 \negA
 \paragraph{Large and Medium Items:}
We first slightly modify the rounding of sizes for the large and medium items. Specifically, we consider separately each subset of at most $c$ consecutive cliques
(= groups) in which at least one group is {\em large}.
For each such subset of groups, we use linear shifting with parameter $Q = \lfloor \eps^{2k+4} \cdot OPT \rfloor$ for all large and medium items that belong exactly to this subset of groups.
We now take together all the large and medium items belonging to small groups (only). We apply to these items linear shifting with parameter $Q = c \lfloor 2\eps \cdot OPT \rfloor$. 
	
For packing the large items and medium items from large groups, we refine our definition of {\em slots} and {\em patterns} as follows.
We assign a distinct label to each subset of at most $c$ consecutive groups in which at least one group is large.
Specifically, we represent the subset $G_r,\ldots, G_{r+i}$ by the label 
$\rho_{r, r+i}$, for $0 \leq i \leq c-1$, $1 \leq r \leq n-i$. 
 For a slot that contains an item of size $s_\ell$ which belongs to groups $G_r,\ldots, G_{r+i}$, at least one of which is large, we use the pair
  $(s_\ell, \rho_{r, r+i})$. All slots for items that belong to small groups only share the same label. 

We note that, as the number of labels remains $O_\eps(1)$, we can use enumeration to optimally pack the large and medium items in large groups. For BCI, we require 
that $OPT(I) > c\frac{3}{\eps^{k+2}}$ to guarantee that Theorem~\ref{lem:swap} holds. For the complementary case, where $OPT(I) \leq c\frac{3}{\eps^{k+2}}$, we can use (with slight modification) the APTAS for GFO (see Appendix~\ref{An APTAS for GBP with fixed size optimum}).

\negA
\negA
\paragraph{Small Items in Large Groups:}
We make the following changes in Algorithm~\ref{Alg:small item from large groups}. We first redefine the sub-types of bins in Step~\ref{step:subclasses2}. Specifically, 
$b_r, b_t \in B_{\ell}$  belong to different sub-types if the items added to $b_r, b_t$ belong to different subsets of groups. 
The recursive call to Algorithm~\ref{Alg:small item from large groups} for sub-types containing items from groups $G_{i}, \ldots, G_{i+q}$ includes the groups $G_{i+q+1},\ldots, G_W$. The number of distinct item sizes in each group after rounding is $O(1)$; thus, the running time remains polynomial, and Theorem~\ref{lem:larggroupComplexity} holds for BCI. 

\negA
\negA
\paragraph{Small Items from Small Groups:}
We describe below the changes in the DP subroutine of Section~\ref{sec:proper_intervals}. Formally, let $G_1, \ldots, G_S$ be the small groups, ordered such that each item belongs to at most $c$ consecutive groups. Let $B[i]$ denote the maximal
 total size which can be packed in a single bin, taking items from groups $G_1, \ldots, G_i$.\footnote{For simplicity, the DP subroutine computes only the maximal solution value rather than the corresponding subset of items; this can be done in polynomial time and space.} We find the values of $B[0], B[1], \ldots, B[c-1]$ by enumerating over all possible packings. This can be achieved in polynomial time, as for each of these values, we consider a constant number of groups. Let $\ell_{i, i+q}$ be the largest remaining item that belongs to groups $G_i, \ldots, G_{i+q}$ (only), for $0 \leq q \leq c-1, i+q \leq S$. If there is no such item, then we set $s_{\ell_{i, i+q}} = 0$. We use the following recursion to compute $B[i]$, $c \leq i \leq S$:
 
$$B[i] = max_{q \in \{0, \ldots, c-1\}} \{s_{\ell_{i-q, i}}+B[i-q-1]\}$$

We prove by induction on $i$ that $B[i]$ is the maximal total size which can be packed in a single bin, taking items from groups $G_1, \ldots, G_i$. The claim holds trivially for $B[0], B[1], \ldots, B[c-1]$ as we enumerate over all possible packings. Assume that the claim holds for $B[i']$ such that  $i' < i$. Then $B[i]$ is the maximal value obtained by taking $s_{\ell_{i-q},{i}}+B[i-q-1]$, where $q \in \{0,\ldots,c-1\}$. This is because in $B[i]$ we cover each group by exactly one item (i.e., taking one item from each group) and each item in a BCI instance belongs to at most $c$ consecutive groups. Thus, by using the induction hypothesis for $B[i-q-1]$ we obtain the claim.

 Let $b$ be the current bin to be packed by the GreedyPack algorithm, and let $s$ be the size of the largest remaining item. Similar to our discussion in Section~\ref{sec:proper_intervals}, the total size packed in $b$ (i.e., $B[S]$) decreases at most by $c \cdot s$ between any two consecutive iterations of the DP subroutine. Thus, it suffices to require that $\mu \geq c s$ (i.e., the free space on average is at least $c$ times larger than the largest remaining item). The correctness of Lemma~\ref{lem:helpgreedy} follows, and thus also the correctness of Algorithm~\ref{Alg:smallGroups}. To guarantee that $\mu \geq c s$, Algorithm~\ref{Alg:smallGroups} is computed with parameter $m'' = c \lfloor \eps OPT \rfloor$.  Note that the total size of discarded items remains $O(\eps) OPT$; thus, Lemma~\ref{lem:discardFinal} holds also for BCI.
 \qed
 

}

\comment{
\section{Definition of Slots and Patterns}
\label{sec:Definitions}
A {\em slot} is characterized by a size, and by a {\em label}. A label can represent explicitly one of the large groups, or can denote `small group' with no indication to which small group it belongs. Denote by $u$ the label for all the small groups. Let $G_{i_1}, \ldots, G_{i_L}$ be the large groups.    
Formally, a slot is a pair $(s_{\ell},j)$, where $s_{\ell}$ is the size of an item $\ell \in I$ and $j \in \{i_1, \ldots, i_L\} \cup \{u\}$. A {\em pattern} is a multiset $\{t_1,\ldots, t_{\beta}\}$ containing at most ${\lfloor{\frac{1}{\varepsilon^{k+1}}}\rfloor}$ elements, where $t_i$ is a slot for each $i \in [\beta]$. 
}

\section{Omitted Proofs}
\label{Omitted Proofs}

\comment{\noindent{\bf Proof of Lemma~\ref{lem:k_val}:}

Assume that, for a given instance $I$, 
	$$\forall k \in \{1,...,\lceil \frac{1}{\varepsilon^2} \rceil \}:  \sum_{\ell \in I :  s_{\ell} \in [\varepsilon^{k+1},\varepsilon^{k})}  s_{\ell} > \varepsilon^{2} \cdot OPT.$$
	For each value of $k$, there is a distinct subset of items whose sizes are in $[\varepsilon^{k+1},\varepsilon^{k})$. 
	Therefore,
	$$\sum_{\ell \in I}  s_{\ell} > \sum_{k = 1}^{\lceil \frac{1}{\varepsilon^{2}}\rceil} \varepsilon^{2} \cdot OPT \geq  OPT.$$ 
	Hence, the items cannot be packed in $OPT$ bins. Contradiction.
\qed    }

\noindent{\bf Proof of Lemma~\ref{lem:few_large_groups}:}
Each large group contains at least $\eps^{k+2} OPT$ items that are large or medium; thus, the total size of a large group is at least $(\eps^{k+2} \cdot OPT)  \eps^{k+1} = \eps^{2k+3} OPT$. Since $OPT$ is an upper bound on the total size of the instance, there are at most $\frac{1}{\eps^{2k+3}}$ large groups. \qed

\noindent{\bf Proof of Lemma~\ref{lem:shiftingNotIncreaseOPT}:}
Given a feasible packing  $\Pi$ of the instance $I$, we define a feasible packing $\Pi'$ of $I'$ as follows. For each large group $G_i$, pack items of class $2$ in bins where items of class 1 of $I$ are packed in $\Pi$, items of class 3 where items of class 2 of $I$ are packed in $\Pi$, etc. 

We note that $\Pi'$ is feasible since shifted items in class $r$, for $r > 1$, are no larger than any non-shifted item in class $r-1$. Moreover, there are no conflicts,
since shifting is done for each group separately using $\Pi$, which is a feasible packing of $I$.
\qed

\noindent{\bf Proof of Lemma~\ref{lem:shiftingCanBeUsedForI}:}
Given a feasible packing $\Pi'$ of the instance $I'$, we define a feasible packing $\Pi$ of $I$ as follows. For each large group $G_i$, pack items of class 2 where items of class 2 of $I'$ are packed in $\Pi'$, items of class 3 where items of class 3 of $I'$ are packed in $\Pi'$, etc. The items of class $r$ in $I'$ are no smaller than the items of the corresponding class in $I$; thus, the capacity constraint is satisfied. Moreover, no conflict can occur since $\Pi'$ is a feasible packing for $I'$.

The discarded items can be packed in $O(\eps)OPT$ extra bins. 
 The number of discarded items from each large group is at most $\eps^{2k+4} OPT$. 
 By Lemma~\ref{lem:few_large_groups}, there are $\frac{1}{\eps^{2k+3}}$ large groups; thus, the number of discarded items is at most $\frac{1}{\eps^{2k+3}} \cdot \eps^{2k+4} \cdot OPT = \eps OPT$. It follows that these items fit in at most $O(\eps) OPT$ extra bins. Hence, the resulting packing of $I$ is feasible and uses at most $(1+O(\eps)) OPT$ bins. \qed

\posA
\noindent{\bf Proof of Lemma~\ref{lem:rounding1}:}
Clearly, $OPT$ is an upper bound on the total size of the large and medium items. Since each of these items has a size at least $\eps^{k+1}$, the overall
number of large and medium items is at most $\frac{OPT}{\eps^{k+1}}$. Hence, after shifting, the number of distinct sizes of large items from small groups is at most	
	
	$$\frac{\frac{OPT}{\varepsilon^{k+1}}}{\lfloor{2\varepsilon \cdot OPT}\rfloor} \leq \frac{\frac{\lfloor \eps OPT \rfloor+1}{\varepsilon^{k+2}}}{{\lfloor{2\varepsilon \cdot OPT}\rfloor}} \leq \frac{1}{\varepsilon^{k+2}}+\frac{1}{\varepsilon^{k+2}\lfloor{\varepsilon \cdot OPT}\rfloor} \leq \frac{2}{\varepsilon^{\frac{1}{\varepsilon^2}+3}} = O(1).$$
	
	The second inequality holds since $\lfloor{\varepsilon \cdot OPT}\rfloor \geq 1$. Using a similar calculation for each large group, we conclude that after shifting of these groups, there can be at most $\frac{\frac{OPT}{\eps^{k+1}}}{\lfloor{\eps^{2k+4} \cdot OPT}\rfloor} \leq \frac{2}{\eps^{3k+6}}$ distinct sizes for each group.
	 By Lemma~\ref{lem:few_large_groups}, there are at most $\frac{1}{\eps^{2k+3}}$ large groups. Hence, there can be at most $\frac{2}{\eps^{3k+6}} \cdot \frac{1}{\eps^{2k+3}} = \frac{2}{\eps^{5k+9}}$ distinct sizes for all large and medium items in large groups. In addition, there are $\frac{2}{\eps^{k+3}} = O(1)$ distinct sizes for large items from small groups. Thus, overall there are at most  $\frac{2}{\eps^{k+3}}+\frac{2}{\eps^{5k+9}} = O(1)$ distinct sizes for large items from small groups and large and medium items from large groups. 
\qed

\posA
\noindent{\bf Proof of Lemma~\ref{lem:dw1}:}
	Let $L$ be the number of large groups. Denote by $G_{i_1}^{\ell}, \ldots, G_{i_L}^{\ell}$ the large and medium items in the large groups. Our scheme enumerates over
	all slot patterns for packing the medium and large items from large groups, and the large items from small groups. The slot patterns indicate how many items of each size are assigned to each bin from each large group.
	
Denote by $T$ the set of slots for an instance $I$, and let $P$ be the set of patterns. Recall that a slot is a $2$-tuple $(s_{\ell},j)$, where $s_{\ell}$ is the size of an item, and $j \in \{i_1, \ldots, i_L\} \cup \{u\}$ labels one of the $L$ large groups, or any of the small groups, represented by a single label $u$. 
	Let $\beta$ be the number of slots in a pattern $p \in P$. We note that $1 \leq \beta \leq {\lfloor{\frac{1}{\varepsilon^{k+1}}}\rfloor}$ since the number of medium/large items that fit in a single bin is at most ${\lfloor{\frac{1}{\varepsilon^{k+1}}}\rfloor}$.  Then, 
	$p$ is defined as a multi-set, i.e.,  $p=\{t_1,\ldots, t_{\beta}\}$, where  $t_i \in T$, for all $1 \leq  i \leq \beta$. 
	
	 By Lemma~\ref{lem:few_large_groups}, there are at most $\frac{1}{\varepsilon^{2k+3}}$ large groups; thus, the number of distinct labels is at most $\frac{1}{\varepsilon^{2k+3}}+1$. By Lemma~\ref{lem:rounding1}, after rounding the sizes of the large and medium items, there are at most $O(1)$ distinct sizes of these items.
	 Therefore, $|T| = O(1)$
	 We conclude that $|P| \leq |T|^{\beta} = O(1)$.
	 
	 We proceed to enumerate over the number of bins packed by each pattern. The number of possible packings is $OPT^{O(1)} = O(N^{O(1)})$. One of these packings corresponds to
	 an optimal solution for the given instance $I$. At some iteration, this packing will be considered and used in later steps for packing the remaining items. 
	 This gives the statement of the lemma.
\qed

\posA
\noindent{\bf Proof of Theorem~\ref{lem:swap}:} 
We prove that for each conflict involving an item $\ell \in G_i$ of size $s_\ell$ in bin $b$, there is an item $y \in G_j \neq G_i$ of size $s_{y} = s_\ell$ in bin $c \neq b$, such that $swap(\ell,y)$ is good. Consider 
a packing of large and medium items by a slot pattern corresponding to an optimal solution. Then, the items are packed in $OPT$ bins with no
overflow, and the only conflicts may occur among items from small groups.

Due to shifting with parameter $Q = \lfloor 2\eps \cdot OPT \rfloor$ for large items from small groups, there are $\lfloor 2\eps \cdot OPT \rfloor-1$ items of size $s_\ell$ in addition to $\ell$ (recall that the last size class, which may contain less items, is discarded).
We prove that the number of items $y$ for which $swap(\ell,y)$ is bad is at most $\lfloor 2\eps \cdot OPT \rfloor-2$; therefore, there exists an item $y$ of size $s_{\ell}$, for which $swap(\ell,y)$ is good. 
We note that $swap(\ell,y)$ is bad if (at least) one of the following holds: $(i)$ $y$ belongs to a group $G_j$ which has an item in bin $b$, or
$(ii)$ there is an item from $G_i$ in bin $c$.

We handle $(i)$ and $(ii)$ separately. 
$(i)$ The number of items of size $s_\ell$ from groups $G_j$ that have an item in bin $b$ is at most $\frac{1}{\eps^{k}} \cdot \eps^{k+2} OPT = \eps^2 OPT$, since at most $\frac{1}{\eps^k}$ small groups can have a large item in $b$, and each such group has at most $\eps^{k+2} OPT$ large items.
$(ii)$ We now bound the number of items of size $s_\ell$ in bins that contain items in $G_i$. We note that the total number of bins containing large items in $G_i$ is at most 
$\eps^{k+2} OPT$, since $G_i$ is small. Also, in each such bin, the total number of large items is at most $\frac{1}{\eps^k}$. Thus, the total number of items of size $s_\ell$ in bins containing items in $G_i$ is at most $\frac{1}{\eps^{k}} \eps^{k+2} OPT = \eps^2 OPT$.
Using the union bound, the number of bad swaps 
for $\ell$, i.e., $swap(\ell,y)$ for some item $y$,
is at most 
$2\eps^2 OPT$. We have $2\eps^2OPT < \eps OPT < \eps OPT+\eps OPT-3 \leq \lfloor 2\eps \cdot OPT \rfloor-2.$ The first inequality holds since we may assume that $\eps < \frac{1}{2}$.
For the second inequality, we note that $OPT > \frac{3}{\eps^{k+2}} > \frac{3}{\eps}$. 
We conclude that there is an item $y$ in the size class of $\ell$
such that $swap(\ell,y)$ is good. 

We now show that the Swapping algorithm is polynomial in $N$. We note that items of some group are in conflict only if they are placed in the same bin. As these are only large items, an item may conflict with at most $\frac{1}{\eps^{k}}$ items. Hence, there are at most $\frac{N}{\eps^{k}} = O(N)$ conflicts.
As finding a good swap takes at most $O(N)$, the overall running time of Swapping is $O(N^2)$. \qed

\posA
\noindent{\bf Proof of Lemma~\ref{lem:classes}:}
Note that at most $\frac{1}{\varepsilon^{k+1}}$ large or medium items can be packed together in a single bin. 
By Lemma~~\ref{lem:rounding1}, after rounding there are at most $\frac{2}{\eps^{k+3}}+\frac{2}{\eps^{5k+11}}$ distinct sizes for these items. Therefore, the number of 
distinct total sizes for bins is at most $(\frac{2}{\eps^{k+3}}+\frac{2}{\eps^{5k+11}})^{\frac{1}{\varepsilon^{k+1}}}$.
By Lemma~\ref{lem:few_large_groups}, there are at most $\frac{1}{\eps^{2k+3}}$ large groups. Thus, the number of subsets of large groups is bounded by
$2^{\frac{1}{\eps^{2k+3}}}$. Each bin can also contain at most $\lceil \frac{1}{\eps^{k+1}} \rceil$ slots assigned to items from small groups. It follows that the total number of bin types is bounded by $$(\frac{2}{\eps^{k+3}}+\frac{2}{\eps^{5k+11}})^{\frac{1}{\varepsilon^{k+1}}} \cdot 2^{\frac{1}{\eps^{2k+3}}} \cdot \frac{1}{\varepsilon^{k+1}} = O(1)$$
\qed

\comment{
\noindent{\bf Proof of Theorem~\ref{lem:larggroupComplexity}:}
We first show that EnumGroups indeed outputs a packing for $G_1,\ldots, G_{W}$. Specifically, 
we prove by induction on $W$ that there is a partition of the items in $G_{1}, \ldots, G_{W}$ to the bin types, and a packing of the assigned items in each bin type. This packing corresponds to an optimal solution. For the base case, let $W=1$. In this case, we need to pack a single group, $G_1$. 
Since the optimal solution consists of $OPT$ bins, one of the partitions of $G_1$ among bin types, considered in
	Step~\ref{step:1}, corresponds to an optimal solution.
	Given such optimal partition, the subset of items $G_1(B_\ell)$ can be feasibly packed in the bins of $B_\ell$, for all $1 \leq \ell  \leq R$. Note that since all bins in $B_\ell$ are {\em identical} (i.e., have the same pattern and the same total size), any feasible packing of  
	$G_1(B_\ell)$ is optimal. 
	
	For the induction step, assume that the claim holds for $W-1$ groups. We prove the claim for $W$ groups. Algorithm~\ref{Alg:small item from large groups} initially 
	guesses (in Step~\ref{step:1}) a partition of the items in $G_1, \ldots, G_W$ among the bin types. Then, in Step~\ref{step:2} the algorithm considers each type of bins $B_\ell$ separately. If there is only one group to pack in $B_\ell$, then an optimal packing for this bin type exists, as shown for the base case.
	If items from two groups or more need to be packed in $B_\ell$, the algorithm packs $B_{\ell}$ with the items in $G_f(B_\ell)$, where $G_f$ is the lowest index group containing items which are assigned to $B_\ell$. Then, $B_{\ell}$ is split (in Step~\ref{step:subclasses}) into sub-types of bins, $B_{{\ell}_1},...,B_{{\ell}_h}$. By the induction hypothesis, one of the packings output by Algorithm~\ref{Alg:small item from large groups} for these groups corresponds to an optimal solution.

	We now show that algorithm EnumGroups has polynomial running time. Recall that EnumGroups
	is used twice: Initially, the algorithm is used 
	(in Section~\ref{subsection:Small Items from Large Groups}) to pack small items in large groups. Then, the 
	algorithm is called as a subroutine in algorithm SmallGroups
	(see Section~\ref{subsection:Small Items from Small Groups}), for packing small items in the {\em special} small groups.
	In both of these calls, the distinct number of item sizes in $G_i \in \{ G_1, \ldots, G_W \}$ is at most $1/\eps^{2k+4}$. Hence, the number of partitions of $G_i$ 
	among bin types is bounded by 
	${|G_i|+(R-1) \choose |G_i|}^{1/\eps^{2k+4}} = {OPT+(R-1) \choose OPT}^{1/\eps^{2k+4}}$. Since, by Lemma~\ref{lem:few_large_groups}, 
	the number of large groups is bounded by $W = O(1)$, and by Lemma~\ref{lem:classes}, the number of bin types is bounded by $R = O(1)$, 
	the time complexity of Step~\ref{step:1} is bounded by 
	$${{OPT+(R-1) \choose OPT}^{1/\eps^{2k+4}}}^{\frac{1}{\eps^{2k+3}}} \leq {\frac{(e(OPT+R-1))}{R}^{(R-1) \cdot O(1)}} = O(N^{O(1)}),$$ which is polynomial. The first inequality holds since ${a \choose b} \leq (e a)^b$ for positive integers $1 \leq b \leq a$. 
	
	Once all groups are partitioned among the bin types, the algorithm proceeds recursively and independently for each bin type; thus, it suffices to show that the recursion depth is $O(1)$, and that the number of bin types in each recursive call is bounded by a constant. By Lemma~\ref{lem:classes}, before the first recursive call to EnumGroups the number of bin types is $O_\eps(1)$. Since the number of 
	distinct item sizes in $G_i$ is also a constant, for all $1 \leq i \leq W$, the resulting number of sub-types of bins in each recursive call is also 
	$O(1)$. As for the number of recursive calls,  we note that since in each call the number of groups decreases by one, the recursion depth is
	$O(1)$. As each step has a polynomial running time, the overall running time of EnumGroups is polynomial.
\qed

\posA
\noindent{\bf Proof of Lemma~\ref{lem:additionalsmallitems}:}
	In Steps~\ref{step:discardx=1} and~\ref{step:discard1} of Algorithm SmallGroups, we discard the $x$ largest small items of each small group. The total size of these items is at least $\eps OPT$; thus, the free space on average is at least $\mu \geq \eps$. In addition, all the remaining items are of size at most $\eps^{k+1} \leq \mu$. There are at most $OPT$ items from each group. Hence, given an almost optimal profile, by Lemma~\ref{lem:helpgreedy}, all the remaining small items can be packed feasibly. \qed

}

\posA
\noindent{\bf Proof of Lemma~\ref{lem:recursiveOPT}:}
We prove that $RecursiveEnum$ satisfies all the conditions.

(i) We first show that the running time is polynomial.
Using arguments similar to the proof of
Lemma~\ref{lem:few_large_groups}, there is a constant number of $t(g_{opt})$-significant groups for each $t \in T$ and $h = 0,\ldots,\alpha$ (Step~\ref{step:t-significantGroups}). These groups can be guessed in polynomial time. Furthermore, guessing the largest item in each of these groups (Step~\ref{step:Representative}), for each size class in the linear shifting, is also done in polynomial time, because there is a constant number of size classes. Then, the enumeration of patterns in Step~\ref{step:t-pattern} can be done in polynomial time, for each type $t$.

By Lemma~\ref{lem:classes}, the initial number of types is a constant. It follows that the number of sub-types is a constant, for each type. Therefore, the overall number of types is a constant. Thus, Step~\ref{step:subTypes} increases the number of types by a constant, and because there are $O(1)$ iterations, the overall running time is polynomial. 

(ii) We show that the increase in the number of bins due to padding types of small cardinality (Step~\ref{step:padding}) is bounded by a constant . The number of types in each iteration is a constant, and the number of iterations is a constant. To each type we add (at most) a constant number of bins. Thus, the total number of bins added is a constant. 

(iii) Recall that we apply linear shifting to each $t(g_{opt})$-significant group, for each type $t$, in each iteration (Step~\ref{step:Representative}). Thus, we discard from $G^t_{i_j}$ at most $\eps^3 |t|$ items (the largest size class).
This holds for any type $t$ and any iteration in which $G_{i_j}$ is  $t(g_{opt})$-significant. 
Thus, in $\alpha$ iterations we discard from $G^t_{i_j}$ at most $\eps |t|$ items. We conclude that the number of items discarded from $G_{i_j}$ in all types in all iterations is at most $\eps OPT$. 

We now bound the total size of discarded items.
For each $b \in E_0$ $f^0_b \leq \eps$. Therefore, $\sum_{b \in E_0} f_b^0 \leq \eps OPT$. Thus, the total size discarded in RecursiveEnum is at most $\eps OPT$, because the linear shifting is done solely on items that are guessed to be packed in bins in $E_0$. The latter is true if our guess corresponds to an optimal solution, as we prove below.  

(iv) For each type $t$ we enumerate over all possible selections of its $t(g_{opt})$-significant groups in Step~\ref{step:t-significantGroups}. One of the guesses corresponds to an optimal solution. Let $G_{i_j}$ be a group that is $t$-significant in an optimal solution $g_{opt}$. Then, applying linear shifting and discarding the largest size class will preserve the feasibility of the packing of the remaining size classes. Hence, our use of shifting preserves the optimal guess. We prove next that our greedy algorithm for choosing the items to each size class 
yields a feasible subset,
given an optimal guess for the groups and representative for each size class. We choose for each size class the largest remaining items that fit into this class. Thus, if the guess is optimal, we are guaranteed to be able to fill the other classes as well. Assume, towards contradiction, that there is a size class $C$ that cannot be filled. Therefore, its items are already taken by previous size classes $C_1, \ldots, C_R$. For a size class $C'$, let $s(C')$ be the size of the largest item in $C'$.
By the ordering of the algorithm, $s(C) \leq s(C_r)$ for all $r \in [R]$.
Thus, there must be items small enough  left for $C$.   \qed

\posA
\noindent{\bf Proof of Lemma~\ref{lem:eviction}:}
We construct a partition $I_t, t \in T$ satisfying conditions $(i)$-$(iii)$ 
for all $t \in T$. Let $g_{opt}$ be a guess corresponding to an optimal solution, i.e., $I_{\alpha}$ can be added to the bins resulting in an optimal solution. Let $I_{t(g_{opt})} \subseteq I_\alpha, t \in T\setminus\{t'\}$ be a packing of items in $I_\alpha$ in the type that yields $g_{opt}$. Recall that in Step~\ref{step:ExtraType} of Algorithm~\ref{Alg:PackSmallItems} we add $\eps OPT$ empty bins as a new type $t'$. The partition is defined as follows. $I_{t'}$ consists of  all $\ell \in I_{\alpha}$ such that $\ell$ is a $t$-non-negligible item from a $t(g_{opt})$-insignificant group in some iteration $h \in [\alpha]$, where $t \in T\setminus\{t'\}$.
Also, for $t \in T \setminus\{t'\}$, $I_t = I_{t(g_{opt}) \setminus I_{t'}}$.

We first prove that $I_{t'}$ satisfies conditions $(i)$-$(iii)$. There are $\alpha$ iterations. For some group $G_i$, $1 \leq i \leq n$, the number of non-negligible items from $G_i$ in types $t$ for which $G_i$ is $t(g_{opt})$-insignificant is at most $\sum_{t \in T}\eps^4 |t| = \eps^4 OPT$.  Thus, after $\alpha$ iterations the number of bins of type $t'$ required to avoid conflicts in $G_i$ is at most $\eps OPT$ and (i) holds for $I_{t'}$. The bins in $t'$ are of capacity $1$ and for all items in $\ell \in I_{t'}: s_{\ell} \leq \eps$ because $I_{\alpha}$ contains only small items. Therefore, condition $(ii)$ also holds for $I_{t'}$. Finally, the total size of these items is at most $\eps OPT$, as these items are packed in bins that are in $E_0$ and condition $(iii)$ holds.\footnote{Recall that we do not use guessing on bins in $D_0$.} 

Next, we prove that $I_{t}, t \in T\setminus\{t'\}$ satisfies conditions $(i)$-$(iii)$.
If condition $(i)$ does not hold for $I_t$ for some group $G_i$, then $I_{t(g_{opt})}$ cannot be packed into bins of type $t$ without causing a conflict (there must be a bin with two items from the same group by the pigeonhole principle) and this is a contradiction. In addition, there must be a partition in which condition $(iii)$ holds, or some $I_{t(g_{opt})}$ cannot be packed in bins of type $t$ contradicting the optimality of the packing up to this point, and the correct guess of $I_{t(g_{opt})}$.

Now, consider condition $(ii)$. The items in $I_t$ are either packed according to $g_{opt}$ such that they are $t$-negligible after iteration $\alpha$, or $t \subseteq E_{\alpha}$. First, for all types $t \subseteq D_{\alpha}$, by the definition of $D_{\alpha}$, $I_{t(g_{opt})}$ contains only $t$-negligible items; thus, condition $(ii)$ holds in this case. 

Note that condition $(ii)$ does not necessarily hold if we omit Step~\ref{step:eviction} in 
Algorithm~\ref{Alg:PackSmallItems}. We show that this step is completed successfully. Note that any bin $b \in E_\alpha$ contains at least $\alpha$ items, since in each iteration at least one item is added to $b$ (else, $b$ is moved to $D_h$ in some iteration $1 \leq h < \alpha$, and $b \notin E_\alpha$). Now, if we omit from bin $b$ and item $\ell$ that was added to $b$ in iteration $h \leq \alpha-3$ then the free space in $b$ increases at least by $\frac{f_b^{\alpha}}{\eps}$. Indeed, item $\ell$ was $t$-non-negligible  in iteration $h \leq alpha -3$; thus, $s_\ell \geq \eps^2 f_b^{\alpha-3}$. For $t \in E_\alpha$, it holds that $f_b^{\alpha-3} \geq \frac{f_b^{\alpha}}{\eps^3}$. Hence, $s_\ell \geq \eps^2 \cdot 
\frac{f_b^{\alpha}}{\eps^3}  = \frac{f_b^{\alpha}}{\eps}$.

We now show that we can evict in the process at most $\eps OPT$ items from each group $G_i$, $1 \leq i \leq n$. Recall that for each $t \in E_\alpha$ the bins of type $t$ contain items from distinct groups. Since there are at least $\alpha$ items in each bin, and $ \alpha > \frac{1}{\eps} +4$, the bins of type $t$  can be partitioned into $1/\eps$ subsets, $B_1^t, \ldots, B_{1/\eps}^t$, each consists of $\eps |t|$ bins. In the worst case,
all the bins of type $t$ contain items from the same set of groups,  $G_{K_1}, \ldots , G_{k_R}$, where $R \geq 1/\eps$. We can now evict from each bin in $B_r^t$, an item from $G_{k_r}$, $1 \leq r \leq 1/\eps$. Thus, we omit from each group at most $\eps |t| \leq \eps OPT$ items.

Let $I_{b(g_{opt})}$ be the set of items packed in bin $b \in t \subseteq E_{\alpha}$ in $g_{opt}$. For any $\ell \in I_{b(g_{opt})}: s_{\ell} \leq f(t)$ because $g_{opt}$ is a feasible packing. After Step~\ref{step:eviction}, the capacity available in each bin in $t$ and in particular in $b$ is at least $\frac{f(t)}{\eps}$. Hence, condition $(ii)$ holds for each $t \subseteq E_{\alpha}$ after the eviction phase. \qed

\comment{
\noindent{\bf Proof of Lemma~\ref{lem:movement}:}
    Let $P'$ be the polytope obtained by removing the following constraint from $P$: all constraints (2), and all constraints (4) except for the constraint in (4) corresponding to Group $G_j$. Let $A$ be the coefficient matrix of the constraints defining $P'$. We first prove that $A$ is totally unimodular.
    
    By Ghouila-Houri characterization \cite{ghouila1962caracterisation} a matrix $H$ is totally unimodular if and only if its columns can be partitioned into two disjoint sets $J_1, J_2$ such that for every row $i$: $$(8) \: \: \: |\sum_{j \in J_1} H_{i,j} -  \sum_{j \in J_2} H_{i,j} | \leq 1$$ Assume, w.l.o.g, that $T = \{t_1, \ldots, t_{|T|}\}$ such that $\forall i > j: \: f(t_i) \geq f(t_j)$ and that the items in $G_j$ are all integers in $[v,w], v,w \in \mathbb{N}$.\footnote{This can be assumed because we can use a permutation on the items identifiers.} 
    
    Define $J_1$ to be all columns corresponding to $(i,t_q)$ such that $i = 0 \: mod(2) ,q = 0 \: mod(2)$, and $i = 1 \: mod(2) ,q = 0 \: mod(2)$ and all other columns to be in $J_2$. We now prove for each row (defined by a constraint of $P'$) that (8) holds. First, it is easy to see that all non zero entries in $A$ are ones. For constraints in (1), they refer to rows with exactly one non zero entry, so (8) holds. As was said, constraints (2) are not considered in $P'$. In (3), for a fixed $i \in I$, (8) holds since half (+-1) of the entries are in $J_1$ and the rest are in $J_2$. For the constraint in (4) corresponding to $G_j$, the only non zero entries are for sequential natural numbers, and because they are all correspond to the same type $t$, half of them (+-1) are in $J_1$ by the definition of $J_1$. 
    
    One of the properties of totally unimodular matrices is that the vertices of the corresponding polytope are integral \cite{wolsey1999integer}. Thus, $x$ cannot be a vertex of $P'$, because $G_j$ is fractional. Therefore, there is a vector $m^j \neq 0$ such that $x+m^j \in P', x-m^j \in P'$. By stating the constraints of $P'$ for both $x+m^j, x-m^j$ and then subtracting the constraints of $x+m^j$ from the constraints of $x-m^j$, it follows that  (5), (6), (7) hold for $m^j$, and thus, $m^j$ is a movement for $G_j$. \qed

\posA
\noindent{\bf Proof of Lemma~\ref{lem:FewFractionalGroups}:}

    Assume towards a contradiction that there are $|T|+1$ fractional groups. W.l.o.g, assume that these groups are $G_1, \ldots, G_{|T|+1}$. By Lemma~\ref{lem:movement} $G_j, j \in [|T|+1]$ has a movement $m^j \neq 0$. Consider the following set of equalities over $\lambda_1, \ldots, \lambda_{|T|+1}$:
    
    $$(9) \: \: \forall t \in T: \:\: \sum_{j=1}^{|T|+1}\lambda_j \sum_{i \in G_j} m^j_{i,t} \cdot s_i = 0$$
    
    These are $|T|$ homogeneous equalities in $|T|+1$ variables. Thus, there exist $\lambda_1, \ldots, \lambda_{|T|+1}$, not all zeros, such that (9) holds. Using scaling we can assume that $x_{i,t}-\lambda_j m^j_{i,t} \in [0,1], \forall i \in I_0, t\in T, j \in [|T|+1]$. It can now be easily verified, by the properties of a movement, that $x-\sum_{j=1}^{|T|+1}\lambda_j m^j \in P, x+\sum_{j=1}^{|T|+1}\lambda_j m^j \in P$. Also, $\sum_{j=1}^{|T|+1}\lambda_j m^j \neq 0$. This is a contradiction to the fact that $x$ is a vertex of $P$. \qed

\posA
\noindent{\bf Proof of Lemma~\ref{lem:FewFractionalItems}:}
    Let $P_j$ be the following polytope. 
    
    $$P_j = \{y \in [0,1]^{G_j \times T}\} \: \: \: s.t:$$

$$ \: \: \:\forall i \in G_j, t \in T: s_i > \eps f(t) \rightarrow y_{i,t} = 0$$

$$ \: \: \:\forall t \in T: \sum_{i \in G_j} y_{i,t} s_i+\sum_{i \in I_0 \setminus G_j} x_{i,t} s_i \leq f(t) |t|$$

$$ \: \: \:\forall i \in G_j: \sum_{t \in T} y_{i,t} = 1$$

$$ \: \: \:\forall t \in T: \sum_{i \in G_j} y_{i,t} \leq |t|-|(I'_t \setminus I_L) \cap G_j|$$

It holds that $y \in [0,1]^{G_j \times T}$, defined by $y_{i,t} = x_{i,t}, i \in G_j, t \in T$ is a vertex of $P_j$. $P_j$ is a 2-dimensional Vector Bin Packing polytope, for which it is known that its vertices are with at most $2|T|$ fractional entries. \qed

}
\comment{
\posA
\noindent{\bf Proof of Lemma~\ref{lem:5ecoupledItems}:}
The number of items discarded from each group is at most $\eps^{k+2} \cdot OPT$, since all groups are small.
 Assume that the total size of these items is strictly larger than $\eps OPT$. Since each discarded item is {\em coupled} with a large conflicting item from the same group, whose size is at least $1/\eps$ times larger (recall that the medium items are discarded), this implies that the total size of large conflicting items is greater than $OPT$. Contradiction. \qed    
}

\posA
\noindent{\bf Proof of Lemma~\ref{lem:greedyInType}:}
    We prove the claim by induction on $|t|$. For the base case, let $|t| = 1$. Since there is only one bin in $t$, $|G_j^t| \leq 1$ by Condition (i). Thus, $I_t$ can be packed in $t$ without conflicts. Also, the total size of all items is at most the free capacity of the bin, by Condition (iii). Hence, we can pack all items feasibly.

For the induction step, assume the claim holds for $|t|-1$ bins. Now, suppose that 
there is a type $t$ with $|t|$ bins. Recall that GreedyPack initially assigns items to the bin of maximum total size. Now, consider two cases.

$(1)$ After the packing of the first bin, the total size of items from $I_t$ in this bin is strictly less than $(1-\delta)f(t)$. Then, we prove in this case that the first bin contains the largest item left in each small group. In any two consecutive attempts in GreedyPack of packing the first bin, the difference in the total size packed is by the size of one item. By Condition (ii), this difference is at most $\delta f(t)$. Therefore, if packing the largest item from each group overflows, GreedyPack would continue until the first iteration that it finds a feasible packing. Observe the last attempt, and assume towards a contradiction that it is not the largest item of each group. Then, the previous attempt overflowed, therefore, the current attempt must be with total size at least $(1-\delta)f(t)$, which is a contradiction to $(1)$.  

The remaining items can be feasibly packed in the remaining $|t|-1$ bins, since taking an arbitrary remaining item from each group cannot overflow because the packing of the first bin does not overflow and contains the largest item from each group.

$(2)$ The first bin is packed with total size at least $(1-{\delta})f(t)$. We prove that all conditions hold for applying the induction hypothesis for the last $|t|-1$ bins. First, GreedyPack packs in the first bin an item from each group. Assume towards a contradiction that GreedyPack fails in doing so in the first bin. Therefore, the sum of the smallest item of each group in $I_t$ overflows from $f(t)$. Let $s_{j,t}^{min}$ be the size of the smallest item in group $G_j^t$. Since there are exactly $|t|$ item from each group in $I_t$, it follows that $S(I_t) >   |t| \sum_{G_j^t} s_{j,t}^{min} > |t|  f(t)$ in contradiction to Condition (iii). Thus, we are left with $|t|-1$ item from each group after the packing of the first bin. 

Second, Condition (ii) is trivially satisfied for any number of bins in $t$. Third, by Case (2) we pack in the first bin total size of at least $(1-{\delta})f(t)$. Therefore, the residual size of $I_t$ after the packing of the first bin is at most $(1-\delta) f(t) |t| - (1-{\delta})f(t) \leq (1-\delta) f(t) (|t|-1)$ and Condition (iii) holds. 

 Hence, by the induction hypothesis, the remaining items in $I_t$ can be packed in bins $2, \ldots, |t|$. \qed

\posA
\noindent{\bf Proof of Lemma~\ref{lem:greedySummary}:}
    Given $g_{opt}$, we can find a partition $I_t, t \in T$ of the remaining items, not violating constraints 1,2,3. Furthermore, in Step~\ref{step:discardCapacity} we add $2\cdot\eps |t|$ extra bins to $t$, thus $S(I_t) \leq (1-\eps)|t|$ ($|t|$ refers to the cardinality of $t$ after Step~\ref{step:discardCapacity}). Therefore, conditions (i), (ii) for Lemma~\ref{lem:greedyInType} hold for each type $t \in T$ with parameter $\delta = \eps$ (except for maybe too many items from some group because of large items from small groups, which GreedyPack discards and by Lemma~\ref{lem:5ecoupledItems} this results with a small number of extra bins). For each $t \in T$ GreedyPack discards items of total size at least $\eps |t|f(t)$. Therefore, the remaining items in $I_t$ are with total size at most $(1-\eps)|t|f(t)$ because of constraint (1). We conclude that besides the discarded items, all items $I_t$ are packed feasibly in $t$. The discarded items are at most $\eps |t|$ from each group in each type $t$ by GreedyPack. Thus, combined, the overall total size discarded is $O(\eps)OPT$ since these are small items, and $O(\eps)OPT$ items are discarded from each group. \qed 
    
 \posA   
\noindent{\bf Proof of Lemma~\ref{lem:discardFinal}:}
We prove the claim by deriving a bound on the number of extra bins required for packing the items discarded throughout the execution of the scheme. During the scheme, we discard $O(\eps) OPT$ items from each group, and discard a total size $O(\eps) OPT$ overall. These items are packed by a $2$-approximation algorithm (BalancedColoring), and thus, the number of extra bins needed for packing these items is $O(\eps)OPT+1$. \qed

\comment{
\noindent{\bf Proof of Lemma~\ref{lem:discardFinal}:}
We prove the claim by deriving a bound on the number of extra bins required for packing the items discarded throughout the execution of the scheme. During the scheme, we discard $O(\eps) OPT$ items from each group, and discard a total size $O(\eps) OPT$ overall. These items are packed by a $2$-approximation algorithm (BalancedColoring), and thus, the number of extra bins needed for packing these items is $O(\eps)OPT+1$. 
\begin{enumerate}
\item \label{1}
Rounding the sizes of large and medium items from large groups. By the proof of Lemma~\ref{lem:shiftingCanBeUsedForI}, the total size of items discarded due to shifting is at most ${\eps \cdot OPT}$, and at most $\eps ^{2k+4} OPT$ items from each large groups are discarded.

\item \label{smallshift}
Rounding the sizes of large items from small groups. The total size of items discarded due to the shifting is at most $\lfloor 2\eps \cdot OPT \rfloor$, and at most $\eps ^{k+2} OPT$ items are discarded from each small group, as this is the maximum number of large items in a small group. Also, we discard the items in the last size class, i.e., at most $\lfloor 2\eps \cdot OPT \rfloor$ items of total size at most $\lfloor 2\eps \cdot OPT \rfloor$.

\item \label{4} Packing medium items from small groups (Section~\ref{subsection:Medium Items of Small Groups}). By Lemma~\ref{lem:k_val},
 the total size of these items is at most $\eps^2 \cdot OPT$. Since the items belong to small groups, their total number in each group is at most $\eps^{k+2} OPT -1$.
 
 \item \label{2} Rounding the sizes of small items in large groups (before packing these items by algorithm EnumGroups).
 We discard 
  $\eps^{2k+4} OPT$ items  from each large group. By Lemma~\ref{lem:few_large_groups}, the number of large groups is $L \leq \frac{1}{\eps^{2k+3}}$.
 The size of each small item is at most $\varepsilon^{k+1}$. Therefore, the total size of discarded items is at most 
 $\eps^{2k+4} OPT \cdot L \cdot \eps^{k+1} \leq \eps \cdot OPT$. 
 
\item \label{3} Items discarded in algorithm SmallGroups.
\begin{enumerate}
\item 
In Step~\ref{step:discardx=1}, the total size of discarded items is at most $\eps \cdot OPT+ \eps^{k+1}$. The number of items discarded from each group is at most $x = 1$. 

\item
In Step~\ref{step:discard1}, the total size of discarded items is at most $2 \eps \cdot OPT$.
The number of items discarded from each group is at most $m= \lfloor \eps OPT \rfloor$.

\item In Step~\ref{step:discardregular}, the total size of discarded items is at most $V_{m+1} \leq m$. 
The number of items discarded from each group is exactly $m+1$.

\item
In Step~\ref{step:linearGrouping} we apply shifting to the items in special groups. By Lemma~\ref{lem:disgroups}, the number of these groups satisfies
$|A| \leq \frac{1}{\eps^2}$. At most $Q= \lfloor \eps^{3} \cdot OPT \rfloor$ small items are discarded from each group. Hence, the total size of these items is
at most $\frac{1}{\eps^2} \cdot \eps^{k+1} \lfloor \eps^{3} \cdot OPT \rfloor \leq \eps \cdot OPT$. The number of items discarded from each group is at most $m$.

\item By Lemma~\ref{lem:5ecoupledItems} The total size of discarded items in Step~\ref{step:discardsmalloverlarge} in Algorithm~\ref{Alg:greedyPack} is at most $\eps OPT$ and at most $\eps^{k+2} \cdot OPT$ items are discarded from each group.

\end{enumerate}
\end{enumerate}

By the above discussion, the total size of discarded items is at most $\eps OPT+4\eps OPT+\eps^2 OPT+\eps OPT+3\eps OPT \leq 10 \eps OPT$ (In \ref{1},~\ref{smallshift}, 
\ref{4}, \ref{2}, \ref{3}, respectively, where we take the worst case in \ref{3}). The maximum number of discarded items from each group is at most $\eps OPT+\eps OPT+\eps OPT+2m+1+\eps^{k+2}OPT \leq 6\eps OPT+1$ (in Steps~\ref{step:discardregular}, \ref{step:linearGrouping} of algorithm SmallGroups, in addition to large items discarded due to shifting). Hence, we can use BalancedColoring to pack the  
items discarded throughout the execution of the scheme in at most $max\{ \lceil 2 \cdot 10\eps OPT \rceil, \lceil 8\eps OPT+6\eps OPT+1 \rceil \} \leq 20\eps OPT+1$ bins. The inequality holds since $OPT > \frac{1}{\eps}$. 
\qed
}

\posA
\noindent{\bf Proof of Theorem~\ref{theorem:APTAS}:}
The feasibility of the packing follows from the way algorithms RecursiveEnum, GreedyPack and SmallGroups assign items to the bins.
	We now bound the total number of bins used by the scheme. As shown in the proof of Lemma~\ref{lem:discardFinal}, given the parameter $\eps \in (0,1)$, the total number of extra bins used for packing the medium items from small groups and the discarded items is at most $O( \eps) OPT +1$. By scaling $\eps$ by a constant (that dos not depend on $\eps$), we have that the total number of bins used by the scheme is $ALG(I) \leq (1+\eps)OPT +1$. As shown above, each step of the scheme has running time polynomial in $N$.

	\comment{
	We note that each item is packed by exactly one way, using the correctness of one of the following lemmas: large and medium items from large groups are packed feasibly by Lemma~\ref{lem:dw1}, the Large items from small groups are packed feasibly by Lemma~\ref{lem:dw2} and the medium items of small groups are packed in extra bins by Lemma~\ref{lem:discardFinal}. Moreover, The small items of large groups are packed feasibly by Theorem~\ref{lem:larggroupCorrectness} and the small items of small groups are packed feasibly by Theorem~\ref{lem:smallgroup}. 
	
	The $OPT$ bins $b_1, \ldots, b_{OPT}$ which we use through all steps are packed with no overflows, as we explicitly prove in previous steps. Combining the number of additional bins that we use through all steps of the packing: the medium items and the discarded items, we need overall $12 \varepsilon OPT+1$ new bins at most by Lemma~\ref{lem:discardFinal}. Assuming that $0 < \varepsilon < 1$ is the desired parameter for the algorithm. Let there be $\varepsilon' = \frac{\varepsilon}{12}$. We use $\varepsilon'$ as the parameter for the algorithm, thus the number of bins that we use is $ALG(I) \leq OPT+12\varepsilon' \cdot OPT+1 = OPT+12\cdot \frac{\varepsilon}{12} \cdot OPT+1 = (1+\varepsilon)OPT +1$. We show for each step that its running time is polynomial. Since there are $O(1)$ steps, the overall running time is polynomial in $N$.
}
\qed

\section{Proof of Theorem \ref{thm:partitionPhase}}

\newcommand{\bm}{\bar{m}}
\newcommand{\by}{\bar{y}}
\newcommand{\bx}{\bar{x}}
\newcommand{\bgam}{\bar{\gamma}}
\newcommand{\bb}{\bar{b}}
\newcommand{\ariel}[1]{{\color{red} (#1)}}

We prove the theorem using the next lemmas. Let $\bx \in P$ be a vertex of $P$. For a group $G_j, j \in [n]$, a {\em movement} is a vector $\bm^j \in \mathbb{R}^{I_\alpha \times T}$ such that:

\begin{align}
&\forall \ell\in I_{\alpha}, t \in T, \ell \notin G_j  \textnormal{ or } x_{\ell, t} \in \{0,1\} : &~~~~~& \bm^j_{\ell,t} = 0
\label{eq:move_supp}\\
&\forall \ell \in I_{\alpha}: &~~~~~&\sum_{t \in T} \bm^j_{\ell,t} = 0
\label{eq:move_select}\\
&\forall t \in T \:s.t\: \sum_{\ell \in G_j} x_{\ell,t} = L_{t,j}: &~~~~~&  \sum_{\ell \in G_j} \bm^j_{\ell,t} = 0
\label{eq:move_group}
\end{align}
where 
$L_{t,j}= |t| - |(I'_t)\setminus I_L) \cap G_j|$ for $t\in T$ and $j\in \{1,\ldots, N\}$. Additionally, define the 
sets $X_j = \{ (\ell,t)\in G_j\times T~|~ s_{\ell}>\eps \cdot f(t)\}$ for $j\in [n]$. 

We say that group $G_j$ is {\em fractional} if there are $\ell \in G_j$ and $t \in T$ such that $\bx_{\ell,t} \in (0,1)$.  

\begin{lemma}
	\label{lem:movement}
	If Group $G_j$ is fractional then $G_j$ has a movement $\bm^j \neq 0$.
\end{lemma}
The proof of Lemma~\ref{lem:movement} utilizes properties of {\em totally unimodular} matrices. A matrix $A$ is totally unimodular if every square submatrix of $A$ has a determinant $1$, $-1$ or $0$. If $A\in \mathbb{R}^{n\times m}$ is totally unimodular and $b\in \mathbb{Z}^{m}$ is an integral vector, it holds that the vertices of the polytope  $P_A =\{\bx \in \mathbb{R}_{\geq 0}^n~|~A \bx \leq \bb\}$ are integral. That is,  if $\bx\in P_A$ is a vertex  of $P_A$ then $\bx\in \mathbb{Z}^n$ \cite{HK56}.  We use the following criteria for total unimodularity, which is a simplified version of a theorem from~\cite{HK56}.
\begin{lemma}
	\label{lem:unimodular}
	Let $A\in \mathbb{R}^{n\times m}$ be a matrix which satisfies the following properties.
	\begin{itemize}
		\item All the entries of $A$ are  in  $\{-1,1,0\}$.
		\item Every column of $A$ has up to two non-zero entries.
		\item If a column of $A$ has two non-zero entries, then these entries have opposite signs.
		\end{itemize}
	\end{lemma}  
Then $A$ is totally unimodular.

\posA
\noindent{\bf Proof of Lemma~\ref{lem:movement}:}
We show the existence of the movement $\bm^j$ using the polytope $P_j$ defined as follows. 
	\begin{equation}
		\label{eq:Pj_def}
	P_j= \left\{ \by\in \mathbb{R}_{\geq 0}^{G_j \times T} ~\middle |~ \begin{array}{lcc}
		\displaystyle \forall (\ell,t)\in X_j&:& \displaystyle\by_{\ell,t}\leq 0\\
	 \displaystyle	\forall \ell\in G_j&:& \displaystyle  \sum_{t\in T \textnormal{ s.t. } (\ell,t)\notin X_j} -\by_{\ell,t} \leq -1\\
	 \displaystyle \forall t\in T&:& 
	 \displaystyle  \sum_{\ell\in G_j \textnormal{ s.t. } (\ell,t)\notin X_j } \by_{\ell,t} \leq L_{t,j}
	\end{array}\right\}.
	\end{equation}

Define  a vector $\by^*\in \mathbb{R}^{G_j \times T}$ by $\by^*_{\ell,t}=\bx_{\ell,t}$. It follows from the definition of $P$ that $\by^*\in P_j$. We can represent the inequalities in \eqref{eq:Pj_def} using a matrix notation as  $P_j =\left\{\by \in \mathbb{R}^{G_j \times T}_{\geq 0} ~\middle|~ A\by \leq \bb \right\}$. It follows that $A$ contains only entries in $\{-1,0,1\}$ and the entries in $\bb$ are all integral.
Furthermore, every column of $A$ contains at most $2$ non-zero entries, and if there are two non-zero entries in a column then they are of a different sign.
By Lemma~\ref{lem:unimodular}, it follows that $A$ is totally unimodular. It thus holds that all the vertices of the polytope $P_j$ are integral. As $\by^*\in P_j$ is non-integral (since $G_j$ is fractional), it follows that $\by^*$ is not a vertex of $P_j$. Hence, there is a vector $\bgam\in \mathbb{R}^{G_j \times T}$, $\bgam \neq 0$  such that $\by^*+\bgam, \by^*-\bgam\in P_j$. 

We define $\bm^j\in \mathbb{R}^{I_{\alpha}\times T}$ by $\bm^j_{\ell,t}=\bgam_{\ell,t}$ for $(\ell,t)\in G_j \times T$ and $\bm^j_{\ell,t} =0$ otherwise. Clearly, $\bm^j\neq 0$ as $\bgam\neq 0$. 
Observe that for $\ell\in I_{\alpha}\setminus G_j$ and $t\in T$ it holds that $\bm^j_{\ell,t}=0$ by definition.
 For $\ell\in G_j$ and $t\in T$ such that $\bx_{\ell,t} =0$, as $\by^*_{\ell,t}+\bgam_{\ell,t}, \by^{*}_{\ell,t} -\bgam_{\ell,t}\geq 0$ and $\by^*_{\ell,t}=\bx_{\ell,t}=0$,  it follows that $\bgam_{\ell,t}=0$. 
 For $\ell \in G_j$ and $t\in T$ such that $\bx_{\ell,t} =1$, it follows that $\by^*_{\ell,t'}=\bx_{\ell,t'} =0$ for  $t'\in T\setminus \{t\}$ by the definition of $P$ as well as $(\ell,t)\not\in X_j$. Thus, by the previous argument, we have $\bgam_{\ell,t'}=0$ for every $t'\in T\setminus \{t\}$. 
 Therefore, 
 \begin{equation}
 	\label{eq:gam_supp_first}
-1 -\bgam_{\ell,t}=- \sum_{t'\in T \textnormal{ s.t } (\ell,t')\not\in X_j} \left(\by_{\ell,t'}  +\bgam_{\ell,t'}\right) \leq -1,
\end{equation} 
where the last inequality is due to $ \by^*+\bgam \in P_j$. 
Similarly, as $\by^*-\bgam \in P_j$, we have
\begin{equation}
	 	\label{eq:gam_supp_second}
	 	 -1 +\bgam_{\ell,t}
=-\sum_{t'\in T \textnormal{ s.t } (\ell,t)\not\in X_j} \left(\by_{\ell,t'}  -\bgam_{\ell,t'}\right) \leq -1,
\end{equation} 
By~\eqref{eq:gam_supp_first} and \eqref{eq:gam_supp_second} we have $\bgam_{\ell,t}=0$. Overall, we have that $\bm^j$ satisfies \eqref{eq:move_supp}. 

For any $\ell\in I_{\alpha}\setminus G_j$ it holds that $\sum_{t\in T} \bm^j_{\ell,t} = 0 $. For $\ell\in G_j$, since  $\by^*+\bgam\in P_j$ we have
\begin{equation}
	\label{eq:t_sum_first}
-1 -\sum_{t\in T} \bgam_{\ell,t} = \sum_{t\in T}-(\bx_{\ell,t} +\bgam_{\ell,t}) =  
 \sum_{t\in T~\textnormal{ s.t. }(\ell,t)\notin X_j}-(\by^*_{\ell,t} +\bgam_{\ell,t}) \leq -1,
 \end{equation}
 where the first equality holds since $\bx\in P$ and the second equality uses $\by^*_{\ell,t} +\bgam_{\ell,t}=0$ for $(\ell,t)\in X_j$. Similarly, since $\by^*-\bgam\in P_j$, we have 
 \begin{equation}
 		\label{eq:t_sum_second}
 -1 +\sum_{t\in T} \bgam_{\ell,t} = \sum_{t\in T}-(\bx_{\ell,t} -\bgam_{\ell,t}) =  
 \sum_{t\in T~\textnormal{ s.t. }(\ell,t)\notin X_j}-(\by^*_{\ell,t} -\bgam_{\ell,t}) \leq -1.
 \end{equation}
By \eqref{eq:t_sum_first} and \eqref{eq:t_sum_second} we have 
  $\sum_{t\in T} \bm^{j}_{\ell,t}=\sum_{t\in T} \bgam_{\ell,t}=0$. Thus, $\bm^j$ satisfies \eqref{eq:move_select}.
  
  Finally, let $t\in T$ such that $\sum_{\ell\in G_j} \bx_{\ell,t} =L_{t,j}$ . As before,
  \begin{equation}
  	\label{eq:gam_group_sum_first}
  	L_{t,j} +\sum_{\ell\in G_j} \bgam_{\ell,t}= \sum_{\ell\in G_j} \left(\by^*_{\ell,t} +\bgam_{\ell,t}\right)
  	= \sum_{\ell\in G_j \textnormal{ s.t. } (\ell,t)\notin X_j} \left(\by^*_{\ell,t} +\bgam_{\ell,t}\right)  \leq L_{t,j}, 
  	\end{equation}
  where the inequality follows from $\by^*+\bgam \in P_j$. Using a similar argument,
  \begin{equation}
  	  	\label{eq:gam_group_sum_second}L_{t,j} -\sum_{\ell\in G_j}\bgam_{\ell,t}= \sum_{\ell \in G_j} \left(\by^*_{\ell,t} -\bgam_{\ell,t}\right) =
  	  	\sum_{\ell\in G_j \textnormal{ s.t. } (\ell,t)\notin X_j} \left(\by^*_{\ell,t} -\bgam_{\ell,t}\right) 
  	  	 \leq L_{t,j}.
  \end{equation}
By \eqref{eq:gam_group_sum_first} and \eqref{eq:gam_group_sum_second}, we have $\sum_{\ell\in G_j} \bm^j_{\ell,t} = \sum_{\ell\in G_j} \bgam^j_{\ell,t} =0$. Thus, $\bm^j$ satisfies \eqref{eq:move_group}. Overall, we show that $\bm^j\neq 0$ is a movement of $G_j$.  
\qed

\begin{lemma}
	\label{lem:FewFractionalGroups}
	There are at most $|T|$ fractional groups.
\end{lemma}

\begin{proof}
	Assume towards a contradiction that there are $|T|+1$ fractional groups. W.l.o.g, assume that these groups are $G_1, \ldots, G_{|T|+1}$. By Lemma~\ref{lem:movement}, $G_j, j \in [|T|+1]$ has a movement $\bm^j \neq 0$. Consider the following set of equalities over $\lambda_1, \ldots, \lambda_{|T|+1}$:
	\begin{equation}
		\label{eq:homegeneous} \forall t \in T: ~~~~ \sum_{j=1}^{|T|+1}\lambda_j \sum_{\ell \in I_{\alpha}} \bm^j_{\ell,t} \cdot s_\ell = 0.
	\end{equation}
	These are $|T|$ homogeneous linear equalities in $|T|+1$ variables. Thus, there exist $\lambda_1, \ldots, \lambda_{|T|+1}$, not all zeros, for which  \eqref{eq:homegeneous} holds.
	
	Define 
	$$K_1 = \min_{(\ell,t)\in I_{\alpha}\times T}\min \left\{ \bx_{\ell,t}, 1-\bx_{\ell,t} ~\middle|~0<\bx_{\ell,t}<1\right\}, $$ 
	$$	K_2= \min \left\{
	L_{t,j} - \sum_{\ell \in G_j} \bx_{\ell,t}~\middle |~1\leq j\leq |T|+1,\sum_{\ell \in G_j} \bx_{\ell,t} < L_{t,j}  \right\},
	$$  
	and $K= \min\{K_1,K_2\}$.  Additionally, define 
	$$M=\max\left\{ |\bm^{j}_{\ell,t}|~|~1 \leq j \leq |T|+1,~ (\ell,t)\in I_{\alpha} \times T\right\}.$$
	Observe that $M,K>0$. 
	Using a scaling argument, we may assume that $\lambda_j\leq \frac{K}{|I_{\alpha}| \cdot M}$ for any $1\leq j\leq |T|+1$. 
	
	Define $\bx^+= \bx + \sum_{j=1}^{|T|+1} \lambda_j \cdot \bm^j$. In the following we show that $\bx^+\in P$. 
	
	For every $\ell\in I_{\alpha}\setminus (G_1\cup \ldots \cup G_{|T|+1})$  and $t\in T$ it holds that $\bm^{1}_{\ell,t}= \ldots = \bm^{|T|+1}_{\ell,t} = 0$ due to \eqref{eq:move_supp}. Thus,
	$\bx^+_{\ell,t} = \bx_{\ell,t}\in  \{ 0,1 \}$. For $\ell \in G_{j}$ with $1\leq j\leq |T|+1$ and $t\in T$, using \eqref{eq:move_supp} once more we have $\bx^+_{\ell,t} =\bx_{\ell,t} + \lambda_j \cdot \bm^j_{\ell,t}$. Following the definitions  of $K$ and $M$, we have 
	$$ 0\leq K -M \cdot \frac{K}{|I_{\alpha}|\cdot M} \leq \bx_{\ell,t} + \lambda_j \cdot \bm^j_{\ell,t}\leq 1-K + M \cdot \frac{K}{|I_{\alpha}|\cdot M} \leq 1.$$
	Thus, $\bx^+\in [0,1]^{I_{\alpha} \times T}$. 	
	
	For every $(\ell,t)\in T$ such that $s_{\ell} > \eps f(t)$, it holds that $\bx_{\ell,t}=0$. Thus by \eqref{eq:move_supp} we have $\bm^{j}_{\ell,t} =0$ for all $1\leq j \leq |T|+1$. Therefore, $\bx^+_{\ell,t} =0$. 
	
	For every $t\in T$ it holds that 
	$$
	\sum_{\ell \in I_{\alpha}} \bx^+_{\ell,t} \cdot s_{\ell}  = 
	\sum_{\ell\in I_{\alpha}} \bx_{\ell,t}\cdot s_{\ell}+ \sum_{j=1}^{|T|+1}\lambda_j \sum_{\ell\in I_{\alpha}}
	\bm^j_{\ell,t} s_{\ell} = \sum_{\ell\in I_{\alpha}} \bx_{\ell,t}\cdot s_{\ell}  \leq f(t) |t|,$$
	where the second equality is by \eqref{eq:homegeneous}, and the inequality is due to $\bx \in P$. 
	
	For every $\ell \in I_{\alpha}$, we have
	$$\sum_{t\in T} \bx^+_{\ell,t} = 
	\sum_{t\in T} \bx_{\ell,t} + \sum_{j=1}^{|T|+1}
	\lambda_j \sum_{t\in T} \bm^j_{\ell,t} = 
	\sum_{t\in T} \bx_{\ell,t} =1,$$
	where the second equality is by \eqref{eq:move_select}, and the last equality follows from $\bx\in P$. 
	
	Finally, let $1\leq j \leq n$  and $t\in T$. 
	If $j>|T|+1$ then $\bx^+_{\ell,t} = \bx_{\ell,t}$ for every $\ell\in G_j$, and thus 
	$\sum_{\ell\in G_j } \bx^+_{\ell,t} = \sum_{\ell \in G_j } \bx_{\ell,t}\leq L_{t,j}$, as $\bx\in P$. Otherwise $1\leq j \leq |T|+1$. Observe that $\bx^+_{\ell,t} =\bx_{\ell,t}+\lambda_j \bm^j_{\ell,t}$ for every $\ell\in G_j$, 
	and consider the following cases.
	\begin{itemize}
		\item $\sum_{\ell\in G_j} \bx_{\ell,t} <L_{t,j}$. Using the definitions of $K$ and $M$, we have
		$$\sum_{\ell\in G_j} \bx^+_{\ell,t}= 
		\sum_{\ell\in G_j} \bx_{\ell,t} + \lambda_j\sum_{\ell\in G_j}  \bm^j_{\ell,t}\leq 
		L_{t,j} -K + |I_{\alpha}| \cdot \frac{K}{|I_\alpha|\cdot M  } \cdot M = L_{t,j}.$$
		\item $\sum_{\ell\in G_j} \bx_{\ell,t} =L_{t,j}$. By \eqref{eq:move_group} we have
	$$\sum_{\ell\in G_j} \bx^+_{\ell,t}= 
	\sum_{\ell\in G_j} \bx_{\ell,t} + \lambda_j\sum_{\ell \in G_j}  \bm^j_{\ell,t}=L_{t,j}.$$
	\end{itemize}
	We showed  that $\sum_{\ell \in G_j} \bx_{\ell,t} \leq L_{t,j}$ in all cases. Overall, we have that $\bx^+ \in P$.
	
	We can also define $\bx^- = \bx-\sum_{j=1}^{|T|+1} \lambda_j\cdot \bm^j$. By a symmetric argument we can show that $\bx^-\in P$ as well. It also holds that $\bx^+,\bx^-\neq \bx$ as there is $1\leq j^*\leq |T|+1$ such that $\lambda_{j^*}                                                                                                                                                                                                                                                                                                                                                                                                                                                                                                                                                                                                                                                                                                                                                                                                                                                                                                                                                                                                                                                                                                                                                                                                                                                                                                          \neq 0$; also, there are $\ell^*\in G_{j^*}$  and $t^*\in T$ such that $\bm^{j^*}_{\ell^*,t^*}\neq 0$. Thus, $\bx^+_{\ell^*,j^*}= \bx_{\ell^*,t^*} + \lambda_{j^*} \bm^{j^*}_{\ell^*,j^*}\neq \bx_{\ell^*,t^*}$. 
	Furthermore, $\bx =\frac{1}{2} \cdot \bx^++\frac{1}{2}\cdot \bx^-$, and we conclude that $\bx$ is not a vertex of $P$. Contradiction. \qed
\end{proof}

\begin{lemma}
	\label{lem:FewFractionalItems}
	If Group $G_j$ is fractional then $$\left|\left\{\ell\in G_j~|~\exists t \in T :~ \bx_{\ell,t} \in (0,1)\right\}\right| \leq 2|T|.$$
\end{lemma}

\begin{proof}
	Define $R_j = (G_j \times T)\setminus X_j$. 
	and let  $Q_j\subseteq \mathbb{R}^{R_j}$ be the set (polytope) of  all the vectors $\by \in \mathbb{R}^{R_j}$ which satisfy the following inequalities:
	
	\begin{align}
	\label{eq:Qdef_nonneg}
	&\forall (\ell,t)\in R_j&:~~~& \by_{\ell,t} \geq 0\\
	&\forall \ell \in G_j&:~~~& \sum_{t\in T\textnormal{ s.t. } (\ell,t)\in R_j} \by_{\ell,t} \geq 1 \label{eq:Qdef_cover}\\ 
	&\forall t\in T &:~~~& 
	\sum_{\ell \in G_j\textnormal{ s.t. } (\ell,t)\in R_j} \by_{\ell,t} \cdot s_{\ell} + \sum_{\ell \in I_{\alpha }\setminus G_j} \bx_{\ell,t} \cdot s_{\ell} \leq f(t) |t| \label{eq:Qdef_capacity}\\
&\forall t\in T &:~~~& 
\sum_{{\ell }\in G_j\textnormal{ s.t. } (\ell,t)\in R_j } \by_{\ell,t} \leq L_{t,j} \label{eq:Qdef_cardinality}	
	\end{align}

Also, define $\by^*\in  \mathbb{R}^{R_j}$ by $\by^*_{\ell,t} = \bx_{\ell,t}$ for $(\ell,t)\in R_j$. As $\bx\in P$, 
it holds that $\by^* \in Q_j$ (recall $\bx_{\ell,t}=0$ for all $(\ell,t)\in X_j$). 

Assume towards contradiction that $\by^*$ is not a vertex of $Q_j$. Thus, there is a vector $\bgam \in \mathbb{R}^{R_j}$, $\bgam\neq 0$  such that $\by+\bgam, \by-\bgam \in Q_j$.  Thus, by \eqref{eq:Qdef_cover}, for all $\ell\in G_j$ it holds that, 
$$
	1+ \sum_{t\in T\textnormal{ s.t. } (\ell,t)\in R_j} \bgam_{\ell,t}   = \sum_{t\in T\textnormal{ s.t. } (\ell,t)\in R_j} \left(\by^*_{\ell,t} + \bgam_{\ell,t} \right) \geq 1,
$$
and 
$$
1- \sum_{t\in T\textnormal{ s.t. } (\ell,t)\in R_j} \bgam_{\ell,t}   = \sum_{t\in T\textnormal{ s.t. } (\ell,t)\in R_j} \left(\by^*_{\ell,t} - \bgam_{\ell,t} \right) \geq 1.
$$
Hence,
\begin{equation}\sum_{t\in T\textnormal{ s.t. } (\ell,t)\in R_j} \bgam_{\ell,t}=0.
	\label{eq:gamsum}
	\end{equation}

Define $\bx^+,\bx^- \in \mathbb{R}^{I_{\alpha}\times T}$ by $\bx^+_{\ell,t} = \bx^-_{\ell,t} =\bx_{\ell,t}$ for $(\ell,t)\notin R_j$, and $\bx^+_{\ell,t}= \bx_{\ell,t}+\bgam_{\ell,t},~\bx^-_{\ell,t} = \bx_{\ell,t}-\bgam_{\ell,t}$ for $(\ell,t)\in R_j$.  Since $\bgam\neq 0$ it follows that $\bx^+ \neq \bx^-$. 

Let $(\ell,t)\in I_{\alpha}\times T$. If $(\ell,t)\in R_j$ then $\bx^+_{\ell,t} =\by^*_{\ell,t} +\bgam_{\ell,t}\geq 0$, since $\by^*+\bgam_{\ell,t} \in Q_j$  and due to \eqref{eq:Qdef_nonneg}. If $(\ell,t)\notin R_j$ then $\bx^+_{\ell,t}= \bx_{\ell,t} \geq 0$. Overall, we have that $\bx^+_{\ell,t}\geq 0$ in both cases. 
Furthermore, if $s_{\ell}>\eps f(t)$ then $\bx^+_{\ell,t} = \bx_{\ell,t} = 0$ as $(\ell,t)\notin R_j$. 

For any $t\in T$, it holds that 
$$\sum_{\ell\in I_{\alpha}} \bx^+_{\ell,t} s_{\ell}
= \sum_{\ell\in  I_{\alpha}\setminus G_j} \bx_{\ell,t} s_{\ell}  +
\sum_{\ell \in G_j\textnormal{ s.t. } (\ell,t)\in R_j} \left(\by^*_{\ell,t} +\bgam_{\ell,t}\right)s_{\ell} \leq f(t)|t|,$$
where the last equality is due to \eqref{eq:Qdef_capacity} and $\by^*+\bgam \in Q_j$. 

Finally, for any $\ell \in I_{\alpha}$, if $\ell \notin G_j$ then
$$\sum_{t\in T} \bx^+_{\ell,t} = \sum_{t\in T} \bx_{\ell,t} =1.$$ 
Also, if  $\ell \in G_j$ then using \eqref{eq:gamsum}, we have
$$\sum_{t\in T} \bx^+_{\ell,t} =
\sum_{t\in T} \bx_{\ell,t}
 \sum_{t\in T~\textnormal{ s.t. } (\ell,t)\in R_j } \bgam_{\ell,t}  =1.$$ 
 
 Overall, we have that $\bx^+\in P$. Using a symmetric argument it follows that $\bx^-\in P$ as well.  Since $\bx= \frac{\bx^+ \bx^-}{2}$, it follows that $\bx$ is not a vertex of $P$. A contradiction. Thus, $\by^*$ is a vertex of $Q_j$. 
 
 Let $F= \{\ell\in G_j~|~\exists t\in T:~\by^*_{\ell,t}\in (0,1)\}$. By the definition of $\by^*$, it suffices to show that $|F|\leq 2|T|$. We say that an inequality from \eqref{eq:Qdef_nonneg}, \eqref{eq:Qdef_cover}, \eqref{eq:Qdef_capacity} and 
\eqref{eq:Qdef_cardinality} is {\em tight} if it holds with equality with respect to $\by^*$. 
 For any $\ell\in G_j$ define $\xi_{\ell} = |\{t\in T~|~(\ell,t)\in R_j\}|$. It follows that $|R_j|= \sum_{\ell\in G_j} \xi_{\ell}$. For $\ell\in F$ up to $\xi_{\ell}-1$ of the inequalities in \eqref{eq:Qdef_nonneg} and \eqref{eq:Qdef_cover} are tight. Also, for $\ell \in G_j\setminus F$ up to $\xi_{\ell}$ of inequalities in \eqref{eq:Qdef_nonneg} and \eqref{eq:Qdef_cover} are tight. As the number of inequalities in \eqref{eq:Qdef_cardinality} and \eqref{eq:Qdef_capacity} is $2|T|$, it follows that the number of tight equalities is at most
 $$
\sum_{\ell\in F} (\xi_{\ell} -1)+ \sum_{\ell\in G_j\setminus F} \xi_\ell +2|T| \leq |R_j| -|F| +2|T|.
$$ 
As $\by^*$ is a vertex, there are at least $|R_j|$ tight inequalities. Thus, $|F|\leq 2|T|$ as required.
 \qed
\end{proof}

Since $|T| = O(1)$, Theorem~\ref{thm:partitionPhase} follows from Lemmas~\ref{lem:FewFractionalGroups} and~\ref{lem:FewFractionalItems}.




\end{document}